\documentclass[10pt]{article}      
\usepackage{fullpage}

\usepackage{amsmath,amssymb,amsfonts,amsthm,amscd} 
\usepackage{graphics}                 
\usepackage{color}                    
\usepackage{mdframed}                 
\usepackage{mathrsfs}
\usepackage{dsfont}
\usepackage{fullpage}
\usepackage{bbold}
\usepackage{url}         
\usepackage{colonequals} 
\usepackage[toc,page]{appendix} 

\usepackage{soul}
\setstcolor{red}

\usepackage{longtable} 

\usepackage{enumitem}    

\newcommand\Item[1][]{%
  \ifx\relax#1\relax  \item \else \item[#1] \fi
  \abovedisplayskip=0pt\abovedisplayshortskip=0pt~\vspace*{-\baselineskip}}

\usepackage{nccmath}   
\usepackage[hidelinks]{hyperref}

\usepackage{todonotes}

\usepackage{tikz}
\usetikzlibrary{arrows,decorations.markings,patterns}

\theoremstyle{plain}
\newtheorem{theorem}{Theorem}[section]
\newtheorem{thm}[theorem]{Theorem}
\newtheorem{formal}[theorem]{Formal Theorem}
\newtheorem{lemma}[theorem]{Lemma}
\newtheorem{lem}[theorem]{Lemma}
\newtheorem{prop}[theorem]{Proposition}
\newtheorem{conjecture}[theorem]{Conjecture}
\newtheorem{corollary}[theorem]{Corollary}
\newtheorem{cor}[theorem]{Corollary}

\theoremstyle{definition}

\newtheorem{definition}[theorem]{Definition}

\theoremstyle{remark}
\AtEndEnvironment{remark}{\qed}

\newcommand{\iparagraph}[1]{\paragraph{\emph{#1}}}

\definecolor{aliceblue}{rgb}{0.94, 0.97, 1.0} 
\definecolor{azuremist}{rgb}{0.94, 1.0, 1.0} 
\definecolor{antiquewhite}{rgb}{0.98, 0.92, 0.84} 
\definecolor{ivory}{rgb}{1.0, 1.0, 0.94}
\newtheorem{mdexample}[theorem]{Example (IPFG).}
\newenvironment{example}%
  {\begin{mdframed}[backgroundcolor=ivory]\begin{mdexample}}%
  {\end{mdexample}\end{mdframed}}
\newtheorem{mdoexample}[theorem]{Example.}
  {\begin{mdframed}[backgroundcolor=ivory]\begin{mdoexample}}%
  {\end{mdoexample}\end{mdframed}}

\newtheorem{remark}[theorem]{Remark}

\numberwithin{equation}{section}
\newtheorem*{theorem*}{Theorem}

\newcommand{\bw}{\mathrm{bw}}

\newcommand{\sym}{\mathrm{sym}}
\newcommand{\asym}{\mathrm{asym}}
\newcommand{\symdiss}{\mathrm{symdiss}}

\newcommand\norm[1]{\left\lVert#1\right\rVert}

\DeclareMathOperator{\Grad}{grad}
\DeclareMathOperator{\Law}{law}
\DeclareMathOperator{\Dom}{Dom}
\DeclareMathOperator{\Ker}{Ker}
\DeclareMathOperator{\Ran}{Ran}

\DeclareMathOperator{\J}{J}

\DeclareMathOperator{\Prob}{\mathrm{Prob}}
\DeclareMathOperator*{\sumsum}{\sum\!\sum}
\newcommand{\super}[1]{^{\scriptscriptstyle{(#1)}}}

\def\vep{\varepsilon}

\def\tp{^\mathsf{T}}

\def\RelEnt{\mathrm{RelEnt}}

\def\RF{\mathcal{R}}
\def\J{\mathcal{J}}
\def\I{\mathcal{I}}
\def\CoF{G}
\def\QP{\mathcal{V}}
\def\Pot{U}

\DeclareMathOperator{\Wasser}{W}

\def\CC{\mathbb{C}}
\def\E{\mathcal{E}}
\def\H{\mathcal{H}}
\def\L{\mathcal{L}}
\def\M{\mathcal{M}}
\def\P{\mathcal{P}}
\def\S{\mathcal{S}}
\def\T{\mathbb{T}}

\def\V{\mathcal{V}}
\def\W{\mathcal{W}}
\def\X{\mathcal{X}}

\def\Z{\mathcal{Z}}
\def\PP{\mathbb{P}}
\def\Q{\mathcal{Q}}
\def\RR{\mathbb{R}}
\def\TT{\mathbb{T}}
\def\NN{\mathbb{N}}
\def\ZZ{\mathbb{Z}}

\def\fluxS{\mathcal{X}^2/2}

\def\Fsy{{F^{\mathrm{sym}}}}
\def\Fasy{{F^{\mathrm{asym}}}}

\newcommand{\Reac}{\mathrm{R}}
\newcommand{\Rf}{{\mathrm{R}_{\mathrm{fw}}}}
\newcommand{\Rb}{{\mathrm{R}_{\bw}}}
\renewcommand{\k}{\kappa}
\newcommand{\kr}{{\kappa_{r}}}
\newcommand{\kbr}{{\kappa_{\bw(r)}}}
\newcommand{\cf}{{c_{r}}}
\newcommand{\cb}{{c_{\bw(r)}}}

\def\div{\mathop{\mathrm{div}}\nolimits}
\newcommand{\grad}{\nabla}
\def\Ddiv{\overline{\div}}
\def\Dgrad{\overline{\nabla}}
\def\Dnabla{\Dgrad}
\DeclareMathOperator{\Cdiv}{\div}
\def\Cnabla{{\nabla}}

\def\pJ{\mathbb{J}}
\def\pG{\mathcal{G}}
\def\pF{\mathcal{F}}

\usepackage{tocloft}
\makeatletter
  \newcommand{\eqnum}{\leavevmode\hfill\refstepcounter{equation}\textup{\tagform@{\theequation}}} 
\makeatother

\date \today
\title{Variational structures beyond gradient flows: \\a macroscopic fluctuation-theory perspective}
\author{Robert I.A.\ Patterson\thanks{WIAS, Mohrenstrasse 39, 10117 Berlin, Germany. Email: \href{mailto:patterson@wias-berlin.de}{patterson@wias-berlin.de}}, D.R.\ Michiel Renger\thanks{TU M\"unchen, Boltzmannstra{\ss}e 3, 85747 Garching, Germany. Email: \href{mailto:d.r.m.renger@tum.de}{d.r.m.renger@tum.de}} and Upanshu Sharma\thanks{
University of New South Wales, Sydney 2052, Australia.\ Email: \href{mailto:upanshu.sharma@unsw.edu.au}{upanshu.sharma@unsw.edu.au}}}

\begin{document}

\setstcolor{purple}

\maketitle

\begin{abstract}

Macroscopic equations arising out of stochastic particle systems in detailed balance (called dissipative systems or gradient flows) have a natural variational structure, which can be derived from the large-deviation rate functional for the density of the particle system. While large deviations can be studied in considerable generality, these variational structures are often restricted to systems in detailed balance. Using insights from macroscopic fluctuation theory, in this work we aim to generalise this variational connection beyond dissipative systems by augmenting densities with fluxes, which encode  non-dissipative effects. Our main contribution is an abstract theory, which for a given flux-density cost and a quasipotential, provides a decomposition into dissipative and non-dissipative components and a generalised orthogonality relation between them. We then apply this abstract theory to various stochastic particle systems  --  independent copies of jump processes, zero-range processes, chemical-reaction networks in complex balance and lattice-gas models -- without assuming detailed balance. For macroscopic equations arising out of these particle systems,  we derive new variational formulations that generalise the classical gradient-flow formulation. 
\end{abstract}

\section{Introduction}
\label{sec:intro}
When studying an evolution equation, it is often helpful to know if it has an associated variational structure, in order to obtain physical insight and tools for mathematical analysis. An important example of such a structure is a gradient flow or dissipative system;  in this case the structure consists of an energy functional and a dissipation mechanism, and the evolution equation is completely characterised by a corresponding minimisation problem involving these two objects. From a thermodynamic point of view, such a variational structure is often related to random fluctuations of an underlying microscopic particle system via a large-deviation principle --- examples include the Boltzmann--Gibbs--Helmholtz free energy and the Onsager--Machlup theory. 

It has recently become clear that macroscopic equations are always dissipative (called gradient flows) if the underlying microscopic stochastic system is in detailed balance\footnote{In this paper we often use the terminology of dissipative systems interchangeably with gradient flows since in non-equilibrium systems the gradient-flow part arises purely due to dissipative effects characterised by the symmetric forces discussed below.}. The energy functional and the dissipation mechanism for such macroscopic equations are then uniquely derived by an appropriate decomposition of the large-deviation rate functional associated to the microscopic systems~\cite{AdamsDirrPeletierZimmer11,AdamsDirrPeletierZimmer13,MielkePeletierRenger14,PeletierRedigVafayi14}. These observations have provided a canonical approach to constructing a variational structure for such macroscopic equations. In addition to having a clear physical interpretation, these variational structures have been used to isolate interesting features of the macroscopic equations and study singular-limit problems arising therein.

So far, this approach has largely been  limited to particle systems in detailed balance and corresponding macroscopic dissipative systems. Since a large deviation study is possible far beyond detailed balance, this leads to the following natural question.
\begin{center}
\emph{Do the large deviations of the underlying particle systems provide a variational structure beyond detailed balance?}
\end{center}
While this is a hard question to answer in general, considerable progress has been made in the case of some specific systems in two seemingly independent directions.

One direction that is tailored to allow for non-dissipative effects is the study of so-called \emph{FIR inequalities}, first introduced for the many-particle limit of Vlasov-type nonlinear diffusions~\cite{DuongLamaczPeletierSharma17}, independent particles on a graph~\cite{HilderPeletierSharmaTse20} and chemical reactions~\cite[Sec.~5]{RengerZimmer2021}. These inequalities bound the free-energy difference and Fisher information by the large-deviation rate functional, providing a useful tool to study singular-limit problems and to derive error estimates~\cite{DLPSS17,PeletierRenger2020TR}. Strictly speaking, these inequalities are not variational structures in the sense that they do not fully determine the macroscopic dynamics. However, in this paper we will construct a variational structure which generalises these inequalities and completely characterises the macroscopic dynamics. 

Another direction of generalising dissipative systems is by using Macroscopic Fluctuation Theory (MFT) \cite{BDSGJLL2015MFT}. The main idea here is to consider, in addition to the usual density of the particle system, the particle fluxes at the microscopic level, and to study the large deviations of these fluxes.
Consequently using time-reversal arguments, MFT explicitly captures the dissipative and non-dissipative effects in the system. However, most MFT literature has been devoted to diffusive scaling of particle systems and corresponding quadratic rate functions. Such rate functions define a Hilbert space with a natural orthogonal decomposition into dissipative and non-dissipative components. Recently non-quadratic rate functions and connections to MFT have been explored in the case of independent particles on a graph~\cite{KaiserJackZimmer2018} and chemical reaction networks~\cite{RengerZimmer2021}, but a general MFT for non-quadratic rate functions is largely open.

Spurred on by these exciting new developments, we provide a partial but affirmative answer to the question posed above. The basis of our analysis is an abstract action functional $(\rho,j)\mapsto\int_0^T\!\L(\rho(t),j(t))\,dt$. This functional will correspond to the large deviations of random particle systems, but this identification is not necessary for our analysis; in this sense our approach is purely macroscopic. Inspired by FIR-inequalities and MFT, we set up an abstract theory whose central outcome will be a series of decompositions of the integrand $\L$ into distinct dissipative and non-dissipative components. These decompositions generalise:
(1) the connection between large deviations and dissipative systems from~\cite{MielkePeletierRenger14} to include non-dissipative effects,
(2) the known cases of FIR inequalities~\cite{HilderPeletierSharmaTse20} to a general setting, and
(3) MFT to non-quadratic action functions.

Finally we apply this abstract theory to the density-flux large-deviation rate functional for various stochastic particle systems without assuming detailed balance, and derive new variational formulations for the corresponding macroscopic equations.

\subsection{Summary of results}
\label{subsec:results}

\iparagraph{Abstract results.}
Consider the macroscopic densities and fluxes $[0,T]\ni t \mapsto (\rho(t),j(t))$ that are evolving according to a coupled system of evolution equations:
\begin{subequations}\label{eq:coupled ev eq}
\begin{align}
    \dot\rho(t) &= -\div j(t),\label{eq:intro cont eq}\\
     j(t) &= j^0(\rho(t)),       \label{eq:evolution flux}
\end{align}
\end{subequations}%
with an associated action functional
\begin{equation}\label{intro:L-dens-flux}
  (\rho,j)\mapsto\int_0^T\!\L(\rho(t),j(t))\,dt,
\end{equation}
where the non-negative cost function $\L$ has the crucial property that for any $(\rho,j)$,
\begin{equation*}
  j=j^0(\rho) \ \Longleftrightarrow \ \L\big(\rho,j\big) = 0,
\end{equation*}
and hence the action~\eqref{intro:L-dens-flux} is minimised by the trajectory~\eqref{eq:evolution flux}.
We will interpret equation~\eqref{eq:intro cont eq} as a continuity equation and call $j^0(\rho)$ the \emph{zero-cost flux} associated to $\L$. Equation~\eqref{eq:coupled ev eq} often describes the macroscopic dynamics arising from a microscopic stochastic particle system and \eqref{intro:L-dens-flux} is typically the corresponding large-deviation rate functional. 

Although writing the flux explicitly in~\eqref{eq:evolution flux} instead of directly studying $\dot\rho(t)=-\div j^0(\rho(t))$ might seem superfluous at first sight, it is motivated by the fact that fluxes can encode information on non-dissipative, for instance divergence-free, effects in the system. Consequently, while studying densities is usually sufficient for dissipative systems~\cite{Onsager1931I, Onsager1931II, Onsager1953I, MielkePeletierRenger14,MPPR2017} (see Section~\ref{subsec:intro-Connection} below for more details), the inclusion of fluxes is better suited to describe non-dissipative effects at the macroscopic level~\cite{BDSGJLL2015MFT,Maes2018}. 

Our abstract theory requires the existence of three objects: a sufficiently regular density-flux cost function $\L(\rho,j)$, an operator that will play the role of divergence and as such defines the continuity equation~\eqref{eq:intro cont eq} and a non-negative \emph{quasipotential} $\QP$ associated to $\L$. The basis of our approach will be the decomposition $\L(\rho,j)=\Phi(\rho,j)+\Phi^*(\rho,F(\rho))-\langle F(\rho),j\rangle$, where $F(\rho):=-d\L(\rho,0)$ is called the \emph{driving force} and $\Phi$ and its convex dual $\Phi^*$ the \emph{dissipation potentials}, see Theorem~\ref{th:MPR force} for details.  This decomposition is standard in the literature \cite{MielkePeletierRenger14,KaiserJackZimmer2018,Maes2018} and corresponds to a (possibly nonlinear) force-flux response relation $j=d\Phi^*(\rho,F(\rho))$ for the zero-cost dynamics; it includes gradient flows as a special case as discussed in Section~\ref{subsubsec:intro-diss}.

Borrowing ideas from MFT, we uniquely decompose this driving force into a \emph{symmetric} and \emph{antisymmetric} part
\begin{equation*}
  F(\rho)=\Fsy(\rho)+\Fasy(\rho).
\end{equation*}
On a macroscopic level, these notions of (anti)symmetry (defined in Section~\ref{subsec:abstract time reversal}) are consistent with the time-reversal symmetry of Markov processes in the context of MFT and large deviations. 
In particular, if the microscopic system is in detailed balance, then $F(\rho)=\Fsy(\rho)$ and the (macroscopic) dynamics is purely dissipative, i.e.\ described by a gradient flow driven by a \emph{quasipotential} $\QP$ \cite{MielkePeletierRenger14}. 
It turns out that even for systems that are not in detailed balance, the symmetric force $F^\sym$ always relates to such a $\V$, which can be defined in terms of the cost $\L$ (see Definition~\ref{def:I0}) and is a natural Lyapunov functional for the system.
In particular, the symmetric part $\Fsy(\rho)$ is a conservative force driven by the quasipotential (energy) $\V$.  

More generally, from a physical point of view, a purely dissipative system is thermodynamically closed, so that the work done is related to the free energy or quasipotential via 
\begin{equation}
  \int_0^T\!\big\langle \Fsy(\rho(t)),j(t)\big\rangle\,dt = -\mfrac12\QP(\rho(T))+\mfrac12\QP(\rho(0)),  
\label{eq::intro QP decay}
\end{equation}
or formulated locally in time for the power
\begin{equation}\label{eq:intro sym power}
  \big\langle \Fsy(\rho(t)),j(t)\big\rangle = -\mfrac12\mfrac{d}{dt} \QP(\rho(t)).
\end{equation}
Thus for non-closed systems one can think of $\Fsy(\rho)$ as an internally generated force and the remainder, $\Fasy(\rho)$, as the force exerted by the system upon the environment.
While
\begin{equation}  \label{eq:intro asym power}
  \big\langle\Fasy(\rho(t)),j(t)\big\rangle
 \quad \text{and} \quad
  \big\langle F(\rho(t)),j(t)\big\rangle
\end{equation}
can be understood as expressions of power or rates of work, in general there is no reason to expect these to be exact differentials.

In our main result, Theorem~\ref{th:FIIRs}, we relate the cost function $\L$ to the three powers from  \eqref{eq:intro sym power} and \eqref{eq:intro asym power}. Specifically, for any $\lambda\in[0,1]$, the cost function $\L$ admits the following decompositions 
\begin{subequations}\label{eq:intro-L-decom}
\begin{align}
  \L(\rho,j)&=\L_{(1-2\lambda)F}(\rho,j) + \RF^\lambda_F(\rho) - 2\lambda \langle F(\rho),j\rangle,  &&\text{with } \  \RF^\lambda_F(\rho)\geq0,  \label{eq:Intro-FIIR-F}\\
    \L(\rho,j)&=\L_{F-2\lambda F^\sym}(\rho,j) + \RF^\lambda_{\Fsy}(\rho) - 2\lambda \langle F^\sym(\rho),j\rangle,  
            &&\text{with }   \ \RF^\lambda_{\Fsy}(\rho)\geq0,\label{eq:Intro-FIIR-sym} \\
  \L(\rho,j)&=\L_{F-2\lambda F^\asym}(\rho,j) + \RF^\lambda_{\Fasy}(\rho) - 2\lambda \langle F^\asym(\rho),j\rangle,
        &&\text{with }  \  \RF^\lambda_{\Fasy}(\rho)\geq0. \label{eq:Intro-FIIR-asym}
\end{align}
\end{subequations}
The parameter $\lambda$ can be used to switch between different forces and the non-negative terms $\L_\CoF(\rho,j)$ are modified versions of $\L$ where the driving force $F(\rho)$ is replaced by a different covector field $\CoF(\rho)$. Consequently, the zero-cost flux of $\L_\CoF$ will be a modified dynamics, different from \eqref{eq:evolution flux}. Of particular interest is the case $\lambda=\frac12$, where the decompositions \eqref{eq:Intro-FIIR-sym} and \eqref{eq:Intro-FIIR-asym} can be seen as two different ways to split $\L$ into \emph{purely dissipative} and \emph{purely non-dissipative} components. Indeed, the modified cost $\L_{\Fsy}$ is related to a purely dissipative system that can be formalised as a gradient flow (see Section~\ref{subsubsec:intro-diss}). By contrast, we interpret the zero-cost flux of $\L_{\Fasy}$ as purely non-dissipative. Although the variational structure and physical interpretation of $\L_{\Fasy}$ remains an open question (see discussion in Section~\ref{sec:discuss}), we show for certain examples that its zero-cost behaviour corresponds to a purely Hamiltonian macroscopic evolution. This idea is clearly illustrated by Figure~\ref{fig:IPFG solutions}, where we plot the phase diagram for the zero-cost flux associated with $\L_F$, $\L_{\Fsy}$ and $\L_{\Fasy}$  in the case of independent Markov jump particles on a three-point state space. For details on this example see Sections~\ref{subsec:flows} and~\ref{sec:IPFG antisym flow}.

\begin{figure}[ht]
\centering
\tikzset{->-/.style={decoration={
  markings,
  mark=at position #1 with {\arrow{>}}},postaction={decorate}}}

\begin{tikzpicture}[scale=4,baseline=-1em]
\tikzstyle{every node}=[font=\scriptsize];

  \draw[fill=lightgray](0.5,0.866)--(0,0)--(1,0)--cycle;
  \draw[->](1,0)--(1.2,0) node[anchor=south]{$\rho_1$};
  \draw[->](0.5,0.87)--(0.6,1.04) node[anchor=east]{$\rho_2$};

  \filldraw (0.5,0.29) circle (0.01) node[anchor=south west]{$\pi$};

  \draw[->-=0.5](0.33,0) .. controls (0.53,0.23) and (0.53,0.29) .. (0.5,0.29);
  \draw[->-=0.5](0.67,0) .. controls (0.57,0.29) and (0.52,0.32) .. (0.5,0.29);
  \draw[->-=0.5](0.83,0.29) .. controls (0.53,0.35) and (0.48,0.32) .. (0.5,0.29);
  \draw[->-=0.5](0.67,0.58) .. controls (0.47,0.35) and (0.47,0.29) .. (0.5,0.29);
  \draw[->-=0.5](0.33,0.58) .. controls (0.43,0.29) and (0.49,0.26) .. (0.5,0.29);
  \draw[->-=0.5](0.17,0.29) .. controls (0.47,0.23) and (0.52,0.26) .. (0.5,0.29);
  
  \node at (0.5,-0.02) [anchor=north]{\normalsize{(a)}};
\end{tikzpicture}
\begin{tikzpicture}[scale=4,baseline=-1em]
\tikzstyle{every node}=[font=\scriptsize];

  \draw[fill=lightgray](0.5,0.866)--(0,0)--(1,0)--cycle;
  \draw[->](1,0)--(1.2,0) node[anchor=south]{$\rho_1$};
  \draw[->](0.5,0.87)--(0.6,1.04) node[anchor=east]{$\rho_2$};

  \filldraw (0.5,0.29) circle (0.01) node[anchor=south west]{$\pi$};

  \draw[->-=0.5](0.33,0) -- (0.5,0.29);
  \draw[->-=0.5](0.67,0) -- (0.5,0.29);
  \draw[->-=0.5](0.83,0.29) -- (0.5,0.29);
  \draw[->-=0.5](0.67,0.58) -- (0.5,0.29);
  \draw[->-=0.5](0.33,0.58) -- (0.5,0.29);
  \draw[->-=0.5](0.17,0.29) -- (0.5,0.29);

  \node at (0.5,-0.02) [anchor=north]{\normalsize{(b)}};
\end{tikzpicture}
\begin{tikzpicture}[scale=4,baseline=-1em]
\tikzstyle{every node}=[font=\scriptsize];

  \draw[fill=lightgray](0.5,0.866)--(0,0)--(1,0)--cycle;
  \draw[->](1,0)--(1.2,0) node[anchor=south]{$\rho_1$};
  \draw[->](0.5,0.87)--(0.6,1.04) node[anchor=east]{$\rho_2$};

  \filldraw (0.5,0.29) circle (0.01) node[anchor=south west]{$\pi$};

  \draw[smooth cycle,->-=0.18,->-=0.51,->-=0.87] plot[tension=0.9] coordinates{(0.5,0.481) (0.33,0.192) (0.67,0.192)};
  \draw[smooth cycle,->-=0.175,->-=0.505,->-=0.865] plot[tension=0.9] coordinates{(0.5,0.577) (0.25,0.144) (0.75,0.144)};
  
  \draw[rounded corners=10,->-=0.22,->-=0.58,->-=0.895] (0.65,0).. controls (0.86,0.05) and (0.89,0.1) .. (0.825,0.303)--(0.675,0.563) ..controls (0.53,0.72) and (0.47,0.72).. (0.325,0.563) -- (0.175,0.303)  .. controls (0.11,0.1) and (0.14,0.05) .. (0.35,0)--cycle; 

  \draw[rounded corners=4] (0.8,0) .. controls (0.95,0.02) and (0.96,0.03) .. (0.9,0.17)--(0.6,0.69) .. controls (0.51,0.82) and (0.49,0.82) .. (0.4,0.69)--(0.1,0.17) .. controls (0.04,0.03) and (0.05,0.02) .. (0.2,0)--cycle;  

  \node at (0.5,-0.02) [anchor=north]{\normalsize{(c)}};
\end{tikzpicture}
\caption{Consider the setting of independent and irreducible Markov jump particles on a three-point state space with generator $Q:=\lbrack\lbrack-3,2,1\rbrack,\lbrack 1,-3,2\rbrack,\lbrack 2,1,-3)\rbrack\rbrack$ and invariant measure $\pi=(\frac13,\frac13,\frac13)$. 
Phase diagram for the (zero-cost) trajectories $\rho(t)$ associated to (a) $\L(\rho(t),j(t))=0$; (b) $\L_\Fsy(\rho(t),j(t))=0$; (c) $\L_\Fasy(\rho(t),j(t))=0$. 
Here $\rho_i$ is the mass at point $i$ and we do not plot $\rho_3$ since $\sum_i\rho_i=1$. The zero-cost trajectories for $\L_\Fsy$ and  $\L_\Fasy$ follow a purely dissipative and Hamiltonian dynamics respectively. 
}
\label{fig:IPFG solutions}
\end{figure}

The middle terms in the right hand side of~\eqref{eq:intro-L-decom} are inspired by~\cite[Def.\ 1.5]{HilderPeletierSharmaTse20},~\cite[Sec.~5]{RengerZimmer2021}, and are called \emph{generalised Fisher informations}.
For $\lambda\in [0,1]$ and covector fields $\CoF=F,\Fsy,\Fasy$, they are defined as
\begin{equation}\label{def:intro-gFI}
  \RF^\lambda_\CoF(\rho):=-\H\big(\rho,-2\lambda \CoF(\rho)\big),
\end{equation} 
where $\H$ is the convex dual of $\L$.
The terminology is motivated by the fact that (see Proposition~\ref{prop:preFIR})
\begin{equation*}
  \lim\limits_{\lambda\rightarrow 0} \mfrac1\lambda\RF^\lambda_{\CoF}(\rho) = \langle \CoF(\rho), j^0(\rho) \rangle,
\end{equation*}
which in the case $\CoF=\Fsy$ is the time derivative or dissipation rate of the quasipotential along the zero-cost path, i.e.\ in the limit $\lambda\rightarrow 0$, $\RF^\lambda_\Fsy$ coincides with the classical Fisher information \cite{HilderPeletierSharmaTse20}. The non-negativity of the generalised Fisher informations in \eqref{eq:intro-L-decom} is essential, since it shows that the three powers in~\eqref{eq:intro sym power} and~\eqref{eq:intro asym power} are non-negative  along the zero-cost flux, thus generalising the second law of thermodynamics.

\iparagraph{Scope.} 
To highlight the minimal underlying structure required to obtain the decompositions~\eqref{eq:intro-L-decom}, analysis will be carried out in a general abstract setting. This implies that our results can be applied to a broad range of models, not restricted to large deviations or to continuity equations of divergence-type. In theory, after properly setting up the spaces, the only requirements of analysis will be the cost function $\L$ together with a continuity equation of the form~\eqref{eq:intro cont eq}. However for specific applications, explicit calculations are restricted to cost functions $\L$ for which the associated quasipotential $\V$ is known. For the purpose of this paper, we define the quasipotential in terms of a Hamilton-Jacobi-Bellman equation (Definition~\ref{def:I0}), and solve it for a number of examples. For cost functions that are derived from large deviations, this definition coincides with the large-deviation rate functional of the invariant measure (see Theorem~\ref{th:LDP QP}). However we reiterate that the abstract definition is purely macroscopic and does not require connections to large deviations.

\iparagraph{Application.}

All three decompositions~\eqref{eq:intro-L-decom} are power balances, split into purely dissipative and purely non-dissipative powers in a physically consistent way. 
From a mathematical perspective, this generalises ideas from dissipative systems to a larger class of systems which include non-dissipative effects. For dissipative systems ($\Fasy(\rho)\equiv0$) these decompositions coincide with the variational formulation of a gradient flow  (see Section~\ref{subsubsec:intro-diss}).  However, our abstract theory only requires a suitably convex cost $\L$ and quasipotential $\V$ for the decompositions (and therefore the corresponding variational ideas) to hold. Lyapunov functions, Fisher informations and dissipation potentials are central ingredients in gradient-flow theory and often difficult to discern in non-dissipative systems (for instance the laws of non-reversible Markov processes). This work provides \emph{explicit} formulae for these objects in terms of the cost and the quasipotential.

For the zero-cost dynamics~\eqref{eq:coupled ev eq}, our results imply that the three powers $\langle F,j\rangle, \langle F^\asym,j\rangle$ and $\langle F^\asym,j\rangle$ are always non-positive, and in particular that $\V$ is a Lyapunov functional with an \emph{explicit} expression for its decay (rather than merely an upper bound).

By contrast, the decay~\eqref{eq::intro QP decay} of the quasipotential $\V$ is \emph{bounded} by a FIR inequality, which connect the cost to the quasipotential and Fisher information. These inequalities are crucial in studying singular limits in non-dissipative systems, for instance to prove compactness of densities and fluxes in suitable topologies. However they are only available in a limited setting. It turns out that since  the modified cost functions $\L_\CoF$ in~\eqref{eq:intro-L-decom} are non-negative, the FIR inequalities naturally arise from these decompositions and therefore we provide a \emph{universal} recipe to arrive at such inequalities. In fact, the decompositions~\eqref{eq:intro-L-decom} explicitly characterise the gap in the FIR inequalities. For more details see Section~\ref{subsubsec:FIR-intro}.

The aforementioned gap in the inequalities corresponds to the $\L_{\CoF}$ on the right-hand side of~\eqref{eq:intro-L-decom}. This new term exactly characterises the effects of non-dissipative effects in the variational structure and the corresponding macroscopic evolution. This is especially revealing for jump processes where we find that purely non-dissipative systems ($\Fsy(\rho)\equiv0$) correspond to Hamiltonian-type structures.

From a physical standpoint, the decompositions~\eqref{eq:intro-L-decom} can be  interpreted as a novel combination of gradient flows and Hamiltonian systems, in a similar spirit to GENERIC (see Section~\ref{subsubsec:GENERIC}). However, we stress that all of our examples -- apart from the lattice gas model -- \emph{cannot} be cast into the GENERIC framework. This work also provides a framework to study physically relevant `open-boundary' jump-process systems (see a recent application in~\cite{RengerSharma22}).

Finally these decompositions also have numerical implications since numerical schemes inspired by gradient-flow structures of evolution equations have gained importance~\cite{CarrilloChertockHuang15} in recent years.
Numerical schemes often add artificial non-reversibility to speed-up convergence to equilibrium, but their analysis is tricky except in special situations~\cite{LelievreNierPavliotis13}. The decompositions~\eqref{eq:intro-L-decom} explicitly characterise the role of Fisher informations and antisymmetric forces and a natural goal would be to optimise this force to speed up convergence.

\iparagraph{Examples.} 
Above we discussed the abstract framework and theory derived from it; this theory is purely macroscopic in that we do not require any connection to particle systems and large deviations. In the latter part of this paper we apply this abstract theory to several microscopic particle systems. 

First, we focus on independent Markov jump particles on a finite graph as a guiding example throughout this paper, and generalise the results of~\cite{KaiserJackZimmer2018}. Second, we study zero-range processes in a scaling which leads to an ordinary differential equation (ODE) in the limit. Third, we study chemical reaction networks in complex balance~\cite{AndersonKurtz11} and generalise the results in~\cite{RengerZimmer2021}. In all these three examples the macroscopic dynamics are ODEs and the large-deviation principle yields an exponential rate functional.   
Finally, we focus on the setting of particles that hop on a lattice in a diffusive limit, which leads to a drift-diffusion equation as the macroscopic evolution. These particles can either be independent random walkers or interact via exclusion. In this setting, the large-deviation principle yields a quadratic rate functional, and we recover the classical MFT results~\cite{BDSGJLL2015MFT}.


\iparagraph{Boundary issues and global-in-time decompositions.}
The decompositions~\eqref{eq:intro-L-decom} do not involve time, and therefore when considering trajectories $t\mapsto (\rho(t),j(t))$, they should be considered as local-in-time or instantaneous  decompositions of $\L(\rho(t),j(t))$ at time $t$. Naively, one would simply integrate in time to obtain global decompositions of the rate functional $\int_0^T\!\L(\rho(t),j(t))\,dt$ for arbitrary trajectories $(\rho,j)$. This argument is formal since, strictly speaking, the decompositions \eqref{eq:intro-L-decom}  hold only for $\rho$, $j$ for which the required terms are defined. More precisely, it turns out that the forces  $F$, $\Fsy$ and $\Fasy$ are well-defined only on a proper subset of the domain of definition for the modified cost functions $\L_\CoF$ and generalised Fisher informations $\RF^\lambda_\CoF$.
This issue is often ignored in the MFT literature.

This issue becomes clear in the various examples we consider. For instance when dealing with independent jump processes on a finite lattice $\X$, the large-deviation cost is well defined for any trajectory in the space of probability measures i.e.\ $\rho(t)\in\P(\X)$ (see Example~\ref{ex:indep LDP}), whereas the symmetric force is only well-defined for trajectories in the space of strictly positive probability measures, i.e.\ $\rho(t)\in \P_+(\X)$ (see~\eqref{IPFG:symForce}). This difference in the domains arises due to the logarithm present in the definition of the symmetric force. Such issues are typically dealt by first extending the domains of definition of the forces involved by appropriately regularising them, second by proving the decompositions on these extended domains, and finally passing to the limit in the regularisations (see for instance the proof of~\cite[Thm.~1.6]{HilderPeletierSharmaTse20}). Although we expect that similar arguments can be applied to~\eqref{eq:intro-L-decom} to arrive at global-in-time decompositions, in this first study we focus on local-in-time results. 

\subsection{Related work}
\label{subsec:intro-Connection}

As mentioned earlier, this work connects and generalises existing literature in various directions. Barring fairly recent works~\cite{KaiserJackZimmer2018,Renger2018b,RengerZimmer2021} which deal with particular examples, the connections between MFT, dissipative systems and FIR inequalities have largely been unexplored in the literature. 
Not all of these works consider fluxes, and so we will also make use of a `contracted' cost function,
\begin{equation}\label{eq:intro-contraction}
  \hat\L(\rho,u):=\inf\{\L(\rho,j): u=-\div j\},
\end{equation}
where the \emph{velocity} $u$ is a placeholder for $\dot\rho(t)$ and $-\div$ is the abstract operator that maps fluxes to velocities as in \eqref{eq:intro cont eq}. This construction is consistent with the notion of contraction in large deviations (see Example~\ref{ex:indep LDP}). Since $\hat\L(\rho,-\div j^0(\rho))=0$, we refer to $u^0(\rho):= -\div  j^0(\rho)$ as the \emph{zero-cost velocity}.

\subsubsection{Dissipative/Gradient systems}
\label{subsubsec:intro-diss}

In the case of dissipative systems $F=\Fsy$ and $\Fasy=0$, and therefore with $\lambda=\tfrac12$ both~\eqref{eq:Intro-FIIR-F} and~\eqref{eq:Intro-FIIR-sym} become
\begin{align}\label{eq:intro flux GF decom}             
  \L(\rho,j) &=\L_0(\rho,j) + \RF^{\frac12}_{\Fsy}(\rho) - \langle \Fsy(\rho),j\rangle \notag\\
             &=\Phi(\rho,j)+\Phi^*\big(\rho,\Fsy(\rho)\big)-\langle \Fsy(\rho),j \rangle,
\end{align}
with the convex dual pair of \emph{dissipation potentials} defined as $\Phi(\rho,j):=\L_0(\rho,j)$ and $\Phi^*(\rho,\zeta):=\sup_j \langle\zeta,j\rangle - \Phi(\rho,j)$. This decomposition of $\L$ is exactly the characterisation of dissipative systems in the density-flux setting~\cite{Maes2018,Renger2018b}; see  Section~\ref{subsec:flows} for a further elaboration.

Using~\eqref{eq:intro sym power}, $\Fsy=-\frac12\grad d\QP$ (see Corollary~\ref{def:sym-asym-force} for definition) and applying the contraction~\eqref{eq:intro-contraction}, we switch to the density setting
\begin{align}
  \hat\L(\rho,u)
    &=\inf\big\{\Phi(\rho,j): u=-\div j\big\}+\Phi^*\big(\rho,\Fsy(\rho)\big)+\big\langle \tfrac12 d\QP(\rho),u \big\rangle \notag\\
    &=:\hat\Psi(\rho,u) + \hat\Psi^*\big(\rho,-\tfrac12 d\QP(\rho)\big) + \big\langle \tfrac12 d\QP(\rho),u \big\rangle,
\label{eq:intro-L-decom-GFs}
\end{align}
where $\hat\Psi$ is the contraction of $\Phi$ and $\hat\Psi$,$\hat\Psi^*$ are convex duals of each other (see \cite[Thm.\ 3]{Renger2018b} for details).

The identity~\eqref{eq:intro-L-decom-GFs} is the standard decomposition of the density cost function that characterises a dissipative system or \emph{generalised gradient flow} in the following sense. For the zero-cost velocity, the left-hand side satisfies $\hat\L(\rho,u^0(\rho))=0$, and the right-hand side of~\eqref{eq:intro-L-decom-GFs} is the Energy--Energy-Dissipation identity (EDI)~\cite{SandierSerfaty04,AmbrosioGigliSavare08,RossiMielkeSavare08}, which is equivalent by convex duality to
\begin{equation}
  u^0(\rho) = d_\xi\hat\Psi^*\big(\rho,-\tfrac12 d\QP(\rho)\big),
\label{eq:GGS}
\end{equation}
where $d_\xi$ is the derivative with respect to the second argument.
In the special case when $\hat\Psi^*(\rho,\xi)=\tfrac12\langle K(\rho)\xi,\xi\rangle$ is a quadratic form with an inverse metric tensor $K(\rho)$ of a manifold,  we arrive at the usual gradient-flow representation of the zero-cost velocity on that manifold
\begin{equation*}
  u^0(\rho)=-\mfrac12 K(\rho) d\QP(\rho)=:-\mfrac12\Grad_{\rho}\QP(\rho).
\end{equation*}
This connection between generalised gradient flows and the symmetry $F=\Fsy$  at the level of densities has been explored more directly in~\cite{MielkePeletierRenger14}, where it was shown that this symmetry holds if $\hat\L$ corresponds to the large-deviation principle of a Markov process in detailed balance. The density-flux formulation~\eqref{eq:intro flux GF decom} of a dissipative system with quadratic dissipation has also been investigated extensively in the literature, see for instance ~\cite{BDSGJLL2015MFT,Maes2018,Renger2018b}. Since we derived this decomposition from~\eqref{eq:Intro-FIIR-F} and~\eqref{eq:Intro-FIIR-sym}, these two decompositions can be thought of as the natural generalisations of the EDI to non-dissipative systems.

\subsubsection{GENERIC}
\label{subsubsec:GENERIC}

The GENERIC framework is specifically designed as a coupling between dissipative and non-dissipative effects in a thermodynamically consistent way \cite{GrmelaOettinger1997I,GrmelaOettinger1997II,Ottinger2005}. Although originally meant to describe evolution equations, recent work has also studied the following natural connection between GENERIC and large deviations from a variational perspective (see~\eqref{eq:intro-L-decom-GFs}),
\begin{equation}
  \hat\L(\rho,u)
    =\hat\Psi\big(\rho,u-\pJ(\rho)d\E(\rho)\big) + \hat\Psi^*\big(\rho,-\tfrac12 d\QP(\rho)\big) + \big\langle \tfrac12 d\QP(\rho),u \big\rangle,
\label{eq:L=GENERIC}
\end{equation}
where the Poisson structure $\pJ$ and energy $\E$ define the Hamiltonian part of the dynamics, and additional non-interaction conditions are required to ensure that the zero-cost velocity
\begin{equation}
  u^0(\rho)=d\hat\Psi^*\big(\rho,-\tfrac12d\QP(\rho)\big) + \pJ(\rho)d\E(\rho)
\label{eq:GENERIC u dec}
\end{equation}
dissipates $\QP$ and conserves $\E$.

Such a connection is discussed in~\cite{DuongPeletierZimmer13} in the particular setting of weakly interacting diffusions and more recently in the context of hypocoercivity~\cite{DuongOttobre21}. More generally, the recent paper~\cite{KZMP2020math} shows that~\eqref{eq:L=GENERIC} can only hold if the underlying microscopic system consists of stochastic dynamics in detailed balance combined with a deterministic drift. The drift may be replaced by stochastic fluctuations as long as they appear deterministic on the large-deviation scale~\cite{Renger2018b}, but any larger scale fluctuations that are not in detailed balance will break down the GENERIC structure. Therefore, the class of large-deviation cost functions with a GENERIC structure is rather limited. 

By contrast, the decompositions~\eqref{eq:intro-L-decom} \emph{always} hold as soon as the quasipotential $\QP$ is identified. The crucial difference is that our decompositions are based on a decomposition of forces, i.e.
\begin{align*}
  u^0(\rho)=-\div j^0(\rho)=-\div d\Phi^*\big(\rho,\Fsy(\rho)+\Fasy(\rho)\big),
\end{align*}
rather than a decomposition of fluxes or velocities as in GENERIC~\eqref{eq:GENERIC u dec}.
Furthermore, generalised orthogonality between $\Fsy$ and $\Fasy$ (see Subsection~\ref{sec:Gen-ortho}) is a natural analogue of the non-interaction conditions used in GENERIC.

\subsubsection{FIR inequalities}
\label{subsubsec:FIR-intro}

Using $\L_{F-2\lambda \Fsy}\geq 0$ and $\Fsy=-\tfrac12\grad d\QP$ (as above) in the decomposition~\eqref{eq:Intro-FIIR-sym}, we find
\begin{equation*}
\tfrac{1}{\lambda}\L(\rho,j)\geq  \tfrac{1}{\lambda} \RF^\lambda_{\Fsy}(\rho) + \langle \grad d\QP,j\rangle.
\end{equation*}
Since $\grad$ is the dual of $-\div$, using the contraction principle~\eqref{eq:intro-contraction} and the definition of the Fisher information~\eqref{def:intro-gFI} it follows that (see Corollary~\ref{corr:FIR-abstract} for details)
\begin{equation}\label{eq:intro-FIR}
\tfrac{1}{\lambda}\hat\L(\rho,u)\geq  -  \tfrac{1}{\lambda}\hat\H(\rho,d\QP(\rho) ) + \langle  d\QP(\rho),u\rangle,
\end{equation}
where $\hat\H$ is the convex dual of $\hat \L$. This is a local-in-time version of the FIR inequality. 

Assume that a smooth trajectory $[0,T]\ni t\mapsto \rho(t)$ satisfies~\eqref{eq:intro-FIR} for every $t$. Substituting $u=\dot\rho$, formally applying the chain rule $\langle  d\QP(\rho),\dot\rho\rangle = \tfrac{d}{dt}\QP(\rho)$, and integrating in time over $[0,T]$ we arrive at the F(``free energy'')-I(``rate functional'')-R(``Fisher information'') inequality~\cite[Thm.~1.6]{HilderPeletierSharmaTse20}
\begin{equation}\label{eq:intro-FIR-global}
  \frac1\lambda\int_0^T \hat \L(\rho(t),\dot\rho(t))dt + \QP(\rho(t)) \geq \QP(\rho(T)) - \frac1\lambda\int_0^T\!\hat\H\big(\rho(t),d\QP(\rho(t))\big)\,dt.
\end{equation}

Therefore, the decomposition~\eqref{eq:Intro-FIIR-sym} can be thought of as a generalisation of~\cite{HilderPeletierSharmaTse20} in various ways. First,~\eqref{eq:Intro-FIIR-sym} holds fairly generally (in the abstract framework) and can be applied to systems well beyond independent copies of Markov jump processes studied in~\cite{HilderPeletierSharmaTse20}. Second,~\eqref{eq:Intro-FIIR-sym} exactly characterises the gap in the inequality~\eqref{eq:intro-FIR} via  $\L_{F-2\lambda \Fsy}$ which we discarded in this discussion due to its non-negativity. And third, a different version of the FIR inequality can also be derived from~\eqref{eq:Intro-FIIR-asym}. 

It should be noted that the FIR inequalities have been used in the literature as \emph{a priori estimates} to study singular limits, and we expect that the decomposition~\eqref{eq:Intro-FIIR-sym} and inequality~\eqref{eq:intro-FIR} will serve the same purpose for a considerably larger class of systems.  However, in this paper we limit ourselves to the local-in-time decompositions~\eqref{eq:Intro-FIIR-sym} as opposed to the global-in-time inequality~\eqref{eq:intro-FIR-global} discussed in~\cite{HilderPeletierSharmaTse20}, since moving from local to global descriptions is a nontrivial technical step outside the scope of this work.

\subsubsection{MFT and (non-)quadratic cost function}

As stated earlier, most MFT literature is concerned with the diffusive scaling of underlying stochastic particle systems which converge to diffusion-type macroscopic partial differential equations and corresponds to quadratic cost functions of the form~\cite{BDSGJLL2015MFT}
\begin{equation*}
  \L(\rho,j)=\mfrac12 \| j-j^0(\rho)\|^2_{\rho}, \quad\text{for some Hilbert norm $\| \cdot\|_{\rho}$}.
\end{equation*}
Crucial arguments in MFT are based on the fact that the dissipative and the non-dissipative effects are orthogonal in this Hilbert space, i.e.\
\begin{equation*}
  \langle \Fsy(\rho),\Fasy(\rho)\rangle_\rho\equiv0.
\end{equation*}

However, even the simple example of independent particles on a finite graph (see Example~\ref{ex:indep LDP}) yields a non-quadratic cost function $\L$, and the aforementioned orthogonality arguments break down. In~\cite{KaiserJackZimmer2018} (for independent jump processes) and~\cite{RengerZimmer2021} (for chemical reactions) these ideas are ported to the non-quadratic setting by introducing a generalised notion of orthogonality, where the pairing is no longer bilinear, and rather satisfies a relation of the form
\begin{equation}\label{eq:intro gen orth}
  \theta_\rho(\Fsy(\rho),\Fasy(\rho))\equiv0.
\end{equation}

By contrast, the abstract theory that we develop is not necessarily based on such orthogonality relations, although we do borrow many notions such as time-reversed cost-functions and forces from MFT. However we will show that within our framework, one can also construct a generalised orthogonality pairing $\theta_\rho$ (fully characterised by $\L$) that satisfies \eqref{eq:intro gen orth}, and coincides with the bilinear pairings $\langle\cdot,\cdot\rangle_\rho$ in case of quadratic cost functions and with $\theta_\rho(\cdot,\cdot)$ from \cite{KaiserJackZimmer2018,RengerZimmer2021} in the case of specific non-quadratic cost functions. This will be the content of Subsection~\ref{sec:Gen-ortho}.

\subsection{Summary of notation and outline of the article}
\begingroup
\centering
\begin{longtable}{c l c}
$\X$ & Finite graph with strict ordering \\
$\fluxS$ & Half the edges on a finite graph $\X$  &~\eqref{def:HalfEdges} \\ 

$s(\cdot|\cdot)$ & Relative Boltzmann function (integrand/summand in relative entropy)  &~\eqref{def:entropy-s} \\ 

$\Z,\W,\phi$ & State-flux triple & Def.~\ref{def:state-flux-triple}\\

$T\Z$, $T^*\Z$ & Tangent and cotangent bundle associated to $\Z$  &\\ 

$T_\rho\Z$, $T_\rho^*\Z$ & Tangent and cotangent space at $\rho\in\Z$ &\\

$\L$, $\H$ & L-function and its convex dual & Def.~\ref{def:L-function}\\

$\hat\L$, $\hat\H$ & Contracted L-function and its convex dual & \eqref{def:FIIR-FIR-contract}\\

$\QP$ & Quasipotential &  Def.~\ref{def:I0}\\

$d\mathcal F$ & Gateaux derivative of a functional $\mathcal F$&\\

$\chi\tp$ & transpose or adjoint operator $\chi\tp:\mathcal M^*\rightarrow\mathcal N^*$ for $\chi:\mathcal N\rightarrow\mathcal M$ &\\

$\Dom(A)$ & domain of an operator $A$ & \\

$F$ & Driving force & Def.~\ref{def:F-disspot}\\

$\Phi^*$, $\Phi$ & Dissipation potential and its dual & Def.~\ref{def:F-disspot}\\

$\hat\Psi^*$, $\hat\Psi$ & Contracted dissipation potential and its dual & \eqref{def:FIIR-FIR-contract-diss}\\

$\L_\CoF$, $\H_\CoF$ & Tilted L-function and its convex dual & Def.~\ref{def:tilt}  \\


$\Dom_\symdiss(A)$ & Subset of $\Dom(A)$ where the dissipation potential is symmetric &~\eqref{eq:symmetric dissipation}\\ 

$\RF^\lambda_\zeta$ & Generalised Fisher information & Def.~\ref{def:genFI}\\

$\overleftarrow\L$, $\overleftarrow\H$ & Reversed L-function and its convex dual & Def.~\ref{def:time-reversedL}\\

$\Fsy$, $\Fasy$ & Symmetric and antisymmetric force & Cor.~\ref{def:sym-asym-force}\\


$\mathcal{M}(\mathcal X)$, $\mathcal{M}_a(\mathcal X)$ & Space of signed measures on $\mathcal X$ (with total mass $a$) & \eqref{def:FinMass-Signed}\\\ 
 
$\mathcal{P}(\mathcal X)$ & Space of probability measures on $ \mathcal X$ &  \\ 

$\mathcal{P}_+(\mathcal X)$ & Space of strictly positive probability measures on a discrete state space $ \mathcal X$ &  \\ 

$\grad,\div$ & Continuous gradient and divergence & \\

$\Dgrad,\Ddiv$ & Discrete gradient and divergence & \eqref{eq:discrete divergence}\\

$\mathds1_x$ & Indicator function associated to $\{x\}$ &  
\end{longtable}
\endgroup

In Section~\ref{sec:abstract} we present the abstract framework and theory. In Section~\ref{sec:IPFG antisym flow} we analyse the zero-cost velocity for the antisymmetric L-function in the setting of independent particles on a finite graph. In Section~\ref{sec:examples} we apply the abstract theory to various stochastic particle systems and conclude with discussion  in Section~\ref{sec:discuss}. In Section~\ref{sec:large deviations} we connect (and thereby motivate) the abstract ideas developed in Section~\ref{sec:abstract} to large deviations.

\section{Abstract theory}\label{sec:abstract}
%

In the introduction we worked with the large-deviation cost; we now work with its abstraction, the so-called the L-function\footnote{We use the terminology ``L-function'' from \cite[Def.~1.1]{MielkePeletierRenger14} as opposed to `Lagrangian' or `cost', since in
practice $\L$ need not correspond to a large-deviation principle, and it often plays a different role as the Lagrangian in mechanics.}. In what follows we first introduce the L-function and other key ingredients of the abstract framework in Section~\ref{subsec:abstract setup}. 
Using these objects we introduce dissipation potentials, tilted L-functions and Fisher information in Section~\ref{subsec:abstract forces}. Using time-reversal-type arguments from MFT, in Section~\ref{subsec:abstract time reversal} we introduce time-reversed L-functions, symmetric and  antisymmetric forces, and in Section~\ref{sec:Gen-ortho} we introduce a  generalised notion of orthogonality satisfied by these forces. Section~\ref{sec:abstract FIIR equalities} contains various decompositions of the L-function and in Section~\ref{subsec:flows} we study the symmetric and antisymmetric L-function. Throughout this section we will use the guiding example of \underline{I}ndependent Markovian \underline{P}articles  on a \underline{F}inite \underline{G}raph (IPFG), which we now introduce. 

\begingroup
\allowdisplaybreaks
\begin{example}\label{ex:indep LDP}

Let $\X$ be a finite graph with strict ordering. Consider $n$ independent Markovian particles $X_1(t),\hdots X_n(t)$ on 
$\X$, with irreducible generator $Q\in\RR^{\X\times\X}$. The particle density (also called empirical measure or mean field), defined as $\rho\super{n}(t):=n^{-1}\sum_{i=1}^n\delta_{X_i(t)}$, is a Markov process on $\RR^\X$
with generator
\begin{equation*}
  (\hat\Q\super{n}f)(\rho)=n\sumsum_{(x,y)\in\X\times\X} \rho_x Q_{xy}\big\lbrack f(\rho-\tfrac1n\mathds1_x+\tfrac1n\mathds1_y)-f(\rho)\big\rbrack,
\end{equation*}
where $\mathds1_x$ is the indicator function for $x\in\X$.
With a suitable initial condition, Varadarajan's Theorem implies that the random process $\rho\super{n}$ converges in the many-particle limit $n\to\infty$ to the deterministic solution of the ODE
\begin{equation}\label{eq:for-Kol}
\dot\rho(t) = Q\tp\rho(t).
\end{equation}

In addition to the empirical measure, we will also track the number of jumps through each edge, which characterises the flux over an edge. For reasons that will be clarified in Section~\ref{subsec:abstract forces}, it is important to consider net fluxes (over the usual one-sided fluxes), defined on half of the edges (for this purpose we impose an arbitrary ordering $<$ on the finite set $\X$)
\begin{equation}\label{def:HalfEdges}
  \fluxS:=\big\{(x,y)\in\X\times\X:x<y\big\}. 
\end{equation}
More precisely, the so-called integrated net flux $W\super{n}_{xy}(t)$ over the edge connecting $x,y\in\X$, is defined as the difference between the number of jumps from $x\to y$ and in the opposite direction from $y\to x$ in the time interval $[0,t]$, 
all rescaled by $\frac1n$. Then the pair $(\rho\super{n}(t),W\super{n}(t))$ is again a Markov process, now in $\RR^\X\times\RR^{\fluxS}$ with the generator
\begin{align*}
  (\Q\super{n}f)(\rho,w)=n\sumsum_{(x,y)\in \fluxS} & \rho_x Q_{xy}\big\lbrack f(\rho-\tfrac1n\mathds1_x+\tfrac1n\mathds1_y,w+\tfrac1n\mathds1_{xy})-f(\rho,w)\big\rbrack \\
  &  \ +   \rho_y Q_{yx}\big\lbrack f(\rho-\tfrac1n\mathds1_y+\tfrac1n\mathds1_x,w-\tfrac1n\mathds1_{xy})-f(\rho,w)\big\rbrack.
\end{align*}
This process converges as $n\to\infty$ to the solution of the macroscopic system
\begin{equation}\label{eq:indep evolution flux}
\begin{cases}
  \dot w_{xy}(t) = \rho_x(t) Q_{xy} - \rho_y(t) Q_{yx}, & (x,y)\in \fluxS,\\
  \dot\rho_x(t) = -\Ddiv_x \dot w(t), &x\in\X,\\
\end{cases}
\end{equation}
where the operator  
\begin{equation}\label{eq:discrete divergence}
  \Ddiv_x j:=\sum_{y\in\X:y> x}j_{xy} -\sum_{y\in\X:y<x} j_{yx},
\end{equation}
is the discrete divergence for net fluxes. Indeed the system~\eqref{eq:indep evolution flux} is of the form~\eqref{eq:coupled ev eq}. 

In the many-particle limit ($n\rightarrow\infty$), the random fluctuations around the mean behaviour decay fast due to averaging effects. The unlikeliness to observe an atypical flux for large but finite $n$  is quantified by the large-deviation principle, formally written as
\begin{equation}
  \Prob\Big( (\rho\super{n},W\super{n})\approx (\rho,w)\Big) \stackrel{n\to\infty}{\sim} e^{-n\I_0(\rho)-n\J(\rho,w)},
\quad 
\J(\rho,w):=
  \begin{cases}
    \int_0^T\!\L\big(\rho(t),\dot w(t)\big)\,dt, &\dot\rho=-\Ddiv \dot w,\\
    \infty,  &\text{otherwise},
  \end{cases}
\label{eq:flux LDP}
\end{equation}
where the $\L$ is given by~\cite{Renger2018a,Kraaij2017} (the flux $j$ is a placeholder for $\dot w$)
\begin{align}
  \L(\rho,j) :=
    \inf_{j^+\in\RR^{\fluxS}_{\geq0}}
    \sumsum_{(x,y)\in \fluxS}\bigl[   s(j^+_{xy} \mid \rho_x Q_{xy}) + s(j^+_{xy}-j_{xy} \mid \rho_y Q_{yx})\bigr], 
  \label{eq:indep L}
\end{align}
which uses the Boltzmann function
\begin{align}
  s(a\mid b) :=
    \begin{cases}
      a\log\frac{a}{b}-a+b, &a,b>0,\\
      b,                    &a=0, b\geq 0\\
      \infty,               &\text{otherwise}.
    \end{cases}
    \label{def:entropy-s}
\end{align}
Here $\I_0$ is the large-deviation rate functional corresponding to the initial distribution of $\rho\super{n}(0)$.
Indeed $\L(\rho,j)$ is non-negative and minimised by \eqref{eq:indep evolution flux}. Due to the contraction principle~\cite[Thm.~4.2.1]{DemboZeitouni09}, the infimum is taken over all non-negative one-way fluxes $(j^+_{xy})_{x<y}$ and $(j^+_{yx}-j_{yx})_{x>y}$.

Applying the contraction principle, the empirical measure satisfies the following large-deviation principle, where $\hat\L$ is related to $\L$ via \eqref{eq:intro-contraction},
\begin{equation*}
  \Prob\Big(\rho\super{n}\approx \rho\Big) \stackrel{n\to\infty}{\sim} \exp \Big\lbrack -n\I_0(\rho(0))-n\int_0^T\!\hat\L(\rho(t),\dot\rho(t))\,dt\Big\rbrack.
\end{equation*}
\end{example}
\endgroup

\subsection{Abstract framework}
\label{subsec:abstract setup}

Although at first sight the general setup in this section may seem heavy, it appears naturally in various specific systems. We illustrate this via our guiding example.

\begin{example}\label{ex:IPFG-state-flux-triple} Consider the example of the independent particles on a finite graph $\X$. Let 
\begin{equation}\label{def:FinMass-Signed}
	\textstyle\M_a(\X):=\{\rho\in\RR^\X:\sum_{x\in\X}\rho_x=a\},
\end{equation}
including vectors with negative coordinates. The states/densities $\rho$ lie in the manifold $\Z:=\M_1(\X)$. Due to the constraint on total mass, $\Z$ is a $(|\X|-1)$-dimensional hyperplane in $\RR^{\X}$, with corresponding local tangent, cotangent spaces and Euclidean pairing between them given by
\begin{equation}\label{eq:IPFG state space}
\begin{gathered}
  T_\rho\Z=\M_0(\X), \qquad 
  T^*_\rho\Z=\RR^{\X}/\mathrm{span{\{(1,1,\ldots,1)\}}}=\{\{\xi+c(1,\hdots,1):c\in\RR\}:\xi\in\RR^\X\},\\
  _{T_\rho^*\Z}\langle \xi,u\rangle_{T_\rho\Z}:=\xi\cdot u,
\end{gathered}
\end{equation}
where $a\cdot b$ is the usual dot product in Euclidean spaces. Cotangents are defined modulo the orthogonal space $(\M_0(\X))^\perp = \mathrm{span}\{(1,1,\ldots,1)\}$, and lead to $\langle \xi+c(1,\hdots,1),u\rangle=\xi\cdot u+c\sum_{x\in\X}u_x=\xi\cdot u$. The integrated net fluxes $w$ simply lie in the Euclidean ``flux space'' $\W:=\RR^{\fluxS}$ (recall \eqref{def:HalfEdges}) with local tangent and cotangent spaces $T_w\W=T_w^*\W=\RR^{\fluxS}$, again paired together with the Euclidean inner product.

Between the two manifolds we define the map $\phi:\W\to\Z$ as 
\begin{align*}
  \phi[w]:=\rho^0-\Ddiv w, &&\text{with differential } d\phi_w =-\Ddiv &&\text{and adjoint operator } d\phi_w\tp=\Dgrad,
\end{align*}
where $\Ddiv$ is the discrete divergence from~\eqref{eq:discrete divergence}, $\Dgrad_{xy}\xi:=\xi_y-\xi_x$ and $\rho_0\in\Z$ is an arbitrary but fixed reference measure. Hence the continuity equation can be abstractly written as $u=d\phi_w j \in T_{\phi\lbrack w\rbrack}\Z$ for $j\in T_w\W$. 
It will be important that the operator $\phi$ is surjective. For an arbitrary $\mu\in\M_1(\X)$, the difference $\mu-\rho^0\in\M_0(\X)$.

Note that the underlying dynamics~\eqref{eq:indep evolution flux} as well as any path with $\J(\rho,w)<\infty$ conserves the total mass as well as the non-negativity of $\rho(t)$, so that the states will in fact be restricted to the simplex $\P(\X)\subset\M_1(\X)\subset\RR^\X$ of probability measures on $\X$ (i.e.\ coordinate-wise non-negative vectors in $\RR^{\X}$ which sum to one).
However, we always work with the full manifold $\M_1(\X)$ so that derivatives and the (co)tangent spaces are well defined without needing to worry about boundaries, boundary points etc. Instead we set $\L(\rho,j)=\infty$ whenever $\rho$ lies on (or outside of) the boundary $\partial\P(\X)$ and the flux $j\in T_\rho\W$ pushes the state in the outward direction. Indeed, the functional $\J(\rho,w)$ and cost $\L(\rho,j)$ from Example~\ref{ex:indep LDP} are defined for all $\rho\in\Z=\RR^\X$, but for any path with $\J(\rho,w)<\infty$, the densities are contained in $\P(\X)$.

\end{example}

For the above example $d\phi_w,d\phi\tp_w$ and the (co)tangent spaces $T_w\W,T_w^*\W$ do not depend on $w$. In practice, $d\phi_w,d\phi\tp_w$ and $T_w\W,T_w^*\W$ might depend on $w$, but only through the corresponding state $\rho=\phi\lbrack w\rbrack$, as for example in a contuinity equation of the form $v=-\div(\rho j)$. By a slight abuse of notation we shall therefore write $d\phi_\rho,d\phi\tp_\rho$ and $T_\rho\W,T_\rho^*\W$ for $\rho\in\Z$. In particular, this allows us to write $\L:T\W\to\RR\cup\{\infty\}$, so that $\L=\L(\rho,j)$ for $(\rho,j)\in T\W$.

Inspired by these observations we now introduce the \emph{state-flux triple}, \emph{L-function} and the \emph{quasipotential}, which are the key ingredients in the abstract framework.

\begin{definition}[{\cite[Sec.\ 4.1]{Renger2018b}}]\label{def:state-flux-triple}
A triple $(\Z,\W,\phi)$ is called a \emph{state-flux triple} if 
\begin{enumerate}[label=(\roman*)]
\item The state-space $\Z$ and the flux-space $\W$ are differentiable Banach manifolds, with corresponding local tangent Banach spaces $T_\rho\Z$ and $T_w\W$.
\item $\phi:\W\to\Z$ is a surjective differentiable operator $\phi:\W\to\Z$. 
\item $T_w\W$ depends on $w$ only through $\rho=\phi[w]$, so that by a slight abuse of notation we can replace  $T_w\W$ by $T_\rho\W$ and write $T\W:=\{(\rho,j):\rho\in\Z, j\in T_\rho\W\}$.
\item $\phi$ has a linear bounded differential that depends on $w$ only through $\rho=\phi[w]$, so that by a slight abuse of notation we write $d\phi_\rho: T_\rho\W\to T_\rho\Z$.
\end{enumerate}
\end{definition}

The Banach structure should be seen as a reference norm only, that we use to define Gateaux derivatives, the Banach dual spaces $T_\rho^*\W, T_\rho^*\Z$ and the duality pairings $_{T_\rho^*\Z}\langle\cdot,\cdot\rangle_{T_\rho\Z}$, $_{T_\rho^*\W}\langle\cdot,\cdot\rangle_{T_\rho\W}$ (where we omit the indices since it will be clear to which spaces the elements belong). Analogously we write $T^*\W:=\{(\rho,\zeta):\rho\in\Z, \zeta\in T_\rho^*\W\}$ and $T^*\Z:=\{(\rho,\xi):\rho\in\Z, \xi\in T_\rho^*\Z\}$. The differential $d\phi_\rho$ corresponds to a continuity equation $u=d\phi_\rho j$, where $d\phi_\rho$ is usually minus a divergence operator or some generalisation thereof. 
The assumption that $d\phi$ is bounded, ensures the existence of a well-defined adjoint. In order to avoid confusion with convex duality, we will denote adjoint operators by $\mathsf{T}$, e.g. $d\phi_\rho\tp:T_\rho^*\Z\to T_\rho^*\W$.

\begin{remark}
Our state-flux triple is essentially identical to the framework of \cite{ACEGKP2023}; there $\Z$ is called the `base manifold', $T\W$ is called the `total manifold', and the differential $d\phi:T\W\to T\Z$ is called the `anchor map'.
\end{remark}

    
\begin{definition}\label{def:L-function} 
For any $\S\subseteq \Z$ define 
\begin{align}\label{eq:subset tangent spaces}
  T_\S\W:=\{(\rho,j)\in T\W:\rho\in \S\}
  &&\text{and}&&
  T_\S^*\W:=\{(\rho,\zeta)\in T^*\W:\rho\in \S\}.
\end{align}

A mapping $\L:T_\S\W\to\RR\cup\{\infty\}$ is called an \emph{L-function on $\S$}, if for all $\rho\in\S$:
\begin{enumerate}[label=(\roman*)]
\item\label{def:L-inf} $\inf \L(\rho,\cdot)=0$,
\item there exists a unique $j^0(\rho)\in T_\rho\W$, called the \emph{zero-cost flow}, which satisfies $\L\big(\rho,j^0(\rho)\big)=0$,
\item\label{def:L-con-lsc} $\L(\rho,\cdot)$ is convex and lower semicontinuous (with respect to the Banach norm on $T_\rho\W$).
\end{enumerate}
\end{definition}
While this definition allows for flexibility in the domain, throughout this paper we will reserve the symbol $\L$ for L-functions on the full space $\mathcal S=\Z$. 
From Section~\ref{subsec:abstract forces} onwards we will encounter functions $\L_\CoF$ that are only defined on proper subsets of $\Z$ (see  Remark~\ref{rem:local vs global} below). The inclusion of $\infty$ in the codomain of $\L$ is essential to encode forbidden fluxes as discussed in Example~\ref{ex:IPFG-state-flux-triple}.

By lower semicontinuity and convexity, $\L(\rho,\cdot)$ is its own convex bidual with respect to the second variable~\cite[Prop.~3.56]{Peypouquet2015}, i.e.\ there exists an $\H:T^*_\S\W\to\RR\cup\{\infty\}$ such that
\begin{equation}\label{eq:L and H}
  \H(\rho,\zeta):=\sup_{j\in T_\rho\W} \langle \zeta,j\rangle - \L(\rho,j) 
  \qquad\text{and}\qquad
  \L(\rho,j)=\sup_{\zeta\in T_\rho^*\W} \langle \zeta,j\rangle - \H(\rho,\zeta).
\end{equation}
It is easy to see that $\L$ is an L-function if and only if for any $\rho\in\Z$, $\H(\rho,0)=0$, $\H(\rho,\cdot)$ is convex, lower semicontinuous, proper and bounded from below by an affine function. Typically $\L(\rho,0)<\infty$, so that $\H(\rho,\cdot)$ is bounded from below.

We are now ready to introduce the following notion of the quasipotential.
\begin{definition}\label{def:I0}
A function $\QP:\Z\to\RR\cup\{\infty\}$ is called a \emph{quasipotential} (corresponding to $\L$) if
\begin{enumerate}[label=(\roman*)]
\item $\inf\QP=0$,
\item for any $\rho\in\Z$ where $\QP$ is Gateaux differentiable, we have
\begin{equation}\label{eq:invariant I_0}
  \H\big(\rho,d\phi_\rho\tp d\QP(\rho)\big)=0.
\end{equation}
\end{enumerate}
\end{definition}
We stress that this notion of a quasipotential is only related to the convex dual $\H$ of some abstract function $\L$, where a priori no stochastic particle system is involved. Both nowhere differentiable functions and the zero function are quasipotentials by definition, and our results are true but mostly trivial in this setting. In all the examples we consider,~\eqref{eq:invariant I_0} will have at least one non-trivial solution and in fact this definition is consistent with the the usual definition from statistical physics when large deviations are involved (see Section~\ref{sec:QuasiLDP}). We envisage that \eqref{eq:invariant I_0} should be understood in the sense of viscosity solutions, however it is not clear how one can define a viscosity solution in the general setup of this section. 

\begingroup
\allowdisplaybreaks
\begin{example} 
In Example~\ref{ex:indep LDP}, the processes $X_1(t),X_2(t),\hdots$ are irreducible and $\X$ is finite which ensures the existence of an invariant measure $\pi\in\P_+(\X)$ (the space of strictly positive probability measures). Consequently, the $n$-particle density $\rho\super{n}(t)$ admits an invariant measure $\Pi\super{n}\in\P(\RR^\X)$, where
\begin{equation*}
  \Pi\super{n}=\Big({\textstyle\bigotimes_{i=1}^n\pi}\Big)\circ \eta_n^{-1}, \qquad \eta_n(x_1,\hdots,x_n):=\tfrac1n\sum_{i=1}^n\delta_{x_i}.
\end{equation*}
By Sanov's theorem, the large-deviation rate functional corresponding to $\Pi\super{n}$ is
\begin{equation*}
  \QP(\rho):=
  \begin{cases}
    \sum_{x\in\X} s(\rho_x \mid \pi_x), &\rho\in\P(\X),\\
    \infty,                             &\rho\notin\P(\X),
  \end{cases}
\end{equation*}
%
where $s(\cdot\mid\cdot)$ is defined in~\eqref{def:entropy-s}, and hence $\QP$ is indeed the quasipotential corresponding to $\L$ in the classical large-deviation sense (see Theorem~\ref{th:LDP QP}).

This can also be checked macroscopically by verifying~\eqref{eq:invariant I_0}, without invoking any connection to large deviations of a microscopic particle system. To check this, we first calculate the convex dual of the L-function~\eqref{eq:indep L}:
\begin{equation*}
  \H(\rho,\zeta):=\sumsum\limits_{(x,y)\in\fluxS} \big[ \rho_xQ_{xy}\big(e^{\zeta_{xy}}-1\big) + \rho_yQ_{yx}\big(e^{-\zeta_{xy}}-1\big)\big].
\end{equation*}
Note that while $\QP(\cdot)$ would be nowhere differentiable as a functional on $\RR^\X$, it is differentiable at all $\rho \in \P_+(\X)$ (which is a subset of the manifold $\M_1(\X)$ introduced in Example~\ref{ex:IPFG-state-flux-triple}) since $\pi_x>0$ for every $x\in\X$ with Gateaux derivative
\begin{align*}
	d\V(\rho) = \big\{(\log(\rho_x/\pi_x)+c)_{x\in\X}:c\in\RR\big\} \in T_\rho^*\Z,
\end{align*}
so that $d\phi_\rho d\V(\rho) = \Dgrad d\V(\rho) = \big( \log(\rho_y/\pi_y) - \log(\rho_x/\pi_x) \big)_{x<y} \in T_\rho^*\W$. 
In fact by the chain rule, $\Dgrad d\V(\rho)$ can also be interpreted as the (classical) derivative of $\V(\phi\lbrack w\rbrack)$ with respect to $w\in\RR^{\fluxS}$; this also explains why the constants $c$ do not play a role after taking the discrete gradient. We then check that $\V$ is a quasipotential by concluding that at all points of differentiability of $\QP$ (i.e. for $\rho\in \P_+(\X)$) using $Q\tp\pi=0$ and $\sum_y Q_{xy}=0$ we find
\begin{align*}
  \H\big(\rho,d\phi_\rho\tp d\QP(\rho)\big) 
&=\sumsum_{(x,y)\in\fluxS}  \Bigl( \rho_xQ_{xy}\Bigl[\mfrac{\rho_y\pi_x}{\rho_x\pi_y} -1\Bigr] + \rho_yQ_{yx}\Bigl[\mfrac{\rho_x\pi_y}{\rho_y\pi_x} -1\Bigr]\Bigr)\\
&=\sumsum_{\substack{x,y\in\X\\x\neq y}}  \mfrac{\rho_y}{\pi_y} \left( Q_{xy}\pi_x - Q_{yx}\pi_y \right)  =\sumsum_{x,y\in\X} Q_{xy}\pi_x \bigl(\mfrac{\rho_y}{\pi_y}-\mfrac{\rho_x}{\pi_x} \bigr)=\sum_{y\in\X} (Q\tp \pi)_y\mfrac{\rho_y}{\pi_y} = 0,
\end{align*}
where the third and fourth equality follows by interchanging the indices in the second terms of the summation.
\label{ex:QP}
\end{example}
\endgroup

\begin{remark}\label{rem:local vs global}
Most of the analysis that follows will be carried out locally for fixed $\rho$. Therefore the $\rho$-dependencies in $\L(\rho,j)$ and $d\phi_\rho$ do not play a role in the calculations. We however include the dependency for two reasons. First, for almost all practical applications, $\L$ and $d\phi_\rho$ will depend on $\rho$, either explicitly or implicitly through the domains of definition $T_\rho\W, T_\rho\Z$. Second, even though writing the $\rho$-dependency is standard in the literature, so far practically all literature on the topic completely ignores the problems at the boundaries, where $\V$ may cease to be differentiable due to the appearance of $\log0$. Our paper is one of the first to make completely precise claims in regards to domain of definitions for various objects involved by very carefully identifying all points $\rho$ for which our results hold; this also motivates the definition of L-functions on subsets $\S$.
\end{remark}

\subsection{Dissipation potentials, tilted L-functions and Fisher information}\label{subsec:abstract forces}
While the concept of a dissipation potential is standard \cite{ColVis90CDNE,Luckhaus1995,Mielke2011}, the connection to convex analysis \cite{MielkePeletierRenger14} and the application to flux spaces is more recent~\cite{Maes2008a,Maes2017a,KaiserJackZimmer2018,Renger2018a, Renger2018b}. Classically, a dissipation potential $\Phi(\rho,j)$ is convex, lower semicontinuous in the second variable, and satisfies $\inf\Phi(\rho,\cdot)=0=\Phi(\rho,0)$. To define the dissipation potential in our context, we first present the following basic result on $\L$, which was originally derived in the context of gradient flows~\cite[Lem.~2.1 \& Prop.~2.1]{MielkePeletierRenger14}, where the driving force is the derivative of a certain free energy. As in the literature~\cite{Schnakenberg1976a,Maes2008a,Maes2017a,KaiserJackZimmer2018, Renger2018a,RengerZimmer2021}, the setting with fluxes allows for more general driving forces. 
We first focus on a driving force $\hat\zeta\in T_\rho^*\W$ for a fixed $\rho$; and later introduce it as a $\rho$-dependent force field $F(\rho)$. 

\begin{theorem}{\cite[Prop.~2.1(i)]{MielkePeletierRenger14}}\label{th:MPR force}
Let $\L$ be an L-function on $\Z$ and fix $\rho\in\Z$. For any $\hat\zeta\in T_\rho^*\W$ and convex lower-semicontinuous $\Phi(\rho,\cdot):T_\rho\W\to\RR\cup\{\infty\}$ with convex dual $\Phi^*$, the following statements are equivalent
\begin{enumerate}[label=(\roman*)] 
\item $\inf\Phi(\rho,\cdot)=0=\Phi(\rho,0)$, and for any $j\in T_\rho\W$ 
\begin{equation}\label{eq:L=Phi Phis1}
  \L(\rho,j)=\Phi(\rho,j) + \Phi^*(\rho,\hat\zeta) - \langle \hat\zeta, j\rangle.
\end{equation}
\item $-\hat\zeta\in\partial\L(\rho,0)$ with 
\begin{equation}\label{eq:Phis1}
  \Phi^*(\rho,\zeta)=\H(\rho,\zeta-\hat\zeta) - \H\big(\rho,-\hat\zeta\big).
\end{equation}
\end{enumerate}
\end{theorem}

We would like to define the driving force as $F(\rho)=\hat\zeta$ and the dissipation potential $\Phi(\rho,j)$ as above. However these exist uniquely only if the subdifferential $\partial\L(\rho,0)$ consists of a singleton, i.e.\ $\L(\rho,\cdot)$ is Gateaux differentiable at $0$, which motivates the following definitions.  

\begin{definition}\label{def:F-disspot} Let $\L$ be an L-function on $\Z$. Define
\begin{equation*}
  \Dom(F):=\big\{\rho\in\Z: j \mapsto \L(\rho,j) \text{ is Gateaux differentiable at } j=0\big\},
\end{equation*}
and recall the definition of the restricted (co)tangent spaces~\eqref{eq:subset tangent spaces}. The driving force $F$ and dissipation potentials (corresponding to $\L$) are defined as
\begin{align}
  F(\rho) &:=-d\L(\rho,0)\in T_\rho^*\W                                                 &&\text{for } \rho\in\Dom(F), \label{eq:F}\\
  \Phi^*(\rho,\zeta) &:=\H\big(\rho,\zeta-F(\rho)\big) - \H\big(\rho,-F(\rho)\big),     &&\text{for } (\rho,\zeta)\in T^*_{\Dom(F)}\W, \label{eq:Phis}\\
  \Phi(\rho,j) &:=\sup_{\zeta\in T_\rho^*\W} \langle\zeta,j\rangle - \Phi^*(\rho,\zeta) &&\text{for }  (\rho,j)\in T_{\Dom(F)}\W. \notag
\end{align}

\end{definition}
Note that, $\Phi^*$ as defined in~\eqref{eq:Phis} indeed satisfies $\inf\Phi^*(\rho,\cdot)=0=\Phi^*(\rho,0)$, since $-F$ is a minimiser of $\H(\rho,\cdot)$ by~\eqref{eq:F}, and consequently $\inf\Phi(\rho,\cdot)=0=\Phi(\rho,0)$ which makes $\Phi$ a dissipation potential. Furthermore combining Theorem~\ref{th:MPR force} with Definition~\ref{def:F-disspot}, for any $(\rho,j)\in T_{\Dom(F)}W$ we have the decomposition
\begin{equation}\label{eq:L-decom-F}
 \L(\rho,j)=\Phi(\rho,j) + \Phi^*(\rho,F) - \langle F, j\rangle.
\end{equation}
In what follows we will make use of 
\begin{equation}\label{eq:symmetric dissipation}
  \Dom_\symdiss(F) :=
   \Big\{
    \rho\in\Dom(F): \H\big(\rho,\zeta+d\L(\rho,0)\big) = \H\big(\rho,-\zeta+d\L(\rho,0)\big) \text{ for all } (\rho,\zeta)\in T^*_{\Dom(F)}\W
   \Big\}.
\end{equation}
The following lemma states that the dissipation potential is indeed symmetric in $\Dom_\symdiss(F)$.
%
\begin{lemma}[{\cite[Prop.~2.1(ii)]{MielkePeletierRenger14}}] \label{lem:symmetric dissipation}
Let $\L$ be an L-function on $\Z$. For $\rho\in\Dom_\symdiss(F)$ the following statements are equivalent
\begin{enumerate}[label=(\roman*)]
\item\label{it:symmetric dissipation H} $\H\big(\rho,\zeta-F(\rho)\big) = \H\big(\rho,-\zeta-F(\rho)\big)$ for all $\zeta\in T^*_\rho\W$,
\item $\L(\rho,j) = \L(\rho,-j) -2\langle F(\rho),j\rangle$ for all $j\in T_\rho\W$,
\item $\Phi^*(\rho,\zeta)=\Phi^*(\rho,-\zeta)$ for all $\zeta\in T^*_\rho\W$,
\item $\Phi(\rho,j)=\Phi(\rho,-j)$ for all $j\in T_\rho\W$.
\end{enumerate}
\end{lemma}

\begingroup
\allowdisplaybreaks
\begin{example}\label{ex:IPFG-diss-pot} In practice the force~\eqref{eq:F} is more easily calculated via the equivalent statement $d\H(\rho,-F(\rho))=0$. Since $\xi=\frac12 \log\frac{d}{c}$ minimises $\xi\mapsto c (e^\xi-1)+d(e^{-\xi}-1)$, we find
\begin{equation*}
  F_{xy}(\rho)=\mfrac12\log\frac{\rho_x Q_{xy}}{\rho_y Q_{yx}}, \qquad \Dom(F)=\P_+(\X).
\end{equation*}
This definition of the driving force has been introduced in~\cite[Sec.\ 2.2]{KaiserJackZimmer2018}.
Using~\eqref{eq:Phis}, the dissipation potentials are given by
\begin{align}
  \Phi^*(\rho,\zeta)&=\sumsum_{(x,y)\in\fluxS} 2\sqrt{\rho_x Q_{xy}\rho_y Q_{yx}}\big(\cosh(\zeta_{xy})-1\big),
  \label{eq:indep Phi*}\\
  \Phi(\rho,j)&=\sumsum_{(x,y)\in\fluxS} 2\sqrt{\rho_x Q_{xy}\rho_y Q_{yx}}\Big(\cosh^*\big(\tfrac{j_{xy}}{2\sqrt{\rho_x Q_{xy}\rho_y Q_{yx}}}\big)+1\Big).\notag
\end{align}
These dissipation potentials are indeed symmetric (since $\cosh$ is even), and therefore $\Dom_\symdiss(F)=\Dom(F)$. Note that, while a priori $\Phi$ and $\Phi^*$ are only defined for strictly positive probability measures, they can easily be extended to the full space $\Z=\P(\X)$. For instance, the observation that $\lim_{a\rightarrow 0} a \cosh^*(\tfrac xa)=0$ if $x=0$ and $+\infty$ otherwise, offers a trivial extension of $\Phi$ to $\Z$, which also reflects the idea ``vanishing jump rates guarantee vanishing fluxes''.

We note that the Hamiltonian corresponding to one-way fluxes is given by
\begin{equation*}
  \H^{\text{one-way}}(\rho,\zeta):=\sumsum_{\substack{x,y\in\X\times\X\\ x\neq y}} \rho_x Q_{xy}(e^{\zeta_{xy}}-1),
\end{equation*}
for which the corresponding driving force does not exist at all, i.e.\ $\Dom(F^\mathrm{one\text{-}way})=\emptyset$ (also see~\cite[Rem.~4.10]{Renger2018a}). Hence one can only construct a meaningful macroscopic fluctuation theory for net fluxes. This further justifies the net-flux approach used in this paper, as opposed to the one-way fluxes typically used for Markov jump processes. 
\end{example}
\endgroup

\begin{remark} In the IPFG example above and all the examples considered in Section~\ref{sec:examples}, $\Dom_\symdiss(F)=\Dom(F)$, i.e.\ the dissipation potential is symmetric. However, in general $\Dom_\symdiss(F)$ may be an (empty) subset of $\Dom(F)$ as the following construction shows. Consider $\Z=\W=\RR$ and $\phi=\mathrm{id}$. Let $\H(\rho,\zeta)=-\zeta+e^\zeta-1$, which corresponds to a real-valued Markov process  with generator $(\Q\super{n}f)(\rho,w):= -\partial_\rho f(\rho,w)-\partial_w f(\rho,w) + n(f(\rho+\tfrac1n,w+\tfrac1n)-f(\rho,w))$.
Then $F\equiv0$ and clearly $\H(\rho,-\zeta-F(\rho))\neq\H(\rho,\zeta-F(\rho))$, which implies that $\Dom_\symdiss(F)=\emptyset$.
\end{remark}

So far we have dealt with L-functions on $\Z$. Using~\eqref{eq:Phis1}, we now introduce L-functions defined on subsets of $\Z$.  For a given $\L$ and an appropriate cotangent field $\CoF(\rho)$, using~\eqref{eq:Phis1} we can define a ($\CoF$-tilted) L-function $\L_\CoF$ defined on a subset of $\Z$. We call this a  `tilted' L-function since its definition is motivated by tilted Markov processes (see Section~\ref{sec:tilt}). Although, technically $\CoF$ is a cotangent field, in this paper we will often refer to it as a force field due to physical considerations. 

\begin{definition}\label{def:tilt}
Let $\L$ be an L-function on $\Z$. For any $\CoF:\Dom(\CoF)\to T_{\Dom(\CoF)}^*\W$ with  $\Dom(\CoF)\subseteq\Z$, the tilted function $\H_\CoF:T_{\Dom(F)\cap\Dom(\CoF)}^*\W\to\RR\cup\{\infty\}$ is defined as
\begin{equation}
  \H_\CoF(\rho,\zeta):=\H\big(\rho,\zeta+\CoF(\rho)-F(\rho)\big)-\H\big(\rho,\CoF(\rho)-F(\rho)\big),
\label{eq:tilted H}
\end{equation}
and  $\L_\CoF:T_{\Dom(F)\cap\Dom(\CoF)}\W\to\RR\cup\{\infty\}$ denotes its convex dual in the second variable.
\end{definition}
%
\begin{lemma} \label{lem:preFIR}
Let $\L$ be an L-function on $\Z$. The tilted function $\L_\CoF$ is an L-function on $\Dom(F)\cap\Dom(\CoF)$, and satisfies the decomposition 
\begin{align}
  \L_\CoF(\rho,j)&=\L(\rho,j) + \H\big(\rho,\CoF(\rho)-F(\rho)\big) + \langle F(\rho)-\CoF(\rho),j\rangle \label{eq:LV}\\
              &=\Phi(\rho,j) + \Phi^*\big(\rho,\CoF(\rho)\big) - \langle \CoF(\rho), j\rangle. \notag
\end{align}
\end{lemma}
The two equalities follow by using convex duality and~\eqref{eq:L=Phi Phis1},~\eqref{eq:Phis1} with $\hat\zeta=F$. 
For special choices of $\CoF(\rho)$ we obtain
\begin{equation}\label{eq:L_LF}
  \L_F(\rho,j)=\L(\rho,j) \quad \text{and} \quad
  \L_0(\rho,j)=\Phi(\rho,j).
\end{equation}
\begingroup
\allowdisplaybreaks
\begin{example}
For any force field $\CoF(\rho)\in\RR^{\fluxS}$ we have
\begin{align*}
  \L_\CoF(\rho,j) &= \inf_{j^+\in\RR^{\fluxS}} \sumsum_{(x,y)\in\fluxS}
      s\big(j^+_{xy}\mid \sqrt{\rho_x Q_{xy}\rho_y Q_{yx}}e^{\CoF_{xy}(\rho)}\big) 
   + s\big(j^+_{xy}- j_{xy} \mid \sqrt{\rho_x Q_{xy}\rho_y Q_{yx}}e^{-\CoF_{xy}(\rho)}\big),\\
  \H_\CoF(\rho,\zeta)&=\sumsum_{(x,y)\in\fluxS} \sqrt{\rho_x Q_{xy}\rho_y Q_{yx}}\Big\lbrack e^{\CoF_{xy}(\rho)}(e^{\zeta_{xy}}-1) + e^{-\CoF_{xy}(\rho)}(e^{-\zeta_{xy}}-1)\Big\rbrack.
\end{align*}
\end{example}
\endgroup

We now define the notion of \emph{generalised Fisher information} which was introduced in Section~\ref{subsec:results}. 
\begin{definition}\label{def:genFI}
Let $\L$ be an L-function on $\Z$. For any $\rho\in \Z$, $\zeta\in T_\rho^*\W$, and $\lambda\in [0,1]$, the generalised Fisher information is 
\begin{equation*}
  \RF^\lambda_\zeta(\rho)=-\H(\rho,-2\lambda \zeta). 
\end{equation*}
\end{definition}
As discussed in Section~\ref{subsec:results}, it is important to choose $\lambda$ and $\zeta$ such that $\RF^\lambda_\zeta$ is non-negative, as this guarantees that the corresponding powers  are non-negative along the zero-cost flux. The following result explores the set of force fields for which this is true (also see Figure~\ref{fig:contour plot}).
\begin{prop}
\label{prop:preFIR}
Let $\L$ be an L-function on $\Z$. For any $\rho\in\Z$ we have
\begin{enumerate}[label=(\roman*)]
\item
The set $\{\zeta\in T_\rho^*\W:  \RF^{\frac12}_\zeta(\rho)\geq0\}$ is convex and includes $\zeta=0$.
\item
In particular, if $\zeta\in T_\rho^*\W$ such that 
\begin{equation}\label{eq:feasible FIR force}
 \RF^{\frac12}_\zeta(\rho) \geq 0,
\end{equation}
then for any $\lambda\in\lbrack0,1\rbrack$
\begin{equation}\label{eq:absFIR}
   \RF^{\lambda}_{\frac12\zeta}(\rho)\geq 0.
\end{equation}
\item For any $\zeta\in T_\rho^*\W$ we have
\begin{equation}\label{eq:Fisher entropy prod}
  \lim\limits_{\lambda \downarrow 0} \tfrac{1}{\lambda} \RF^{\lambda}_\zeta(\rho) = 2\langle \zeta,j^0(\rho) \rangle.
\end{equation}
where $j^0$ is the zero-cost flux for $\L$ (see Definition~\ref{def:L-function}).
\end{enumerate}
\end{prop}

\begin{proof}
\emph{(i)} Since $\L$ is an L-function, $\H(\rho,\cdot)$ is convex with $\H(\rho,0)=0$ and the assertion follows.\\
\emph{(ii)} Using convexity, $-\RF^{\lambda}_{\frac12\zeta}(\rho)=\H(\rho,-\lambda\zeta)=\H(\rho,-\lambda\zeta +(1-\lambda)0)\leq \lambda\H(\rho,-\zeta) + (1-\lambda)\H(\rho,0)\leq0$.\\
\emph{(iii)} By definition of L-functions, $\L(\rho,\cdot)$ has unique minimiser $j^0(\rho)$, which is equivalent to $\partial\H(\rho,0)=\{j^0(\rho)\}=\{d\H(\rho,0)\}$. The claim then follows from the definition of the Gateaux derivative.
\end{proof}
Note that~\cite[Thm.~1.7]{HilderPeletierSharmaTse20} is a special case of this result for the IPFG example. Following~\cite{HilderPeletierSharmaTse20}, we call $\RF^\lambda$ the generalised Fisher information since it generalises the classical notion of Fisher information as the dissipation rate of free energy along the solutions of the zero-cost flux of the L-function. This property follows by using~\eqref{eq:Fisher entropy prod} with appropriate choices for $\zeta$.  In the next section we construct $\zeta$ for which $ \RF^{\frac12}_\zeta(\rho)= 0$ and the above result can be applied.

\subsection{Reversed L-function, symmetric and antisymmetric forces}\label{subsec:abstract time reversal}
Inspired by the notion of time-reversibility in MFT we now introduce the reversed L-function which will then be used to define symmetric and antisymmetric forces.
From now on we assume that $\QP$ is a quasipotential associated to $\L$ in the sense of Definition~\ref{def:I0}.

\begin{definition}\label{def:time-reversedL} Let $\L$ be an L-function on $\Z$. For any $\rho\in\Z$ where $\QP$ is Gateaux differentiable and any $j\in T_\rho\W$, we define the \emph{reversed L-function} as 
\begin{equation*}\label{def:time-reversed L}
  \overleftarrow\L(\rho,j):= \L(\rho,-j) + \langle d\phi_\rho\tp d\QP(\rho),j\rangle.
\end{equation*}
\end{definition}

This notion of the reversed L-function is motivated by the large-deviations of time-reversed Markov processes (see Section~\ref{sec:time-rev} for details). Note that we use the name reversed L-function as opposed to time-reversed L-function since there is no time variable in this abstract framework.

The following result states that $\overleftarrow\L$ is indeed an L-function, and discusses the driving force and  dissipation potential associated to it.
\begin{prop} Let $\L$ be an L-function on $\Z$. For any $\rho\in\Z$ where $\QP$ is Gateaux differentiable we have
\begin{enumerate}[label=(\roman*)]
\item The convex dual of $\overleftarrow\L\!(\rho,\cdot)$ is $\overleftarrow\H\!(\rho,\zeta)=\H\big(\rho,d\phi_\rho\tp d\QP(\rho)-\zeta\big)$.
\item If $\overleftarrow\jmath^0(\rho)$ is the zero-cost flux in the sense that $\overleftarrow\L\!\big(\rho,\overleftarrow\jmath^0(\rho)\big)=0$, then $-\overleftarrow\jmath^0(\rho)\in\partial\H\big(\rho,d\phi_\rho \tp d\QP(\rho)\big)$, and it is unique if $\H(\rho,\cdot)$ is Gateaux differentiable at $d\phi_\rho \tp d\QP(\rho)$.  
Furthermore $\overleftarrow\L$ is an L-function on $\{\rho\in\Z:\QP \text{ is Gateaux differentiable in }\rho\}$ and $\QP$ is a  quasipotential corresponding to $\overleftarrow\L$.

\item Additionally, if $\rho\in\Dom(F)$ (recall Definition~\ref{def:F-disspot}), then the driving force and dissipation potentials corresponding to $\overleftarrow\L$ are given by
\begin{equation*}
  \overleftarrow F\!(\rho)= -F(\rho) - d\phi_\rho\tp d\QP(\rho), \quad
  \overleftarrow\Phi\!(\rho,j)=\Phi(\rho,-j), \quad
  \overleftarrow\Phi\!^*(\rho,\zeta)=\Phi^*(\rho,-\zeta).
\end{equation*}
\end{enumerate}
\end{prop}
\begin{proof}
\emph{(i)} Follows by a straightforward calculation of the convex dual.\\
\emph{(ii)} Using the Fermat's rule $0\in \partial \overleftarrow\L\!(\rho,\overleftarrow\jmath^0(\rho))$, and therefore $\overleftarrow\jmath^0(\rho)\in\partial\overleftarrow\H\!(\rho,0)$. Using Definition~\ref{def:time-reversed L} and since $\L$ is an L-function, $\overleftarrow\L$ is convex, lower semicontinuous and using~\eqref{eq:invariant I_0} satisfies $\inf \overleftarrow\L(\rho,\cdot)=0$. Consequently $\overleftarrow\L$ is an L-function on $\Dom(\Fsy)$ (see~\eqref{def:Domain-Fsym} below) and  $\QP$ is a quasipotential associated to $\overleftarrow\L$.\\
\emph{(iii)} Using \eqref{eq:F} we find
\begin{equation*}
  -\overleftarrow F\!(\rho):=d\overleftarrow\L\!(\rho,0)=-d\L(\rho,0)+d\phi_\rho\tp d\QP(\rho)=F(\rho)+d\phi_\rho\tp d\QP(\rho)
  \end{equation*}
and using~\eqref{eq:Phis} we find  
\begin{align*}
  \overleftarrow\Phi\!^*(\rho,\zeta)&:=\overleftarrow\H\!\big(\rho,\zeta-\overleftarrow F\!(\rho)\big) - \overleftarrow\H\!\big(\rho,-\overleftarrow F\!(\rho)\big) 
    =\H\big(\rho,d\phi_\rho\tp d\QP(\rho)+\overleftarrow F\!(\rho)-\zeta \big) - \H\big(\rho,d\phi_\rho\tp d\QP(\rho)+\overleftarrow F\!(\rho) \big)\\
    &=\H\big(\rho,-F(\rho)-\zeta \big) - \H\big(\rho,-F(\rho) \big) = \Phi^*(\rho,-\zeta).
\end{align*}
Consequently $\overleftarrow\Phi\!(\rho,j)=\Phi(\rho,-j)$.
\end{proof}

Motivated by this result, we decompose the driving force $F$ (recall~\eqref{eq:F})  into a symmetric and antisymmetric part with respect to the reversal, i.e.\ $F^\sym=\frac12(F+\overleftarrow F)$ and $F^\asym=\frac12(F-\overleftarrow F)$. The following result summarises these ideas.
\begin{cor} \label{def:sym-asym-force}
Let $\L$ be an L-function on $\Z$. Define
\begin{align}\label{def:Domain-Fsym}
  \Dom(F^\sym):=\{\rho\in\Z:\QP \text{ is Gateaux differentiable at } \rho\},
  &&\text{and}&&
  \Dom(F^\asym):= \Dom(F)\cap\Dom(F^\sym),
\end{align}
and
\begin{align}
  F^\sym(\rho)&:=-\tfrac12 d\phi_\rho\tp d\QP(\rho)           &&\text{for } \rho\in\Dom(F^\sym),\label{def:Dom-Fsym}\\
  F^\asym(\rho)&:= F(\rho)+\tfrac12 d\phi_\rho\tp d\QP(\rho)  &&\text{for } \rho\in\Dom(F^\sym). \nonumber
\end{align}
Then for any $\rho\in\Dom(F^\asym)$,
\begin{equation}\label{eq:F-revF-decom}
  F(\rho)= F^\sym(\rho)+F^\asym(\rho), \quad \text{and}\quad
  \overleftarrow F(\rho)=F^\sym(\rho) - F^\asym(\rho).
\end{equation}
\end{cor}
Note that while we make use of the reversed L-function to construct the symmetric and antisymmetric force, it does not explicitly appear in their definition. 
In the case of zero antisymmetric force, i.e.\ $\Fasy(\rho)=0$, the driving forces satisfy $F(\rho)=\overleftarrow F(\rho)=F^\sym(\rho)$, which is the setting of dissipative systems (see Section~\ref{subsec:flows}).

\begingroup
\allowdisplaybreaks
\begin{example} We have
\begin{align*}
  \overleftarrow\H\!(\rho,\zeta) &= \sumsum_{(x,y)\in\fluxS} \rho_x \mfrac{\pi_y}{\pi_x}Q_{yx} (e^{ \zeta_{xy}}-1) + \rho_y \mfrac{\pi_x}{\pi_y}Q_{xy} (e^{-\zeta_{xy}}-1),\\
  \overleftarrow\L\!(\rho,j) &= \inf_{j^+\in\RR^{\fluxS}_{\geq0}}
    \sumsum_{(x,y)\in\fluxS} 
    s\big(j^+_{xy}\mid \rho_x \mfrac{\pi_y}{\pi_x} Q_{yx}\big) +  s\big(j^+_{xy}-j_{xy}\mid \rho_y \mfrac{\pi_x}{\pi_y} Q_{xy}\big),\\
  \overleftarrow F\!_{xy}(\rho) &= \mfrac12\log\frac{\rho_x\tfrac{\pi_y}{\pi_x} Q_{yx}}{\rho_y \tfrac{\pi_x}{\pi_y}Q_{xy}}.
\end{align*}
\sloppy{The expression $\frac{\pi_{x}}{\pi_{y}} Q_{xy}$ is the generator matrix for a single time-reversed jump process~\cite[Thm.~3.7.1]{Norris1998}.
Again, beware that a priori $\overleftarrow{\H}$ and $\overleftarrow{\L}$ are only defined on $\Z = \Dom(F)$, but can be continuously extended to $\P(\X)$ in a straightforward manner.}

The symmetric and antisymmetric (with respect to the reversal) components of the driving force are (also see~\cite{KaiserJackZimmer2018})
\begin{equation}\label{IPFG:symForce}
  F^\sym_{xy}(\rho)=\mfrac12\log\frac{\pi_y\rho_x}{\pi_x\rho_y}
\quad\text{ and } \quad
  F^\asym_{xy}(\rho)=\mfrac12\log\frac{\pi_xQ_{xy}}{\pi_y Q_{yx}},
\end{equation}
with $\Dom(F)=\Dom(\Fsy)=\Dom(\Fasy)=\P_+(\X)$.
Note that for reversible Markov chains, i.e.\ those satisfying \emph{detailed balance},  $F^\asym=0$. \end{example}
\endgroup

Recall the generalised Fisher information $\RF^{\lambda}_\zeta$ from Definition~\ref{def:genFI}, and that we are looking for force fields that make this quantity non-negative. The following result shows that $\RF^{\frac12}_\zeta(\rho)=0$ for $\zeta=2F(\rho),2F^\sym(\rho)$, $2F^\asym(\rho)$.  This will be crucial to derive the key decompositions of $\L$ in Section~\ref{sec:abstract FIIR equalities}.

In this result we make use of (analogous to \eqref{eq:symmetric dissipation}),
\begin{equation}
  \Dom_\symdiss(F^\asym) :=
   \Big\{
    \rho\in\Dom(F^\asym):  \ \H\big(\rho,\zeta+d\L(\rho,0)\big) = \H\big(\rho,-\zeta+d\L(\rho,0)\big), \ \forall \zeta\in T^*_{\rho}\W
   \Big\}.
\label{eq:symmDom}
\end{equation}
Note that $\Dom_\symdiss(F^\asym)\subseteq \Dom_\symdiss(F)$ since $\Dom(\Fasy)\subseteq\Dom F$.

%

\begin{lemma}\label{lem:zero H forces}
Let $\L$ be an L-function on $\Z$. We have
\begin{enumerate}[label=(\roman*)]
\item  $\forall \rho\in\Dom(F): \  \RF^{\frac12}_F(\rho) \geq 0$ and $\forall \rho\in\Dom_{\symdiss}(F): \  \RF^{\frac12}_{2F}(\rho)= 0$,
\item $\forall \rho\in\Dom(F^\sym): \   \RF^{\frac12}_{2\Fsy}(\rho) =0$,
\item $\forall \rho\in\Dom_\symdiss(F^\asym): \   \RF^{\frac12}_{2\Fasy}(\rho) =0$.
\end{enumerate}
\end{lemma}

\begin{proof} \emph{(i)} 
Since $-F$ minimises $\H$, it follows that 
$\H(\rho,-F)=\inf \H(\rho,\cdot)\leq 
\H(\rho,0)=-\inf\L(\rho,\cdot) = 0$, and therefore $\RF^{\frac12}_F(\rho) = -\H(\rho,-F)\geq 0$. If the dissipation potential is symmetric, the choice $\zeta=-F(\rho)$ in Lemma~\ref{lem:symmetric dissipation}\ref{it:symmetric dissipation H} gives $\RF^{\frac12}_{2F}(\rho) = \H\big(\rho,-2F(\rho)\big)=\H(\rho,0)=0$.\\
\emph{(ii)} The claim follows since~\eqref{eq:invariant I_0} holds for all $\rho\in\Dom(F^\sym)$. \\
\emph{(iii)} With $\zeta=\overleftarrow F(\rho)=F^\sym(\rho)-F^\asym(\rho)$ in Lemma~\ref{lem:symmetric dissipation}\ref{it:symmetric dissipation H} we find
$\H\big(\rho,-2F^\asym(\rho)\big) = \H\big(\rho,-2F^\sym(\rho)\big)=0$.
\end{proof}

Figure~\ref{fig:contour plot} is a schematic diagram of force fields $\zeta$ for which $\RF^\lambda_\zeta$ is non-negative. Note that, while there are various possibilities for such $\zeta$, we focus on $\zeta=2F(\rho),2F^\sym(\rho),2F^\asym(\rho)$ since they correspond to the physically relevant powers defined in~\eqref{eq:intro sym power} and~\eqref{eq:intro asym power}.

\begin{figure}[ht]
\centering
\begin{tikzpicture}
\tikzstyle{every node}=[font=\scriptsize];

  \draw[fill=lightgray,rotate around={-15:(0,0)}] (-2.285,-1.283) ellipse [x radius=3, y radius=2]; 


  \draw[rotate around={-15:(0,0)}] (0,0) -- (-4.57,-2.58);
  \filldraw[rotate around={-15:(0,0)}] (-2.285,-1.29) circle (0.05) node[anchor=north]{$F(\rho)$}; 
  \filldraw[rotate around={-15:(0,0)}] (-4.57,-2.58) circle (0.05) node[anchor=east]{$2F(\rho)$};
  
  \draw[rotate around={-15:(0,0)}] (0,0) -- (-4.57,0);
  \filldraw[rotate around={-15:(0,0)}] (-2.285,0) circle (0.05) node[anchor=north east]{$\Fsy(\rho)$};
  \filldraw[rotate around={-15:(0,0)}] (-4.57,0) circle (0.05) node[anchor=east]{$2\Fsy(\rho)$};
  
  \draw[rotate around={-15:(0,0)}] (0,0) -- (0,-2.58);
  \filldraw[rotate around={-15:(0,0)}] (0,-1.29) circle (0.05) node[anchor=east]{$\Fasy(\rho)$};
  \filldraw[rotate around={-15:(0,0)}] (0,-2.58) circle (0.05) node[anchor=east]{$2\Fasy(\rho)$};

  \filldraw(0,0) circle (0.05) node[anchor=south west]{$0$};

  \draw[->](-7,0)--(2,0);
  \draw[->](0,-3.3)--(0,2);

  \draw[dotted, rotate around={-15:(-2.54,-0.65)}](-2.54,-0.65) ellipse [x radius=3, y radius=2, scale=0.58];
  \draw[dotted, rotate around={-15:(-2.54,-0.65)}] (-2.54,-0.65) ellipse [x radius=3, y radius=2, scale=0.82];
  \draw[dotted, rotate around={-15:(-2.54,-0.65)}] (-2.54,-0.65) ellipse [x radius=3, y radius=2, scale=1.16];
  \draw[dotted, rotate around={-15:(-2.54,-0.65))}] (-2.54,-0.65) ellipse [x radius=3, y radius=2, scale=1.3];
\end{tikzpicture}
\caption{Contour lines of a possible concave function $\zeta\mapsto\RF^{\frac12}_{\zeta}(\rho)$ for a fixed $\rho$, where the superlevel set $\{\zeta\in T_\rho^*\W:\RF^{\frac12}_{\zeta}(\rho)\geq0\}$ is depicted in gray. By Definitions~\ref{def:F-disspot} and \ref{def:genFI}, $F(\rho)$ is a maximiser for $\zeta\mapsto \RF^{\frac12}_{\zeta}(\rho)$, and assuming $\rho\in\Dom_\symdiss(\Fasy)$, Lemma~\ref{lem:zero H forces} says that $2F(\rho)$, $2\Fsy(\rho)$ and $2\Fasy(\rho)$ all lie on the $0$-contour line. By the convexity of the superlevel set $\{\RF^{\frac12}_\zeta(\rho)\geq0\}$ (see Proposition~\ref{prop:preFIR}), any convex combination $\zeta$ between $0$ and $2F(\rho)$, $2\Fsy(\rho)$ or $2\Fasy(\rho)$, drawn by the three lines, yield non-negative $\RF^{\frac12}_\zeta(\rho)\geq0$.
}
\label{fig:contour plot}
\end{figure}

\begin{remark}\label{rem:rev-HL-relations}
For all $\rho\in \Dom(F^\asym)$, we can write the reversed function as a tilting in the sense of~\eqref{eq:tilted H}
\begin{equation*}
  \overleftarrow\H(\rho,\zeta)=\H_{-\overleftarrow{F}}(\rho,-\zeta).
\end{equation*}
Using~\eqref{eq:LV}, the corresponding reversed L-function then satisfies
\begin{equation*}
  \overleftarrow\L(\rho,j) = \L_{-\overleftarrow F}(\rho,-j) =\L(\rho,-j) + \H\big(\rho,d\phi\tp_\rho d\QP(\rho)\big) - \langle d\phi\tp_\rho d\QP(\rho),j\rangle
              =\Phi(\rho,-j) + \Phi^*\big(\rho,-\overleftarrow{F}\big) - \langle \overleftarrow{F}, j\rangle, 
\end{equation*}
where we have used $F+\overleftarrow F = - d\phi_\rho\tp d\QP(\rho)$.
\end{remark}

\subsection{Generalised orthogonality}\label{sec:Gen-ortho}

Before we continue with deriving the main decompositions~\eqref{eq:intro-L-decom} of the L-function, we elaborate further on the decomposition of the driving force $F$ into the symmetric force $\Fsy$ and antisymmetric force $\Fasy$, and investigate the natural question whether these forces are orthogonal in some sense. It turns out that they are indeed orthogonal in a generalised sense, and using this notion of orthogonality we can already derive decompositions~\eqref{eq:intro-L-decom} for $\lambda=\frac12$.
As discussed in the introduction, in MFT  the dissipation potentials are often squares of appropriate Hilbert norms $\|\cdot\|_{\rho}$, and in that setting one can write
\begin{align*}
  \Phi^*\big(\rho,\zeta^1+\zeta^2\big)&:=\tfrac12\lVert \zeta^1+\zeta^2\rVert_{\rho}^2 =\tfrac12\lVert \zeta^1\rVert_{\rho}^2 + \langle \zeta^1,\zeta^2\rangle_{\rho} + \tfrac12\lVert \zeta^2\rVert_{\rho}^2\\
&= \Phi^*\big(\rho,\zeta^1\big) + \langle \zeta^1,\zeta^2\rangle_{\rho} + \Phi^*\big(\rho,\zeta^2\big),
\end{align*}
where $\langle \cdot,\cdot\rangle_{\rho}$ is the inner product induced by the norm. Typically $\Fsy$ and $\Fasy$ are orthogonal in the sense that $\langle \Fsy,\Fasy\rangle_{\rho}=0$. We reiterate these ideas in Section~\ref{sec:LatticeGas} which deals with the classical MFT setting of lattice gases. However this orthogonality relation is specific to the quadratic setting. A generalised  notion of orthogonality was introduced in~\cite{KaiserJackZimmer2018} for non-quadratic dissipation potential~\eqref{eq:indep Phi*} corresponding to independent Markov chains which have $\cosh$-type structure (see Example~\ref{ex:IPFG-diss-pot}) and this principle was further generalised to chemical reaction networks in~\cite{RengerZimmer2021} (see Section~\ref{subsec:reacting particle system} for details). Based on these results, we now provide a notion of generalised orthogonality which applies to arbitrary dissipation potentials arising within the abstract framework of this section (and does not require any specific structure).

\begin{definition}\label{def:genOrth}
 For any $\rho\in\Dom(F)$ and $\zeta^2 \in T^*_{\rho}\W$, define the \emph{modified dissipation potential} $\Phi^*_{\zeta^2}:T^*_{\rho}\W\rightarrow\mathbb R\cup\{\infty\}$ and the \emph{generalised orthogonality pairing} $\theta_\rho:T^*_{\rho}\W\times T^*_{\rho}\W\rightarrow \mathbb R\cup\{\infty\}$ as
\begin{align*}
  \Phi^*_{\zeta^2}(\rho,\zeta^1) &:= \tfrac12\left[\H\big(\rho,\zeta^1+\zeta^2-F(\rho)\big)+\H\big(\rho,-\zeta^1+\zeta^2-F(\rho)\big)\right]-\H\big(\rho,\zeta^2-F(\rho)\big),\\
    & \ = \tfrac12 \left[\Phi^*(\rho,\zeta^1+\zeta^2) + \Phi^*(\rho,-\zeta^1+\zeta^2)\right] -\Phi^*(\rho,\zeta^2), \\
  \theta_\rho(\zeta^1,\zeta^2) &:= \tfrac12 \left[\H\big(\rho,\zeta^1+\zeta^2-F(\rho)\big) -\H\big(\rho,-\zeta^1+\zeta^2-F(\rho)\big)\right]\\
    & \ = \tfrac12\left[\Phi^*(\rho,\zeta^1+\zeta^2)-\Phi^*(\rho,-\zeta^1+\zeta^2)\right],
\end{align*}
where we have used~\eqref{eq:Phis} to arrive at the  equalities.  
\end{definition}


The following result collects the properties of $\Phi_{\zeta^2}$ and $\theta_\rho$ clarifying the notion of orthogonality in the abstract framework. Recall the definition of $\Dom_\symdiss(F^\asym)$ from \eqref{eq:symmDom}.
\begin{prop}\label{prop:ortho-rel}
Let $\L$ be an L-function on $\Z$. For any $\rho\in\Dom(F)$, $\Phi_{\zeta^2}^*(\rho,\cdot)$ is convex, lower semicontinuous and $\inf \Phi^*_{\zeta^2}(\rho,\cdot)=0=\Phi_{\zeta^2}^*(\rho,0)$. Furthermore, for any $\zeta^1,\zeta^2\in T_\rho^*\W$, the dissipation potential $\Phi^*$ admits the decomposition
\begin{equation*}
  \Phi^*(\rho,\zeta^1+\zeta^2) = \Phi^*(\rho,\zeta^1) + \theta_\rho(\zeta^2,\zeta^1) + \Phi^*_{\zeta^1}(\rho,\zeta^2) = \Phi^*(\rho,\zeta^2) + \theta_\rho(\zeta^1,\zeta^2) + \Phi^*_{\zeta^2}(\rho,\zeta^1).
\end{equation*}
Moreover the generalised orthogonality pairing satisfies
\begin{align*}
  \theta_\rho\big(F^\sym(\rho),F^\asym(\rho)\big) &=0   &&\text{for all }\rho\in\Dom(F^\asym),\\
  \theta_\rho\big(F^\asym(\rho),F^\sym(\rho)\big) &=0   &&\text{for all }\rho\in \Dom_\symdiss(F^\asym),
\end{align*}
and therefore we have 
\begin{equation}\label{eq:orth split Fsym Fasym}
\begin{aligned}
  \Phi^*\big(\rho,F(\rho)\big) &= \Phi^*\big(\rho,F^\asym(\rho)\big) + \Phi^*_{F^\asym(\rho)}\big(F^\sym(\rho)\big)   &&\text{for all }\rho\in\Dom(F^\asym),\\
  \Phi^*\big(\rho,F(\rho)\big) &= \Phi^*\big(\rho,F^\sym(\rho)\big) + \Phi^*_{F^\sym(\rho)}\big(F^\asym(\rho)\big)    &&\text{for all }\rho\in \Dom_\symdiss(F^\asym).
\end{aligned}
\end{equation}
\end{prop}
\begin{proof}
The convexity, lower semicontinuity of $\Phi_{\zeta^2}^*$ follows from the convexity, lower semicontinuity of $\Phi^*$ and $\Phi_{\zeta^2}^*(\rho,0)=0$ follows from the definition. Using convexity of $\Phi^*$ we find
\begin{equation*}
\Phi^*_{\zeta^2}(\rho,\zeta^1)\geq \Phi^*\left(\rho,\tfrac12(\zeta^1+\zeta^2)+\tfrac12(-\zeta^1+\zeta^2)\right)-\Phi^*(\rho,\zeta^2)=0,
\end{equation*}
and therefore $\inf \Phi^*_{\zeta^2}(\rho,\cdot)=0$.
 The two decompositions follow immediately by adding $\Phi^*_{\zeta^2}$ and $\theta_\rho$. Using Lemma~\ref{lem:zero H forces} we find
\begin{align*}
  2\theta_\rho\big(F^\sym(\rho),F^\asym(\rho)\big)&=\H\big(\rho,F^\sym(\rho)+F^\asym(\rho)-F(\rho)\big) - \H\big(\rho,-F^\sym(\rho)+F^\asym(\rho)-F(\rho)\big)\\
    &=\H(\rho,0) - \H\big(\rho-2F^\sym(\rho)\big)=0,\\
  2\theta_\rho\big(F^\asym(\rho),F^\sym(\rho)\big)&=\H\big(\rho,F^\sym(\rho)+F^\asym(\rho)-F(\rho)\big) - \H\big(\rho,F^\sym(\rho)-F^\asym(\rho)-F(\rho)\big)\\
    &=\H(\rho,0) - \H\big(\rho-2F^\asym(\rho)\big)=0.
\end{align*}
where the second decomposition additionally requires that $\rho\in \Dom_\symdiss(F^\asym)$. 
\end{proof}

%
From the general decomposition~\eqref{eq:L-decom-F} and the generalised orthogonality result above, we can already provide two distinct decompositions of $\L$, as derived in \cite[Cor.~4.3]{RengerZimmer2021} for the case of chemical reactions.
\begin{corollary}\label{cor:decomps from orth}
Let $\L$ be an L-function on $\Z$. Then for all $(\rho,j)\in T_{\Dom(\Fasy)}\W$,
\begin{align*}
  \L(\rho,j)&=\Phi(\rho,j) + \Phi^*\big(\rho,\Fasy(\rho)\big) - \langle\Fasy(\rho),j\rangle + \Phi^*_{\Fasy}\big(\rho,\Fsy(\rho)\big) - \langle\Fsy(\rho),j\rangle,
\intertext{and for all $(\rho,j)\in T_{\Dom_\symdiss(\Fasy)}\W$,}
  \L(\rho,j)&=\Phi(\rho,j) + \Phi^*\big(\rho,\Fsy(\rho)\big) - \langle\Fsy(\rho),j\rangle + \Phi^*_{\Fsy}\big(\rho,\Fasy(\rho)\big) - \langle\Fasy(\rho),j\rangle.
\end{align*}
\end{corollary}
In both decompositions, we may interpret the first three terms as an L-function with a modified force, the fourth term as a Fisher information, and the last term as a power (see Remark~\ref{rem:Fisher-Mod-DissPot} for details). 

\begin{example} Using Definition~\ref{def:genOrth} we have (see also~\cite{KaiserJackZimmer2018}) 
\begin{align*}
  \Phi^*_{\zeta^2}(\rho,\zeta^1) &= 2\sumsum_{(x,y)\in\fluxS} \sqrt{\rho_xQ_{xy}\rho_yQ_{yx}}\cosh(\zeta^2_{xy})\big(\cosh(\zeta^1_{xy})-1\big),\\
  \theta_\rho(\zeta^1,\zeta^2) &= 2\sumsum_{(x,y)\in\fluxS} \sqrt{\rho_xQ_{xy}\rho_yQ_{yx}}\sinh(\zeta^2_{xy})\sinh(\zeta^1_{xy}).
\end{align*}
\end{example}

\subsection{Decomposing the L-function}\label{sec:abstract FIIR equalities}
We now present decompositions of the L-function, which are the main results of the abstract theory presented so far. Using $\CoF=F,\Fsy,\Fasy$ in~\eqref{eq:LV} and encoding convex combinations via the parameter $\lambda$, we arrive at three distinct decompositions of $\L$; this corresponds to all the points on the three lines depicted in Figure~\ref{fig:contour plot}.

\begin{theorem}\label{th:FIIRs}
Let $\L$ be an L-function on $\Z$. It admits the following decompositions
\begin{enumerate}[label=(\roman*)]
\item For any $\rho\in\Dom_{\symdiss}(F)$, $j\in T_\rho\W$ and $\lambda\in[0,1]$,
\begin{equation}\label{eq:FIIR F}
  \L(\rho,j)=\L_{(1-2\lambda)F}(\rho,j) + \RF^\lambda_F(\rho)  - 2\lambda \langle F(\rho),j\rangle  \ \text{ with } \RF^\lambda_F(\rho)\geq0.
\end{equation}
\item For any $\rho\in\Dom(F^\asym)$, $j\in T_\rho\W$ and $\lambda\in[0,1]$,
\begin{equation}\label{eq:FIIR Fsym}
  \L(\rho,j)=\L_{F-2\lambda F^\sym}(\rho,j) + \RF^\lambda_{\Fsy}(\rho)  - 2\lambda \langle F^\sym(\rho),j\rangle 
            \  \text{ with }\RF^\lambda_{\Fsy}(\rho)\geq0.
\end{equation}
\item For any $\rho\in\Dom_\symdiss(F^\asym)$, $j\in T_\rho\W$ and $\lambda\in[0,1]$,
\begin{equation}\label{eq:FIIR Fasym}
  \L(\rho,j)=\L_{F-2\lambda F^\asym}(\rho,j) + \RF^\lambda_{\Fasy}(\rho) - 2\lambda \langle F^\asym(\rho),j\rangle 
         \  \text{with }\RF^\lambda_{\Fasy}(\rho)\geq0.
\end{equation}
\end{enumerate}
\end{theorem}
\begin{proof}
The decompositions follow directly from  Lemma~\ref{lem:preFIR}. The non-negativity of the Fisher informations follows from Proposition~\ref{prop:preFIR} and Lemma~\ref{lem:zero H forces}.
\end{proof}
\begin{remark}
The decomposition~\eqref{eq:FIIR F} holds for $\rho\in\Dom_\symdiss(F)$. Since by Lemma~\ref{lem:zero H forces}(i), $\RF_F^{\frac12}(\rho)\geq 0$ for any $\rho\in \Dom(F)$, we also have the following decomposition for any $\rho\in\Dom(F)$, $j\in T_\rho\W$ and $\lambda\in[0,\tfrac12]$
\begin{equation*}
  \L(\rho,j)=\L_{(1-\lambda)F}(\rho,j) + \RF^\lambda_{F}(\rho)  - \lambda \langle F(\rho),j\rangle  \ \text{ with } \RF^\lambda_{ F}(\rho)\geq0.
\end{equation*}
The non-negativity of $\RF^\lambda_{F}(\rho)$ follows by repeating the proof of Proposition~\ref{prop:preFIR}(ii)  for $\lambda\in[0,\tfrac12]$.
\end{remark}

The following result exhibits the significance of the choices $\lambda=\tfrac12,1$, and that the decompositions for other values can be seen as generalisations. 
\begin{cor}[$\lambda=\tfrac12,1$] 
With the choice $\lambda=\tfrac12$, the decompositions~\eqref{eq:FIIR F},~\eqref{eq:FIIR Fsym} and~\eqref{eq:FIIR Fasym}  respectively become  
\begin{align}
\L(\rho,j)&=\L_0(\rho,j) +\RF^\frac12_F(\rho) -  \langle F(\rho),j\rangle =\Phi(\rho,j) + \Phi^*\big(\rho,F(\rho)\big) - \langle F(\rho),j\rangle,\label{eq:FIIR F 1/2}\\
\L(\rho,j)&=\L_{F^\asym}(\rho,j) +\RF^\frac12_{\Fsy}(\rho) - \langle F^\sym(\rho),j\rangle\label{eq:FIIR Fsym 1/2},\\
\  \L(\rho,j)&=\L_{F^\sym}(\rho,j) +\RF^\frac12_{\Fasy}(\rho) - \langle F^\asym(\rho),j\rangle\label{eq:FIIR Fasym 1/2}.
\end{align}
With the choice $\lambda=1$, the decompositions~\eqref{eq:FIIR F},~\eqref{eq:FIIR Fsym} and~\eqref{eq:FIIR Fasym} respectively become  
\begin{align}
\L(\rho,j)&=\L_{-F}(\rho,j) -  2\langle F(\rho),j\rangle,\label{eq:FIIR-full-lam1}\\  
\L(\rho,j)&=\L_{-\overleftarrow F}(\rho,j) - 2\langle F^\sym(\rho),j\rangle =\overleftarrow\L(\rho,-j) - 2\langle F^\sym(\rho),j\rangle, \nonumber\\
\L(\rho,j)&=\L_{\overleftarrow F}(\rho,j) - 2\langle F^\asym(\rho),j\rangle, \nonumber
\end{align}
where $F,\overleftarrow F$ satisfy the relations~\eqref{eq:F-revF-decom}.  
\end{cor}
The second equality in~\eqref{eq:FIIR F 1/2} follows from~\eqref{eq:L_LF} and~\eqref{eq:Phis} where we use $\H(\rho,0)=0$ and the Fisher-information term vanishes by Lemma~\ref{lem:zero H forces}. 
A careful analysis of the zero-cost flux for $\L_{\Fsy}$ and $\L_{\Fasy}$ will be presented in Subsection~\ref{subsec:flows} and Section~\ref{sec:IPFG antisym flow}.

\begin{remark}\label{rem:Fisher-Mod-DissPot}
Using \eqref{eq:L-decom-F}, we see that \eqref{eq:FIIR Fsym 1/2} and \eqref{eq:FIIR Fasym 1/2} are the same decompositions as those in Corollary~\ref{cor:decomps from orth} which use generalised orthogonality, and that the two corresponding Fisher informations are in fact modified dissipation potentials (as introduced in Section~\ref{sec:Gen-ortho})
\begin{equation*}
 \RF^\frac12_{\Fsy}(\rho) =  \Phi^*_{F^\asym}\big(\rho,F^\sym(\rho)\big), \   \ 
 \RF^\frac12_{\Fasy}(\rho) =  \Phi^*_{F^\sym}\big(\rho,F^\asym(\rho)\big).
\end{equation*}
This also explains the non-negativity of these Fisher informations for $\lambda=\frac12$.
\end{remark}

\begingroup
\allowdisplaybreaks
\begin{example} \label{ex:IPFG-FIIR}
Decompositions~\eqref{eq:FIIR F},~\eqref{eq:FIIR Fsym} and~\eqref{eq:FIIR Fasym} hold with the tilted L-functions 
\begin{align*}
  \L_{(1-2\lambda)F}(\rho,j) &= \inf_{j^+\in \RR^{\fluxS}_{\geq0}} \sumsum_{(x,y)\in\fluxS} s\big(j^+_{xy} \mid (\rho_x Q_{xy})^{1-\lambda}(\rho_y Q_{yx})^\lambda\big) \\
&\hspace{3cm} + s\big(j^+_{xy} - j_{xy} \mid (\rho_y Q_{yx})^{1-\lambda}(\rho_x Q_{xy})^{\lambda}\big),\\
 \L_{F-2\lambda F^\sym}(\rho,j) &= \inf_{j^+\in \RR^{\fluxS}_{\geq0}} \sumsum_{(x,y)\in\fluxS} s\big(j^+_{xy} \mid (\rho_x Q_{xy})^{1-\lambda}(\rho_y\tfrac{\pi_x}{\pi_y}Q_{xy})^\lambda \big) \\
&\hspace{3cm}+ s\big(j^+_{xy} - j_{xy} \mid (\rho_y Q_{yx})^{1-\lambda}(\rho_x\tfrac{\pi_y}{\pi_x}Q_{yx})^\lambda  \big),\\
  \L_{F-2\lambda F^\asym}(\rho,j) &= \inf_{j^+\in \RR^{\fluxS}_{\geq0}} \sumsum_{(x,y)\in\fluxS} s\big(j^+_{xy} \mid (\rho_x Q_{xy})^{1-\lambda}(\rho_x\tfrac{\pi_y}{\pi_x}Q_{yx})^\lambda\big) \\
&\hspace{3cm} + s\big(j^+_{xy} - j_{xy} \mid (\rho_y Q_{yx})^{1-\lambda}(\rho_y\tfrac{\pi_x}{\pi_y}Q_{xy})^\lambda\big),
\end{align*}
and the corresponding Fisher informations
\begin{align*}
 \RF^\lambda_F(\rho)=-\H\big(\rho,-2\lambda F(\rho)\big) &=\sumsum_{\substack{x,y\in\X\\x\neq y}} \rho_xQ_{xy} - (\rho_xQ_{xy})^{1-\lambda}(\rho_yQ_{yx})^{\lambda}, \\
  \RF^\lambda_{\Fsy}(\rho)=-\H\big(\rho,-2\lambda F^\sym(\rho)\big) &=\sumsum_{\substack{x,y\in\X\\x\neq y}}  \rho_x Q_{xy} - (\rho_x Q_{xy})^{1-\lambda}(\rho_y \tfrac{\pi_x}{\pi_y}Q_{xy})^\lambda,\\
  \RF^\lambda_{\Fasy}(\rho)=-\H\big(\rho,-2\lambda F^\asym(\rho)\big) &=\sumsum_{\substack{x,y\in\X\\x\neq y}} \rho_xQ_{xy} - (\rho_x Q_{xy})^{1-\lambda}(\rho_x\tfrac{\pi_y}{\pi_x} Q_{yx})^\lambda.
\end{align*}
While non-negativity of these Fisher informations is guaranteed by construction, it can also be proven directly by using $(1-\lambda) a+\lambda b \geq a^{1-\lambda}b^\lambda$. For $\lambda=\frac12$, all three Fisher informations are of the form ${\sumsum}_{x\neq y}(\sqrt{\cdot}-\sqrt{\cdot})^2$; interpreting the difference as an abstract discrete gradient, this is reminiscent of the usual Fisher information in continuous space $\frac12\int\!(\grad\sqrt{\rho(x)})^2\,dx$.

These decompositions provide new variational characterisations for the IPFG example, which coincide with the classical gradient-flow structure for Markov chains satisfying detailed balance (see Section~\ref{subsec:flows}) and lead to the FIR inequality as a special case (see Example~\ref{ex:IPFG-FIR} below). The decomposition~\eqref{eq:FIIR Fsym} with $\lambda=\frac12$ was first discussed in~\cite[Cor.~4]{KaiserJackZimmer2018}.

All three L-functions $\L_{(1-2\lambda)F}$, $\L_{F-2\lambda F^\sym}$ and $\L_{F-2\lambda F^\asym}$ are the large-deviation cost functions for processes with altered jump rates.
In particular, $\L_{F^\sym}=\L_{F-F^\asym}$ is the large-deviation cost function corresponding to the jump process with jump rates for a particle to jump from $x$ to $y$ given by
\begin{equation*}
  \kappa^\sym_{xy}(\rho):=\rho_x\sqrt{Q_{xy} Q_{yx}\mfrac{\pi_y}{\pi_x}}=\rho_x\sqrt{Q_{xy}\overleftarrow{Q}_{xy}},
\end{equation*}
where we write $\overleftarrow{v}_{xy}:=v_{yx}\frac{\pi_y}{\pi_x}$ for the jump rate of a single time-reversed jump process~\cite[Thm.~3.7.1]{Norris1998}. The linearity in $\rho_x$ reflects that the system consists of \emph{independent} Markov particles with generator $\sqrt{Q_{xy}\overleftarrow{Q}_{xy}}$ \cite{Renger2018a,Kraaij2017}.

Similarly, $\L_{F^\asym}=\L_{F-F^\sym}$ is the large-deviation cost function corresponding to a system with jump rates for one particle to jump from $x$ to $y$ given by~\cite{PattersonRenger2019}
\begin{equation}\label{eq:asym-jump-rates}
  \kappa^\asym_{xy}(\rho):=Q_{xy} \sqrt{\rho_x\rho_y \mfrac{\pi_x}{\pi_y}}=\sqrt{\rho_x\rho_y}\sqrt{Q_{xy}\overleftarrow{Q}_{yx}}.
\end{equation}
We can interpret $\L_{F^\asym}(\rho,j)$ as the flux large-deviation cost function corresponding to a system of interacting particles with jump rates $n\kappa_{xy}^\asym(\rho)$~\cite{AAPR21}. It should be noted that the usual large-deviation proof techniques break down in this particular case due to the non-uniqueness of solution to the limiting antisymmetric ODE (see Proposition~\ref{prop:IPFG-Ham-state}). 
\end{example}
\endgroup

The next corollary connects the decomposition~\eqref{eq:FIIR Fsym} to an (abstract-)FIR inequality (recall Section~\ref{subsubsec:FIR-intro}) only defined on the state-space $\Z$ and with no dependence on the flux-space $\W$. In order to make this connection we introduce the contracted L-function $\hat \L:T_\rho\Z\rightarrow\mathbb R\cup\{\infty\}$ defined as
\begin{equation}\label{def:FIIR-FIR-contract}
  \hat\L(\rho,u):= \inf_{\substack{j\in T_{\rho}\W: \, u=d\phi_{\rho}j}} \L(\rho,j).
\end{equation}
The definition of $\hat\L$ is inspired by the contraction principle in large-deviation theory, where $\hat\L$ is the large-deviation rate functional only on the state space (recall Example~\ref{ex:indep LDP}). This connection will be further clarified in Proposition~\ref{prop:LDP-contract}.

\begin{cor}[FIR inequality]\label{corr:FIR-abstract}
Let $\L$ be an L-function on $\Z$. 
For any $\rho\in\Dom(F^\asym)$, $u\in T_\rho\Z$ and $\lambda\in [0,1]$ we have
\begin{equation*}
\hat \L(\rho,u)\geq \RF^\lambda_{\Fsy}(\rho) + \lambda\langle d\QP(\rho),u \rangle,
\end{equation*}
where $\hat\L$ (with convex dual $\hat\H$) is defined in~\eqref{def:FIIR-FIR-contract} and $\RF^\lambda_{\Fsy}(\rho)=-\hat\H(\rho,\lambda d\QP)$.
\end{cor}
\begin{proof}
Using convex duality and~\eqref{def:FIIR-FIR-contract} it follows that $\RF^\lambda_{\Fsy}(\rho)=-\H(\rho,\lambda d\phi\tp_\rho d\QP)=-\hat\H(\rho,\lambda d\QP)$.
Using~\eqref{eq:FIIR Fsym} and the definition of $\Fsy$~\eqref{def:Dom-Fsym} we find
\begin{align*}
\hat\L(\rho,u) &=  \inf_{\substack{j\in T_{\rho}\W: \, u=d\phi_{\rho}j}} \bigl[ \L_{F-2\lambda F^\sym}(\rho,j)  - 2\lambda \langle F^\sym(\rho),j\rangle\bigr] + \RF^\lambda_{\Fsy}(\rho)\\
&= \inf_{\substack{j\in T_{\rho}\W: \, u=d\phi_{\rho}j}} \bigl[ \L_{F-2\lambda \Fsy}(\rho,j) \bigr] + \RF^\lambda_{\Fsy}(\rho) +   \lambda \langle d\QP(\rho),u\rangle\\
&\geq \RF^\lambda_{\Fsy}(\rho) +   \lambda \langle d\QP(\rho),u\rangle,
\end{align*}
where the second equality follows since $\langle d\phi\tp_\rho\eta,j\rangle = \langle \eta,d\phi_\rho j\rangle$ and the inequality follows since tilted L-functions are non-negative by definition (see Lemma~\ref{lem:preFIR} \& Definition~\ref{def:L-function}).  
\end{proof}
\begingroup
\allowdisplaybreaks
\begin{example} \label{ex:IPFG-FIR}
We now comment on the connection with the FIR inequality in~\cite{HilderPeletierSharmaTse20}. Let $\rho\in C^1([0,T];\Dom(\Fsy))$, where we have abused notation so that $\rho$ is now a trajectory, and recall that $\Dom(\Fsy)=\P_{+}(\X)$. Since $\dot\rho(t)\in T_{\rho(t)}\Z$, using Corollary~\ref{corr:FIR-abstract}, for any $t\in [0,T]$ and $\lambda\in [0,1]$  we have
\begin{equation*}
  \hat\L(\rho(t),\dot\rho(t))\geq \RF^\lambda_{\Fsy}(\rho(t))+\lambda \tfrac{d}{dt}\QP(\rho(t)),
\end{equation*}
where we have used $\langle d\QP(\rho(t)),\dot\rho(t) \rangle =  \frac{d}{dt}\QP(\rho(t))$. Integrating in time, which is allowed since $\rho$ is a sufficiently smooth curve, we find
\begin{equation*}
  \mfrac{1}{\lambda}\int_0^T\hat\L(\rho(t),\dot\rho(t))dt + \QP(\rho(0)) \geq \mfrac{1}{\lambda}\int_0^T\RF^\lambda_{\Fsy}(\rho(t))dt+\QP(\rho(T)).
\end{equation*}
This is exactly the FIR inequality in~\cite[Thm.~1.6]{HilderPeletierSharmaTse20}, although this paper has two crucial generalisations. First, using approximation arguments, in~\cite{HilderPeletierSharmaTse20} the class of admissible curves is extended to $\rho\in AC([0,T];\Z)$, i.e.\ absolutely continuous curves in $\Z=\P(\X)$ instead of $\P_+(\X)$ discussed above (recall the discussion in Section~\ref{subsubsec:FIR-intro}). Second, in~\cite{HilderPeletierSharmaTse20} the relative entropy $\RelEnt(\rho(t)|\mu(t))$ with respect to any time-dependent solution $\mu$ of the corresponding macroscopic dynamics (which is the forward Kolmogorov equation)
\begin{equation}\label{eq:IPFG-ForKol}
  \dot\mu(t)=Q\tp\mu(t),
\end{equation}
is used as opposed to the quasipotential $\QP(\rho)=\mathrm{RelEnt}(\rho(t)|\pi)$, where $\pi$ is the invariant measure of~\eqref{eq:IPFG-ForKol}. We believe that this generalisation from the invariant measure $\pi$ to any time dependent solution $\mu(t)$ is a feature of the linear forward Kolmogorov equations (similar results also hold for linear Fokker-Planck equations~\cite[Thm.~1.1]{BogachevRocknerShaposhnikov16},~\cite[Thm.~4.18]{DLPSS17} arising from diffusion processes), and cannot be expected to hold in the setup of our paper where we are interested in nonlinear macroscopic equations. This is also the case for nonlinear diffusion processes~\cite[Thm.\ 2.3]{DuongLamaczPeletierSharma17}.
\end{example}
\endgroup

\subsection{Symmetric and antisymmetric L-functions}
\label{subsec:flows}

In this section we focus on the two terms $\L_\Fsy$ and $\L_\Fasy$ in the decompositions \eqref{eq:FIIR Fasym 1/2} and \eqref{eq:FIIR Fsym 1/2} respectively. Observe that $\L=\L_\Fsy$ if $\Fasy=0$, and therefore $\L_\Fsy$ corresponds to a system with a purely symmetric force. The relation between such systems with gradient flows is well known and follows from the theory in the previous sections, but for completeness we will make this connection explicit here. Similarly, $\L_\Fasy$ corresponds to a system with a purely antisymmetric force; in the level of abstraction of our current paper such systems are less understood. Motivated by our analysis in Section~\ref{sec:IPFG antisym flow} and the examples in Section~\ref{sec:examples} we conjecture below that these L-functions are related to Hamiltonian systems.

We first discuss the purely symmetric case. Note that when particle systems and large-deviations are involved, $\L_\Fsy$ is the large-deviation cost function of a microscopic system in detailed balance (see Corollary~\ref{cor:generalised DB}). In what follows we will make use of the contracted dissipation potential $\hat\Psi:T_\rho\Z\rightarrow\mathbb R\cup\{\infty\}$ defined as
\begin{equation}\label{def:FIIR-FIR-contract-diss}
\hat\Psi(\rho,u):= \inf_{\substack{j\in T_{\rho}\W: \, u=d\phi_{\rho}j}} \Phi(\rho,j).
\end{equation}

\begin{cor}[EDI]
\label{cor:symmetric=EDP}
Let $\L$ be an L-function on $\Z$ and $\rho\in\Dom(\Fasy)$. For any $j\in T_\rho\W$ we have
\begin{equation}\label{eq:FIIR-EDI}
  \L_\Fsy(\rho,j) = \Phi(\rho,j) + \Phi^*(\rho,-\tfrac12 d\phi_\rho\tp d\QP(\rho))+\tfrac12 \langle d\phi_\rho\tp d\QP(\rho),j\rangle,
\end{equation}
and for any $u\in T_\rho\Z$ we have 
\begin{equation}\label{eq:FIIR-EDI-cont}
  \hat\L_\Fsy(\rho,u) = \hat\Psi(\rho,u) + \hat\Psi^*\big(\rho,-\tfrac12 d\QP(\rho)\big)+\tfrac12 \langle  d\QP(\rho),u\rangle,
\end{equation}
where $\hat\L_\Fsy$, $\hat\Psi$ are defined in~\eqref{def:FIIR-FIR-contract},~\eqref{def:FIIR-FIR-contract-diss} and $\hat\Psi^*(\rho,\xi)=\Phi^*(\rho,d\phi_\rho\tp\xi)$ is the convex dual of $\hat\Psi$.
Additionally if $\rho\in \Dom_\symdiss(F^\asym)$, then for any $j\in T_\rho\W$ and $u\in T_\rho\Z$ we have the symmetry relations
\begin{equation}\label{eq:GF-symmetry}
  \L_\Fsy(\rho,j) - \L_\Fsy(\rho,-j)=\langle d\phi\tp_\rho d\QP(\rho),j \rangle, \ \ \hat\L(\rho,u) - \hat\L(\rho,-u)=\langle d\QP(\rho),u \rangle. 
\end{equation}
\end{cor}
\begin{proof}Using $\Fasy=0$ we have $F(\rho)=\Fsy(\rho)$, and the decomposition~\eqref{eq:FIIR-EDI} then follows from~\eqref{eq:FIIR Fsym 1/2} since $\L_0(\rho,j)=\Phi(\rho,j)$ (see~\eqref{eq:L_LF}), $\RF^\frac12_{\Fsy}(\rho)=\Phi^*(\rho,\Fsy(\rho))$ and using the definition of $\Fsy$~\eqref{def:Dom-Fsym}. The decomposition~\eqref{eq:FIIR-EDI-cont} follows by applying the infimum in~\eqref{def:FIIR-FIR-contract} to~\eqref{eq:FIIR-EDI} and noting that by definition of convex duality $\hat\Psi^*(\rho,\xi)=\Phi^*(\rho,d\phi\tp \xi)$ for any $\xi\in T^*_\rho\Z$. The first symmetry relation follows by Lemma~\ref{lem:symmetric dissipation}(ii) and the second symmetry relation following by taking the infimum of the first symmetry relation on both sides.   
\end{proof}
Note that the decomposition~\eqref{eq:FIIR-EDI} also follows from~\eqref{eq:FIIR Fasym 1/2} by using~\eqref{eq:L=Phi Phis1}, but for $\rho\in \Dom_\symdiss(\Fasy)$. 
Let us first comment on the \emph{contracted} symmetric function $\hat\L_\Fsy$. Clearly, its zero-cost velocity $u^0(\rho)$ satisfies the EDI
\begin{equation*}
  \hat\Psi\big(\rho,u^0(\rho)\big) + \hat\Psi^*\big(\rho,-\tfrac12 d\QP(\rho)\big)+\tfrac12 \langle  d\QP(\rho),u^0(\rho)\rangle = 0,
\end{equation*}
which is equivalent by convex duality to a generalised gradient flow~\eqref{eq:GGS}. Summarising Corollaries~\ref{cor:generalised DB} and \ref{cor:symmetric=EDP}, if a microscopic system is in detailed balance, the large-deviation cost function $\L=\L_\Fsy$ has a purely symmetric force, and hence induces a generalised gradient flow. This connection between gradient flows and detailed balance was first discussed in this generality in \cite{MielkePeletierRenger14}. For the IPFG example, the second symmetry relation in~\eqref{eq:GF-symmetry} correspond to the classical gradient structure for finite-state Markov chains in detailed balance~\cite[Sec.\ 4.1]{MielkePeletierRenger14} and the decomposition~\eqref{eq:FIIR-EDI} is the corresponding flux formulation of the gradient structure for this example~\cite[Sec.\ 4.5]{Renger2018a}. Note that, strictly speaking~\eqref{eq:FIIR-EDI} is not a gradient flow in the density-flux space. However a careful rewriting allows us to see $\L_\Fsy$ as a gradient flow, as summarised in the following remark.

%
%
\begin{remark}
With $\L^\W_\Fsy(w,j):=\L_\Fsy(\phi\lbrack w \rbrack,j)$, and applying the chain rule $d_w\QP^\W(w)=d\phi_{\phi\lbrack w \rbrack}\tp d_\rho\QP(\phi\lbrack w\rbrack)$, we arrive at 
\begin{equation}  \label{eq:FIIR-EDIW} 
  \L^\W_\Fsy(w,j) = \Phi^\W(w,j) + {\Phi^\W}^*\big(w,-\tfrac12 d_w\QP^\W(w)\big)+\tfrac12 \langle d_w\QP^\W(\rho),j\rangle.
\end{equation}
In this formulation $\L_{\Fsy}$ is indeed a gradient flow in the density-flux space~\cite{Renger2018b}.
\end{remark}

As far as we are aware, the purely antisymmetric cost $\L_\Fasy$ has not been studied in the literature, and we could not produce rigorous results for it in the abstract setting of this section. However, as will be discussed in forthcoming sections, we are able to show that for certain examples the zero-cost \emph{velocity} associated to $\L_\Fasy$ is non-dissipative, in the sense that one can associate a non-trivial conserved energy and a skew-symmetric operator to it, which motivates the following conjecture.
\begin{conjecture}\label{conj}
 Let $\L$ be an L-function on $\Z$ and $\hat\L_{\Fasy}$ be the contracted L-function corresponding to $\hat\L_{\Fasy}$, i.e.\ 
\begin{equation*}
\hat\L_{\Fasy}(\rho,u):= \inf_{\substack{j\in T_{\rho}\W: \, u=d\phi_{\rho}j}} \L_{\Fasy}(\rho,j). 
\end{equation*}
Then there exists an energy $\E:\Z\rightarrow\RR$ and a skew-symmetric operator $\pJ:\rho\mapsto (T_\rho^*\Z\to T_\rho\Z)$ such that the zero-cost velocity of $\hat\L_\Fasy$ can be written as
\begin{equation*}
  u^0(\rho)=\pJ(\rho)D\E(\rho).
\end{equation*}
\end{conjecture}
Clearly, the skew-symmetry of $\pJ(\rho)$ implies that the energy $\E(\rho(t))$ will be conserved along solutions of $\dot\rho(t)=\pJ(\rho(t))D\E(\rho(t))$. In fact, for the IPFG and lattice gas examples, the corresponding $\pJ$ even satisfies the Jacobi identity, so that the purely antisymmetric velocity has a Hamiltonian structure (see Sections~\ref{sec:IPFG antisym flow}, \ref{sec:LatticeGas} for details).


\section{Formal connection with large deviations}\label{sec:large deviations}

In Section~\ref{sec:abstract} we focussed on the purely macroscopic setting. In this section we motivate the abstract structures introduced therein by connecting them to Markov processes and their large deviations. Although the results presented in this section are largely known in the literature in specific settings, we include them here in a more general setting to provide rationale for the abstract framework discussed in the last section. While these results are formal due to the level of generality at which we work, they can be made rigorous case by case. 

Throughout this section we assume a \emph{microscopic} dynamics described by a sequence of Markov processes $(\rho\super{n}(t),W\super{n}(t)\big)$ defined on $\Z\times\W$. Typically, $\rho\super{n}(t)$ is the empirical measure, concentration or density corresponding to $\mathcal{O}(n)$ particles, and $W\super{n}(t)$ is the integrated/cumulative particle flux (recall Example~\ref{ex:indep LDP} and see Section~\ref{sec:examples} for further examples). For now, we assume a fixed deterministic initial condition $\rho\super{n}(0)$ for the empirical measure; this will be relaxed later on. We always assume that the initial condition for the flux satisfies $W\super{n}(0)=0$ almost surely, since the particles have not moved yet at initial time. For any $t\geq0$, the integrated flux $W\super{n}(t)$ contains all information required to reconstruct the current state of the system, i.e.\ almost surely 
\begin{equation*}
\rho\super{n}(t)=\phi\lbrack W\super{n}(t)\rbrack.
\end{equation*}
Equivalently, if the random paths allow for a notion of (measure-valued) time-integration, we write
\begin{equation*}
  \dot\rho\super{n}(dt)=d\phi_{\rho\super{n}(t)} \dot W\super{n}(dt).
\end{equation*}

%

We assume that the sequence $(\rho\super{n}(t),W\super{n}(t)\big)$ satisfies a law of large numbers, whereby the microscopic process $\big(\rho\super{n}(t),W\super{n}(t)\big)$ converges to a macroscopic, deterministic trajectory $(\rho(t),w(t))$, which satisfies an equation of the form~\eqref{eq:coupled ev eq}, where at this stage we are only interested in the instantaneous flux $j=\dot w$. Consequently, the corresponding path probability measures $\PP\super{n}=\Law(\rho\super{n},W\super{n})$ will concentrate on that path $(\rho,w)$ as $n\to\infty$. 

Finally we assume that the sequence $(\rho\super{n}(t),W\super{n}(t)\big)$ satisfies a corresponding large-deviation principle in $\Z\times\W$, which can be formally written as 
\begin{equation}\label{eq:ldp ldp}
  \PP\super{n}\big( (\rho\super{n},W\super{n})\approx (\rho,w)\big) \sim e^{-n\int_0^T\!\L(\rho(t),\dot w(t))\,dt}.
\end{equation}
This large-deviation principle characterises the exponentially vanishing probability of paths starting from the fixed deterministic initial conditions which do not converge to the macroscopic path $(\rho,w)$. The function $\L$ is non-negative and its zero-cost flux corresponds to the macroscopic path, since for that path $\PP\super{n} \sim 1$. 

In what follows, we first focus on the classical technique for proving the aforementioned large-deviation statement, which motivates the tilted L-function introduced in Lemma~\ref{lem:preFIR}. Consequently we motivate the Definition~\ref{def:I0} of the quasipotential via the large deviations of invariant measures, and the Definition~\ref{def:time-reversedL} of the reversed L-function using time-reversal.

\subsection{Tilting, contraction and mixture}\label{sec:tilt}

Rigorous proofs of large-deviation principles for Markov processes tend to be rather technical. We nevertheless briefly review the classical proof technique, since it is closely related to the macroscopic framework introduced in Subsection~\ref{subsec:abstract forces}. For an example of this technique see~\cite[Chap.~10]{KipnisLandim1999}.

\begin{formal}\label{th:ldp}
Let $\Q\super{n}$ be the generator of the Markov process $(\rho\super{n}(t),W\super{n}(t))$, define
\begin{align*}
  \H\super{n}(\rho,w,\zeta) 
    &:=
  \frac1n e^{-n\langle\zeta,w\rangle}\Q\super{n} e^{n\langle\zeta,w\rangle},
\end{align*}
and let the limit $\H(\rho,\zeta)=\lim_{n\to\infty}\H\super{n}(\rho,w,\zeta)$ exist and be dependent on $w$ only via the relation $\rho=\phi\lbrack w\rbrack$. Then the process $(\rho\super{n},W\super{n})$ satisfies the large-deviation principle~\eqref{eq:ldp ldp} with
\begin{equation*}
  \L(\rho,j):=\sup_{\zeta\in T_\rho^*\W} \langle \zeta,j\rangle - \H(\rho,\zeta).
\end{equation*}
\end{formal}
The assumption that $\H$ depends on $w$ only via $\rho=\phi\lbrack w\rbrack$ will generally be justified if the noise only depends on the state $\rho$ of the system.

\begin{proof}[Main proof technique]
In order to derive the large deviations \eqref{eq:ldp ldp} for a given, atypical path $(\rho,w)$, one changes the probability measure $\PP\super{n}$ to a \emph{tilted} probability measure $\PP\super{n}_\zeta$. The tilting is defined via a time-dependent force field $\zeta(t)$ to be chosen later, and the Radon-Nikodym derivative is explicitly given by (see~\cite{Palmowski2002} for the generator of the tilted process and related technical details)
\begin{align}
  \frac{d\PP\super{n}_\zeta}{d\PP\super{n}}(\hat\rho,\hat w)
    &=
  \exp\Big\lbrack n   \int_0^T\!\Big( \langle \zeta(t),\dot{\hat w}(dt)\rangle - \H\super{n}\big(\hat\rho(t),\hat w(t),\zeta(t)\big) \Big)\,dt \Big\rbrack.
\label{eq:ldp tilting RN}
\end{align}
One can then (formally) estimate, for a small ball $\mathcal{B}_\vep(\rho,w)$ around the given atypical path $(\rho,w)$,
\begin{align*}
  -\mfrac1n\log\PP\super{n}\big( \mathcal{B}_\vep(\rho,w)\big)
    &=
  -\mfrac1n\log \int_{\mathcal{B}_\vep(\rho,w)}\!\frac{d\PP\super{n}}{d\PP\super{n}_\zeta}(\hat\rho,\hat w) \,\PP\super{n}_\zeta\big(d(\hat\rho,\hat w)\big) \\
    &\approx
  \mfrac1n\log \frac{d\PP\super{n}_\zeta}{d\PP\super{n}}(\rho,w) - \mfrac1n\log \PP\super{n}_\zeta\big(\mathcal{B}_\vep(\rho,w)\big)
  \qquad\text{(for small $\vep$)}\\
    &=
  \int_0^T\!\Big( \langle \zeta(t),\dot w(dt)\rangle - \H\super{n}\big(\rho(t),w(t),\zeta(t)\big) \Big)\,dt- \mfrac1n\log \PP\super{n}_\zeta\big(\mathcal{B}_\vep(\rho,w)\big).
\end{align*}
We choose $\zeta(t)$ to be optimum in $\sup_{\hat\zeta}\langle\hat\zeta,\dot w(t)\rangle-\H(\rho(t),\hat\zeta)$. It turns out that with this choice, the tilted probability $\PP\super{n}_\zeta$ will concentrate on the given path $(\rho,w)$ and therefore the final term in the right hand side vanishes (even for small $\vep$), which results in
\begin{equation*}
  -\mfrac1n\log\PP\super{n}\big( \mathcal{B}_\vep(\rho,w)\big)
    \stackrel{n\to\infty}{\approx}
  \int_0^T\!\sup_{\zeta} \Big( \langle \zeta,\dot w(dt)\rangle - \H\big(\rho(t),\zeta\big) \Big)\,dt
    =\int_0^T\!\L\big(\rho(t),\dot w(t)\big)\,dt.
\end{equation*}
\end{proof}

\begin{remark} On this formal level we do not specify the precise topological space in which the large-deviation principle holds; typically one can choose the Skorohod space $D(0,T;\Z\times\W)$, possibly requiring weaker topologies on $\Z\times\W$. However, this topological setting does not influence the geometric picture of Section~\ref{subsec:abstract setup}. We also stress that although the described proof strategy is classic, there are known cases were it fails \cite{Heydecker2023}. A different proof technique is developed in~\cite{FengKurtz06}, but the main argument described above are the same.
\end{remark}

Following similar arguments one can derive the large deviations of the tilted measures. 
\begin{corollary}\label{cor:LDP tilted L}
 For a given path $\zeta(t)$, 
the tilted probability $\PP\super{n}_\zeta$ from~\eqref{eq:ldp tilting RN} satisfies the large-deviation principle
\begin{equation}
  \PP\super{n}_\zeta\big( (\rho\super{n},W\super{n})\approx (\rho,w)\big) \sim e^{-n\int_0^T\!\L_{\zeta(t)}(\rho(t),\dot w(t))\,dt},
\label{eq:ldp tilted ldp}
\end{equation}
where $\L_{\zeta}$ is the convex dual of 
\begin{equation*}
  \H_\zeta(\rho,\hat\zeta):=\H(\rho,\zeta+\hat\zeta)-\H(\rho,\zeta).
\end{equation*}
\end{corollary}
The proof follows from the same arguments as Formal Theorem~\ref{th:ldp}, with \eqref{eq:ldp tilting RN} replaced by
\begin{align*}
  \frac{d\PP\super{n}_{\zeta+\hat\zeta}}{d\PP\super{n}_{\zeta}}(\hat\rho,\hat w)
    &=
  \frac{d\PP\super{n}_{\zeta+\hat\zeta}}{d\PP\super{n}_{\phantom{\hat \zeta}}}(\hat\rho,\hat w) \frac{d\PP\super{n}_{\phantom{\hat \zeta}}}{d\PP\super{n}_{\zeta \phantom{\hat \zeta}}}(\hat\rho,\hat w)\\
    &=
  \exp\Big\lbrack n   \int_0^T\!\Big( \langle \hat\zeta(t),\dot{\hat w}(dt)\rangle - \H\super{n}\big(\hat\rho(t),\hat w(t),\zeta(t)+\hat\zeta(t)\big) + \H\super{n}\big(\hat\rho(t),\hat w(t),\zeta(t)\big) \Big)\,dt \Big\rbrack.
\end{align*}
Note that $\H_{\zeta-F}$ is exactly as in~\eqref{def:tilt} and consequently we interpret the tilted L-functions introduced in  Definition~\ref{def:tilt} as the large-deviation cost functions for the tilted probability measures.

From the Formal Theorem~\ref{th:ldp}, one immediately obtains the following large-deviation principle for the state by applying the contraction principle {\cite[Thm.~4.2.1]{DemboZeitouni09}}, which motivates the definition~\eqref{eq:intro-contraction} \footnote{In practice, the mapping from paths $(\rho,w)$ to paths $\rho$ is clearly continuous. In order to make the statement rigorous one only needs to show that the infimum can be moved inside the integral.}
\begin{prop}\label{prop:LDP-contract}
Assume that the large-deviation principle~\eqref{eq:ldp ldp} holds for the pair $(\rho\super{n},W\super{n})$. Then the large-deviation principle also holds for $\rho\super{n}$, i.e.\ 
\begin{equation}\label{eq:ldp contraction}
  \PP\super{n}( \rho\super{n}\approx \rho\big) \sim e^{-n\int_0^T\!\hat\L(\rho(t),\dot\rho(t))\,dt}, 
  \quad\text{with}\quad
  \hat\L(\rho,\dot\rho):=\inf_{j:\dot\rho=d\phi_\rho j} \L(\rho,j).
\end{equation}
Moreover, $\hat\H(\rho,\xi):=\sup_{\dot\rho\in T_\rho\Z} \langle \xi,\dot\rho\rangle - \hat\L(\rho,\dot\rho)=\H(\rho,d\phi_\rho\tp\xi)$.
\end{prop}
So far we have assumed that the initial condition $\rho\super{n}(0)$ is fixed and deterministic. If the initial condition is random then we have the following result, which will be useful in what follows.
\begin{prop}[Mixing {\cite{Biggins2004}}]\label{prop:ldp mixture}
Assume that the large-deviation principle~\eqref{eq:ldp ldp} holds for the pair $(\rho\super{n},W\super{n})$ with a deterministic initial condition. If the initial condition is replaced by a sequence $\rho\super{n}(0)\in\Z$ which satisfies the large-deviation principle
\begin{equation*}
  \PP\super{n}\big( \rho\super{n}(0)\approx \rho\big) \sim e^{-n\I_0(\rho)}
\end{equation*}
for some functional $\I_0:\Z\to\lbrack0,\infty\rbrack$ and $W\super{n}(0)=0$ almost surely, then 
the pair $(\rho\super{n},W\super{n})$ with random initial condition $\rho\super{n}(0)\in\Z$ satisfies 
the large deviation principle 
\begin{equation}\label{eq:ldp mixture}
  \PP\super{n}\big( (\rho\super{n},W\super{n})\approx (\rho,w)\big) \sim e^{-n\I_0(\rho(0)) -n\int_0^T\!\L(\rho(t),\dot w(t))\,dt}.
\end{equation}
\end{prop}

\begin{remark}
The abstract framework introduced in Subsection~\ref{subsec:abstract setup} automatically fixes the state $\rho(0)=\phi\lbrack0\rbrack$, which coincides with deterministic initial conditions in context of large deviations. Strictly speaking, to work with varying random initial conditions would require additional flexibility in the abstract framework. This can be achieved by either replacing the mapping $\phi$ (recall Definition~\ref{def:state-flux-triple}) by a family of mappings $(\phi_{\rho(0)})_{\rho(0)}$, or by keeping a fixed reference state $\phi\lbrack0\rbrack$, and redefining the initial integrated flux as $w(0)\in\phi^{-1}\lbrack\rho(0)\rbrack$, exploiting the surjectivity of $\phi$. To keep the notation simple, we stick to the setup of a deterministic initial condition, and with a slight abuse of notation always tacitly assume that $\rho(t)=\phi\lbrack w(t)\rbrack=\phi_{\rho(0)}(w(t))$.
\end{remark}

\subsection{Quasipotential}\label{sec:QuasiLDP}
We now motivate Definition~\ref{def:I0} of the quasipotential $\QP$. The following result is largely  known in the literature, see for instance~\cite[Sec.~2.2]{BDSGJLL2002}, \cite[Sec.~3.3]{Bouchet2020},~~\cite[Sec.~4]{BorkarSundaresan13}  and~\cite[Cor.~2]{JiaJiangLi21}, although it is not often made explicit at the level of generality used in this section.   
\begin{theorem}\label{th:LDP QP} Assume that the Markov process $\rho\super{n}(t)$ satisfies the large-deviation principle~\eqref{eq:ldp contraction} and has an invariant measure $\Pi\super{n}\in\P(\Z)$ that satisfies the large-deviation principle
\begin{equation}\label{eq:inv ldp}
  \Pi\super{n}\big(\mu\super{n}\approx \mu\big)\sim e^{-n\QP(\mu)},
\end{equation}
where $\mu\super{n}$ denotes a random variable distributed with $\Pi\super{n}$. Then we have
\begin{enumerate}[label=(\roman*)]
\item $\displaystyle \QP(\mu) \equiv \inf_{\substack{\hat\rho\in C^1_b([0,T];\Z):\\ \hat\rho(T)=\mu}} \Big\{ \QP\big(\hat\rho(0)\big) + \int_0^T\!\hat\L\big(\hat\rho(t),\dot{\hat\rho}(t)\big)\,dt \Big\}$ \quad for any $T\geq0$,
\eqnum\label{eq:V=V+L}
\item\label{prop:I0}
 $\displaystyle \H\big(\mu,d\phi_\mu\tp d\QP(\mu)\big)= \hat\H\big(\mu,d\QP(\mu)\big)\equiv0$, 
\end{enumerate}
where $\hat\L,\hat\H$ are defined in Proposition~\ref{prop:LDP-contract}.
\end{theorem}

Note that \eqref{eq:V=V+L} implies that $\V$ is always a Lyapunov function along the zero-cost dynamics, which can also be deduced from the decomposition~\eqref{eq:FIIR Fsym}. 

\begin{proof}[Formal proof]
For arbitrary $T>0$ and fixed deterministic initial condition $\rho\super{n}(0)=\rho(0)$, the state $\rho\super{n}_T$ satisfies the large-deviation principle~\cite[Thm.~4.2.1]{DemboZeitouni09},
\begin{align}
  P_T\super{n}\big(d\mu\mid\rho(0)\big)
    &:=
  \PP\super{n}\big( \rho\super{n}(T) \approx \mu \mid \rho\super{n}(0)=\rho(0)\big) 
    \sim
  e^{-nI_T(\mu\mid\rho(0))}, \quad\text{with} \notag\\
  I_T(\mu\mid\rho(0))&:=\inf_{\substack{\hat\rho\in C^1_b([0,T];\Z):\\\hat\rho(0)=\rho(0),\hat\rho(T)=\mu}} \,\,\int_0^T\!\hat\L\big(\hat\rho(t),\dot{\hat\rho}(t)\big)\,dt.
\label{eq:ConLag}
\end{align}
By definition the invariant measure is invariant under the transition probability, i.e.\ for any $T>0$,
\begin{equation*}
  \Pi\super{n}(d\mu)=\int\!P_T\super{n}(d\mu\mid\rho(0))\Pi\super{n}(d\rho(0)).
\end{equation*}
Hence the large-deviation functional of the left-hand side is equal to the large-deviation rate of the right-hand side, which using a mixing argument~\cite{Biggins2004} is given by
\begin{align*}
  \QP(\mu)
    =\inf_{\rho(0)\in\Z} \big\{ \QP(\rho(0)) + I_T\big(\mu\mid\rho(0)\big) \big\} 
    =\inf_{\rho(0)\in\Z} \inf_{\substack{\hat\rho\in C^1_b([0,T];\Z):\\\hat\rho(0)=\rho(0),\hat\rho(T)=\mu}} \Big\{ \QP(\rho(0)) +  \,\,\int_0^T\!\hat\L\big(\hat\rho(t),\dot{\hat\rho}(t)\big)\,dt \Big\}
\end{align*}
which proves the first claim. 
%
From here on the arguments are purely macroscopic. We proceed by noting that 
\begin{equation*} 
  \Xi_T(\rho):=\inf_{\substack{\hat\rho\in C^1_b([0,T];\Z):\\ \hat\rho(T)=\rho}} \QP\big(\hat\rho(0)\big) + \int_0^T\!\hat\L\big(\hat\rho(t),\dot{\hat\rho}(t)\big)\,dt,
\end{equation*}
which has the form of the value function from classical control theory, and hence solves the Hamilton-Jacobi-Bellman equation
\begin{equation}\label{eq:HJB}
  \dot \Xi_T(\rho) = -\hat\H\big(\rho,d \Xi_T(\rho)\big),
  \quad
  \Xi_0(\rho) = \QP(\rho).
\end{equation}
We have already shown that $\Xi_T\equiv\QP$ does not depend on $T$, and therefore $\dot\Xi_T(\rho)\equiv 0$, which proves the second claim.
\end{proof}

\begin{remark}
Strictly speaking, $\QP$ should be a \emph{viscosity solution} of the Hamilton-Jacobi-Bellman~\eqref{eq:HJB} and hence also of the stationary version Theorem~\ref{th:LDP QP}\ref{prop:I0}. However, it is not precisely clear to us which boundary conditions should be imposed in the definition of the viscosity solution. This issue is particularly challenging since most classical Hamilton-Jacobi-Bellman theory is developed for quadratic $\hat\H$ only. Therefore, Theorem~\ref{th:LDP QP}\ref{prop:I0} should be seen as formal. We remind the reader that a viscosity solution $\QP(\rho)$ is a solution in the classical sense at points of differentiability. At least on a formal level, this already suffices for the applications in this paper.
\end{remark}

\begin{remark}\label{rem:non-unique inv meas}
In Theorem~\ref{th:LDP QP}\ref{prop:I0} we do not require that the invariant measure is unique, neither do we claim that the quasipotential $\QP(\rho)$ will be unique. 
In particular, we do not require stable points $\pi\in\Z$ for which $\hat\L(\pi,0)=0$ to be unique. In case of uniqueness, the quasipotential from Theorem~\ref{th:LDP QP}\ref{prop:I0} will also satisfy the classical definition of the quasipotential~\cite{Freidlin1994}
\begin{equation*}
  \QP(\rho)=\inf_{\substack{\hat\rho\in C^1_b(-\infty,0;\Z):\\ \hat\rho(0)=\rho}} \int_{-\infty}^0\!\hat\L\big(\hat\rho(t),\dot{\hat\rho}(t)\big)\,dt.
\end{equation*}
In case of multiple stable points, one usually defines a family of \emph{non-equilibrium quasipotentials} indexed by the stable points~\cite{Freidlin1994}. Any one of these will also satisfy Theorem~\ref{th:LDP QP}\ref{prop:I0}, which is sufficient for our purpose. Therefore the abstract theory from Section~\ref{sec:abstract} can  be constructed with any of these quasipotentials.
\end{remark}

\subsection{Time reversal}\label{sec:time-rev}

In the following proposition we relate the large-deviation rate functions for Markov processes and their time-reversed counterparts, which motivates the notion of reversed L-function introduced in Definition~\ref{def:time-reversedL}. Since the proof below is standard in MFT, we only outline the proof idea for completeness.

\begin{prop}[{\cite[Sec.~II.C]{BDSGJLL2015MFT}, \cite[Sec.~4.2]{Renger2018a}}]\label{prop:LDP time reversal}
Let $\big(\rho\super{n}(t),W\super{n}(t)\big)$ be a Markov process with random initial distribution $\Pi\super{n}$ for $\rho\super{n}(0)$ and $W\super{n}(0)=0$ almost surely, where $\Pi\super{n}\in\P(\Z)$ is the invariant measure of $\rho\super{n}(t)$. Define the time-reversed process \footnote{This construction requires a vector structure on $\W$. For all applications that we have in mind this holds trivially, as long as we work with net fluxes (see the discussion in Example~\ref{ex:IPFG-diss-pot}).}
\begin{equation*}
  \overleftarrow \rho\super{n}(t):=\rho\super{n}(T-t), \quad
  \overleftarrow W\super{n}(t):=W\super{n}(T-t) - W\super{n}(T).
\end{equation*}
Assume that $\Pi\super{n}$ satisfies a large-deviation principle~\eqref{eq:inv ldp},   $\big(\rho\super{n}(t),W\super{n}(t)\big)$ with deterministic initial condition satisfies a large-deviation principle~\eqref{eq:ldp ldp} with cost function $\L$, and $\big(\overleftarrow\rho\super{n}(t),\overleftarrow W\super{n}(t)\big)$ with deterministic initial condition  satisfies a large-deviation principle~\eqref{eq:ldp ldp} with cost function $\overleftarrow\L$. Then for any $(\mu,j)\in\Z\times\W$, $\overleftarrow\L$ is related to $\L$ and $\QP$ via the relation
\begin{equation*}
  \overleftarrow\L(\mu,j) = \L(\mu,-j) + \langle d\phi_\rho\tp d\QP(\mu),j\rangle.
\end{equation*}
\end{prop}
\begin{proof} Note that if $\rho\super{n}(0)$ is distributed according to $\Pi\super{n}$, then so is $\overleftarrow\rho\super{n}(0)$, and if $W\super{n}(0)=0$ almost surely, then $\overleftarrow W\super{n}(0)=0$ almost surely as well. Since 
\begin{equation*}
  \PP\super{n}\big( \big(\rho\super{n},W\super{n}\big) \in (d\rho,dW)\big) = \PP\super{n}\big( \big(\overleftarrow\rho\super{n},\overleftarrow W\super{n}\big) \in (d\overleftarrow\rho,d\overleftarrow W)\big),
\end{equation*}
using Proposition~\ref{prop:ldp mixture},  we find for all paths $(\rho,w)$,
\begin{equation*}
 \QP(\rho(0)) + \int_0^T\!\L\big(\rho(t),\dot w(t)\big)\,dt = \QP(\rho(T)) + \int_0^T\!\overleftarrow\L\big(\rho(t),-\dot w(t)\big)\,dt.
\end{equation*}
Since the equality above holds for any $T>0$, we can write
\begin{align*}
  \big\langle d\phi_{\rho(0)}\tp d\QP(\rho(0)),\dot w(0)\big\rangle &= \big\langle d\QP(\rho(0)),\dot\rho(0)\big\rangle 
  = \lim_{T\to0} \frac{\QP(\rho(T))-\QP(\rho(0))}{T} \\
&= \lim_{T\to0} \frac1T \int_0^T\!\big\lbrack \L(\rho(t),\dot w(t)) - \overleftarrow\L(\rho(t),-\dot w(t))\big\rbrack\,dt= \L\big(\rho(0),\dot w(0)) - \overleftarrow\L\big(\rho(0),-\dot w(0)\big),
\end{align*}
for any $\rho(0)$ and $\dot w(0)$ (assuming sufficient regularity on $t\mapsto\L(\rho(t),\dot w(t)) - \overleftarrow\L(\rho(t),-\dot w(t))$).  The claimed result then follows by choosing any path $\rho,w$ for which $\rho(0)=\mu$ and $\dot w(0)=j$.
\end{proof}

A special and important case of the previous result pertains to detailed balance.

\begin{corollary} Let $\big(\rho\super{n}(t),W\super{n}(t)\big)$ and $\big(\overleftarrow\rho\super{n}(t),\overleftarrow W\super{n}(t)\big)$ be as in Proposition~\ref{prop:LDP time reversal}. If, under initial distribution $\Pi\super{n}\in\P(\Z)$ of $\rho\super{n}(0)$ and $\overleftarrow\rho\super{n}(0)$ and $W\super{n}(0)=\overleftarrow W\super{n}(0)=0$ almost surely,
\begin{equation}
  \PP\super{n}\big( \big(\rho\super{n},W\super{n}\big) \in (d\rho,dW)\big) = \PP\super{n}\big( \big(\overleftarrow\rho\super{n},\overleftarrow W\super{n}\big) \in (d\rho,dW)\big),
\label{eq:generalised DB}
\end{equation}
then $\L=\overleftarrow\L$.
\label{cor:generalised DB}
\end{corollary}

For the applications that we have in mind, the condition~\eqref{eq:generalised DB} holds precisely when $\rho\super{n}(t)$ is in detailed balance with respect to $\Pi\super{n}$, see for example \cite[Prop.~4.1]{Renger2018a}. The relation $\L=\overleftarrow\L$ is the time-reversal symmetry from \cite{MielkePeletierRenger14}, which implies that $\L$ induces a gradient flow, or $\Fasy=0$ in the context of this paper.

\section{Zero-cost velocity for IPFG antisymmetric L-function}
\label{sec:IPFG antisym flow}

In Subsection~\ref{subsec:flows} we argued that the both the purely symmetric flux and velocity are dissipative, that is, they are generalised gradient flows of the energy $\frac12\V$ (and $\frac12\V^\W$ respectively). Moreover, $\L_\Fsy$ \emph{defines} the variational structure of those gradient flows via the equalities \eqref{eq:FIIR-EDI} and \eqref{eq:FIIR-EDIW}.

The interpretation of $\L_{F^\asym}$ is more complicated. In general $\L_{F^\asym}$ will not have $\QP$ as its quasipotential, and using Lemmas~\ref{lem:symmetric dissipation} and \ref{lem:preFIR} for any $\rho\in\Dom_\symdiss(F^\asym)$ and $j\in T_\rho\W$ it
satisfies the time-reversal relation
\begin{equation*}
  \L_{-F^\asym}(\rho,j)=\L_{F^\asym}(\rho,-j).
\end{equation*}
This relation in fact holds for any tilted L-function, but $-F^\asym$ can be interpreted as the time-reversed counterpart of $F^\asym$ in the sense that $\overleftarrow{F^\sym+F^\asym}=F^\sym-F^\asym$ (see Remark~\ref{rem:rev-HL-relations}). Formally this means that time-reversal reverses the fluxes, which is a physical indication that $\L_{F^\asym}$ might correspond to Hamiltonian dynamics, as proposed in  Conjecture~\ref{conj}.

In this section we illustrate this principle for the IPFG example with L-function $\L$ from Example~\ref{eq:indep evolution flux}. As far as we are aware this is 
has not been studied in the literature, and as a first step we will focus solely on the trajectories of the zero-cost velocity $u(t)=\dot\rho(t)=u^0(\rho(t))$ of $\L_{F^\asym}$, largely ignoring fluxes as well as the variational structure.

Let $(\rho,j)$ satisfy $\L_{F^\asym}\big(\rho(t),j(t)\big)=0$ or equivalently $j(t)\in\partial\Phi^*\big(\rho(t),F^\asym(\rho(t))\big)$, where the subdifferential is with respect to the second variable. Substituting $\lambda=\tfrac12$ in $\L_{F-2\lambda\Fsy}$ (defined in Example~\ref{ex:IPFG-FIIR}), for any $x\in\X$, $\rho:[0,T]\to\P(\X)$ satisfies the ODE \footnote{Although $F^\asym$ is only defined on the interior, this ODE can be defined on the whole domain by continuous extension of $d\Phi^*(\rho,F^\asym)$.}
\begin{equation}\label{eq:IPFG antisymmetric ODE}
  \dot\rho_x(t) = -\Ddiv_x j(t) = \sum_{\substack{y\in\X\\y\neq x}} \Bigl( Q_{yx} \sqrt{\mfrac{\pi_y}{\pi_x}} - Q_{xy}\sqrt{\mfrac{\pi_x}{\pi_y}}\Bigr) \sqrt{\rho_x(t)\,\rho_y(t)}.
\end{equation}

Introducing the change of variables $\omega_x(t):=\sqrt{\rho_x(t)}$, the zero-cost velocity~\eqref{eq:IPFG antisymmetric ODE} transforms into a linear ODE with a matrix $A\in\mathbb R^{\X\times\X}$, i.e.\
\begin{align}\label{eq:IPFG antisymmetric ODE u} 
  \dot \omega(t) = \frac12 A\omega(t), \quad \text{with} \quad A_{xy}:= Q_{yx} \sqrt{\mfrac{\pi_y}{\pi_x}} - Q_{xy}\sqrt{\mfrac{\pi_x}{\pi_y}}.
\end{align}
Solutions to this equation have a nice geometric interpretation, see Figure~\ref{fig:rotations in u-space} for an example in three dimensions. Clearly, $\lvert\omega(t)\rvert_2^2=\lvert\rho(t)\rvert_1=1$ and so the solutions are confined to 
the unit sphere $S^{\X-1}$. On the other hand, the matrix $A$ is skewsymmetric with imaginary eigenvalues and represents rotations around the axis $\sqrt{\pi}$, implying that the solutions are confined to a plane perpendicular to $\sqrt{\pi}$. Therefore, solutions $\omega(t)$ lie on the intersection of these planes with the unit sphere, resulting in periodic orbits that conserve the distance of the plane to the origin. In the following result we show that this transformed system is indeed a Hamiltonian system with a suitable energy and Poisson structure which satisfies the Jacobi identity (see Lemma~\ref{lem:JacId-Alternate} for a useful alternative characterisation of the Jacobi identity in our context).


\begin{prop}
\label{prop:IPFG-Ham-trans}
The ODE~\eqref{eq:IPFG antisymmetric ODE u} admits a Hamiltonian structure $(\mathbb R^{\X\times\X},\tilde\E,\widetilde\pJ)$, i.e.\ $\dot \omega = \widetilde\pJ(\omega)\grad\tilde\E(\omega)$, where the linear energy $\tilde\E:\mathbb R^{\X}\rightarrow\mathbb R$ and Poisson structure $\widetilde\pJ:\mathbb R^{\X}\rightarrow \mathbb R^{\X\times\X}$ are given by 
\begin{equation*}
  \tilde \E(\omega):=1-\sqrt{\pi}\cdot \omega, \qquad \widetilde\pJ(\omega):=\frac12\Bigl(\sqrt{\pi}\otimes \left(A\omega\right) - \left(A\omega\right)\otimes \sqrt{\pi}\Bigr).
\end{equation*}
Here $\omega\cdot v$ is the standard Euclidean inner product and $\omega\otimes v$ is the outer product of vectors $\omega,v$. 
\end{prop}
\begin{proof}
In Appendix~\ref{app:HamFlow} we present a Hamiltonian structure for a general class of ODEs, which includes the transformed system~\eqref{eq:IPFG antisymmetric ODE u}. The proof of Proposition~\ref{prop:IPFG-Ham-trans} follows directly from Theorem~\ref{thm:lin-Ham} with the choice $d=|\X|$, $\omega_*=\sqrt{\pi}$ and observing that $|\omega_*|^2=\sum_x \pi_x = 1$ and $A\sqrt{\pi}=A^T\sqrt{\pi}=0$ since $\pi$ is the invariant solution corresponding to the original dynamics~\eqref{eq:IPFG antisymmetric ODE}.
\end{proof}

We would now like to transform the Hamiltonian structure of the transformed ODE~\eqref{eq:IPFG antisymmetric ODE u} back to obtain a Hamiltonian structure for the original non-linear equation~\eqref{eq:IPFG antisymmetric ODE}. This transforms the positive octant of the sphere in Figure~\ref{fig:rotations in u-space} to the simplex in Figure~\ref{fig:IPFG solutions}(c). However, transforming back via $\omega_x(t)=\sqrt\rho_x(t)$ is valid only if $\omega_x(t)\geq 0$ for every $x\in\X$. In the following result we state the criterion for this to hold. 

\begin{prop}\label{prop:IPFG-Ham-state}
Define the threshold 
\begin{equation*}
  \sigma:=\min_{x\in\X} \big(1-\sqrt{1-\pi_x}\big),
\end{equation*}
the energy $\E:\mathbb R^{\X}\rightarrow\mathbb R$ and the Poisson structure $\pJ:\mathbb R^{\X}\rightarrow \mathbb R^{\X\times\X}$ as 
\begin{equation*}
 \E(\rho):= 1- \sqrt{\pi}\cdot\sqrt{\rho}, \quad \left(\pJ(\rho)\right)_{xy}:=2 \sum_{z\in \X}  \bigl( \sqrt{\pi_x} A_{yz} - \sqrt{\pi_y} A_{xz}\bigr)\sqrt{\rho_x \rho_y\rho_z},
\end{equation*}
where $A$ is defined in~\eqref{eq:IPFG antisymmetric ODE u}.
If the energy of the initial distribution $\rho^0\in \P(\X)$ for the ODE~\eqref{eq:IPFG antisymmetric ODE} satisfies $0\leq \E(\rho^0)<\sigma$, then~\eqref{eq:IPFG antisymmetric ODE} has a unique solution and admits a Hamiltonian structure $(\mathbb R^{\X\times\X},\E,\pJ)$, i.e.\ $\dot \rho = \pJ(\rho)\grad\E(\rho)$. 
If the energy of the initial distribution satisfies  $\E(\rho^0)\geq\sigma$, then~\eqref{eq:IPFG antisymmetric ODE} has non-unique, non-energy-conserving solutions.
\end{prop}
\begin{proof}
We first analyse the critical case, where the periodic orbit $\omega(t)$ of \eqref{eq:IPFG antisymmetric ODE u} touches one of the boundaries of $S^{\X-1}\cap\RR^\X_{\geq0}$. The energy level of such an orbit can be calculated by solving the constrained minimisation problem
\begin{equation*}
  \min \big\{\tilde\E(\omega):\omega \in S^{\X-1}, \omega_x=0 \text{ for some } x\in \X\big\} =\min_{x\in\X} \, \min \big\{\tilde\E(\omega):\omega \in S^{\X-1}, \, \omega_x=0 \big\}.
\end{equation*}
Assume $x\in\X$ is optimal. For the interior minimisation problem, the optimal $\omega$ with $\omega_x=0$ solves
\begin{equation*}
  0 = \partial_{\omega_y} \big\lbrack \tilde\E(\omega) + \tfrac12\lambda\lvert\omega\rvert^2_2\big\rbrack = -\sqrt{\pi_y} + \lambda \omega_y, \quad \text{for all }y\neq x,
\end{equation*}
where the Lagrange multiplier $\lambda\geq0$ is such that the constraint $\lvert\omega\rvert_2^2=1$ holds. It follows that $\omega_y=\sqrt{\pi_y}/\sqrt{1-\pi_x}$, and so $\tilde\E(\omega)=1-\sqrt{1-\pi_x}=:\sigma$, yielding the critical case.

Using Proposition~\ref{prop:IPFG-Ham-trans} we thus find that if $\E(\rho^0)=\tilde\E(\omega^0)<\sigma$, the solution $\omega(t)$ of the linear system satisfies $\tilde\E(\omega(t))=\tilde\E(\omega^0)$ and remains positive (coordinate-wise), so that $\rho(t)=\sqrt{\omega(t)}$ solves \eqref{eq:IPFG antisymmetric ODE}, and has the corresponding transformed Hamiltonian structure. Note that this is possible since Poisson structures are preserved by coordinate transformations~\cite[Sec.~4.2]{Mielke1991}. The uniqueness of the thus constructed solution $\rho(t)$ follows since $\sqrt{\rho_x(t)\rho_y(t)}$ is strictly bounded away from zero, and therefore the right hand side of ~\eqref{eq:IPFG antisymmetric ODE} is Lipschitz. 


Now we show the non-uniqueness when $\E(\rho^0)\geq\sigma$, for simplicity with $|\X|=3$ only. The idea is to use the argument above to construct an energy-conserving solution until time $t_1$ it hits a boundary, say $\hat x=0$, a solution that moves along the boundary until an arbitrary time but sufficiently large time $t_1+\delta>0$, and an energy-conserving solution that moves away from the boundary again. See Figure~\ref{fig:IPFG solutions}(c). More precisely, let $\omega^0_x=\sqrt{\rho^0_x}$ and define
\begin{equation*}
	\rho_x(t):=\begin{cases}
		(e^{\tfrac12 A t}\omega^0)_x^2,										&0\leq t<t_1,\\
		(e^{\tfrac12 \bar{A} t}\omega^1)_x^2,				&t_1\leq t\leq t_1+\delta,\\
		(e^{\tfrac12 A t}\omega^2)_x^2,							&t>t_1+\delta	.
	\end{cases}		
\end{equation*}
Here $t_1:=\min\{t\geq0:(e^{\tfrac12 A t}\omega^0)_{\hat x}=0\}$, $\omega^1:=e^{\tfrac12 A t_1}\omega^0$ and $\omega^2:=e^{\tfrac12 A (t_1+\delta)}\omega^1$, and $\bar A_{xy}:= A_{xy}\mathds1_{\{x,y\neq \hat x\}}$. Note that $\delta>0$ must be large enough so that outgoing instead of incoming characteristics cross the boundary $\hat x=0$ and small enough that the corners in the simplex are avoided. It is easily checked that $\rho(t)$ is continuously differentiable and satisfies the ODE~\eqref{eq:IPFG antisymmetric ODE}. Since $\delta>0$ is arbitrary we have constructed an infinite number of solutions.
\end{proof}

\begin{figure}[ht]
\centering
\tikzset{->-/.style={decoration={
  markings,
  mark=at position #1 with {\arrow{>}}},postaction={decorate}}}

\begin{tikzpicture}[scale=2]
\tikzstyle{every node}=[font=\scriptsize];

  \draw[fill=gray] (0,1) arc (90:0:1) .. controls (0.9,-0.12) and (0.1,-0.25) .. (-0.5,-0.25) .. controls (-0.5,0.7) and (-0.12,0.9) .. (0,1);
  \draw[pattern=north west lines, pattern color=lightgray](1.6,0)--(0,1.6)--(-0.8,-0.4)--(1.6,0); 
  \draw[rotate around={45:(0.27,0.4)},->-=0.33,->-=0.67,->-=0.95,fill=lightgray](0.27,0.4) ellipse (0.3 and 0.4);

  \draw[dotted](0,0) --(0.27,0.4);
  \draw(0.27,0.4) -- (0.34,0.5) node[anchor=south]{$\sqrt{\pi}$};

  \filldraw (0.27,0.4) circle (0.01);
  \filldraw (0.34,0.5) circle (0.01);

  \draw[dotted](0,0)--(1.6,0);
  \draw[->](1.6,0)--(1.8,0) node[anchor=south]{$\omega_1$};
  \draw[dotted](0,0)--(0,1.6);
  \draw[->](0,1.6)--(0,1.8) node[anchor=west]{$\omega_3$};
  \draw[dotted](0,0)--(-0.8,-0.4);
  \draw[->](-0.8,-0.4)--(-1,-0.5) node[anchor=south]{$\omega_2$};
\end{tikzpicture}
\caption{For $|\X|=3$, the trajectories $\omega(t)$ rotate around the $\sqrt{\pi}$-axis, and lie at the intersection of the two-dimensional sphere $S^2$ and a plane perpendicular to the $\sqrt{\pi}$-axis. The transformation $\rho_x=\sqrt{\omega}_x$ maps the (octant) sphere to the simplex of Figure~\ref{fig:IPFG solutions}(c).} 
\label{fig:rotations in u-space}
\end{figure}

In the following remark we comment on the role of $\lambda\neq\frac12$ in $\L_{F-2\lambda F^\sym}$.

\begin{remark}
 One can also study the zero-cost velocity associated to $\L_{F-2\lambda F^\sym}$ from \eqref{eq:FIIR Fsym} for  $\lambda\in(0,1)$. For $\lambda<\frac12$, the symmetric part is dominant and the trajectories spiral inwards towards $\pi$, i.e.\ $\pi$ is a spiral sink, and for $\lambda>\frac12$, the antisymmetric part is dominant and the trajectories spiral outwards from $\pi$, i.e.\ $\pi$ is a spiral source (compare with Figure~\ref{fig:IPFG solutions}(c) for $\lambda=\frac12$).  
\end{remark}

\begin{remark}
As pointed out to us by Andr{\'e} Schlichting, the energy $\E(\rho)=\frac12\sum_{x\in\X}(\sqrt{\pi_x}-\sqrt{\rho_x})^2$ is exactly the squared Hellinger distance  between $\rho$ and the steady state $\pi$. At this stage we do not know the physical meaning behind the Hellinger distance, but it appears naturally in the context of purely time-antisymmetric flows.
\end{remark}

\section{Examples}\label{sec:examples}

Throughout Section~\ref{sec:abstract} we applied the abstract theory developed therein to the example of independent Markovian particles. We now apply the abstract theory to three examples of interacting particle systems. In Section~\ref{subsec:zero-range process} we consider the example of zero-range processes with an atypical scaling limit which leads to an ODE system in the limit as opposed to the usual parabolic scaling. Section~\ref{subsec:reacting particle system} deals with the case of chemical reaction networks in complex balance. Finally in Section~\ref{sec:LatticeGas} we consider the case of lattice gases with parabolic scaling (which lead to diffusive systems).


For each of these examples we derive the decompositions in Theorem~\ref{th:FIIRs},
\begin{align*}
  \L(\rho,j)&=\L_{(1-2\lambda)F}(\rho,j) + \RF^\lambda_F(\rho) - 2\lambda \langle F(\rho),j\rangle,\\
  \L(\rho,j)&=\L_{F-2\lambda F^\sym}(\rho,j) + \RF^\lambda_{\Fsy}(\rho) - 2\lambda \langle F^\sym(\rho),j\rangle,\\  
  \L(\rho,j)&=\L_{F-2\lambda F^\asym}(\rho,j) + \RF^\lambda_{\Fasy}(\rho) - 2\lambda \langle F^\asym(\rho),j\rangle,
\end{align*}
and explicitly calculate all the different terms. We stress that these decompositions were \emph{previously unknown} for zero-range processes and chemical reactions; we include the lattice gas example to show that for quadratic cost functions our decompositions coincide with existing results in MFT. 

We expect that by using approximation arguments similar to \cite[Thm~1.6]{HilderPeletierSharmaTse20}, \cite[Sec.~5]{RengerZimmer2021} and~\cite[Part~II.A]{Hoeksema23}, one can derive global-in-time decompositions of the rate functionals $\int_0^T\! \L(\rho(t),j(t))\,dt$; this is beyond the scope of the current paper.

\subsection{Zero-range processes}\label{subsec:zero-range process}

\iparagraph{Microscopic particle system.} To simplify and unify notation, we first consider the irreducible Markov process on a finite graph $\X$ from the IPFG example, with generator (represented by a matrix) $Q\in\RR^{\X\times\X}$, 
and assume that it has a unique and coordinate-wise positive invariant measure $\pi\in\P_+(\X)$. Similar to the setup in Example~\ref{ex:indep LDP} we study the Markov process $(\rho\super{n}(t),W\super{n}(t))$ on $\P(\X)\times\fluxS$, where $\rho\super{n}(t)$ is the particle density of \emph{interacting} particles and $W\super{n}(t)$ is the integrated net flux (both defined in Example~\ref{ex:indep LDP}). The interaction between the particles is so that the jump rate $n\kappa_{xy}(\rho)$ from $x$ to $y$ only depends on the density at the source node $x$ only (``zero-range'')
\begin{equation*}
  \kappa_{xy}(\rho)=\kappa_{xy}(\rho_x)=Q_{xy}\pi_x \eta_x\big(\mfrac{\rho_x}{\pi_x}\big),
\end{equation*}
for a family of functions $\eta_x:\lbrack0,\infty)\to\lbrack0,\infty)$ satisfying:
\begin{enumerate}[label=(\roman*)]
\item each $\eta_x$ is strictly increasing,
\item $\eta_x(0)=0$ and $\eta_x(1)=1$, 
\item $\log\eta_x(z)$ is integrable near $z=0$.
\end{enumerate}
The condition $\eta_x(0)=0$ ensures that $\rho_x\geq 0$, i.e.\ there are no negative densities. The  condition $\eta_x(1)=1$ ensures that $\pi$ is also an invariant measure for the many-particle limit~\eqref{eq:0range-limit}, and is assumed only for convenience (see Remark~\ref{rem:eta-prop} below). The integrability condition is necessary and sufficient for the large-deviation principle to hold~\cite{AAPR21}. Observe that the particular choice $\eta_x\equiv\mathrm{id}$ corresponds to the IPFG model.

The pair $\big(\rho\super{n},W\super{n}(t)\big)$ has the $n$-particle generator 
\begin{align*}
 (\Q\super{n} f)(\rho,w)= n\sumsum_{(x,y)\in\fluxS} &  \kappa_{xy}(\rho_x)\bigl[ f(\rho-\tfrac1n\mathds1_x+\tfrac1n\mathds1_y,w+\tfrac1n\mathds1_{xy})-f(\rho,w)\bigr] \\&
+ \kappa_{yx}(\rho_y)\bigl[ f(\rho-\tfrac1n\mathds1_y+\tfrac1n\mathds1_x,w-\tfrac1n\mathds1_{xy})-f(\rho,w)\bigr].
\end{align*}

As opposed to the typical diffusive scaling for zero-range processes~\cite{BDSGJLL2015MFT}, we keep the graph $\X$ fixed. The many-particle limit for this process as $n\rightarrow\infty$ is the solution to the ODE system~\cite[Sec.~3.1]{RengerZimmer2021}
\begin{equation}\label{eq:0range-limit}
\begin{cases}  
\dot w_{xy}(t) = \kappa_{xy}(\rho_x(t)) - \kappa_{yx}(\rho_y(t)), \  \ & (x,y)\in\fluxS, \\
  \dot\rho_x(t) = -\Ddiv_x \dot w(t), \ \ & x\in\X
\end{cases}
\end{equation}
where
$\Ddiv$ is again the discrete divergence defined in~\eqref{eq:discrete divergence}. 
The Markov process $(\rho\super{n}(t),W\super{n}(t))$ satisfies a large-deviation principle with the rate functional~\eqref{eq:flux LDP} where the corresponding $\L$ and its dual $\H$ are now given by \cite{PattersonRenger2019,GabrielliRenger2020,AAPR21}
\begin{equation}\label{def:Lag-0range}
\begin{aligned}
  \L(\rho,j) &= \inf_{j^+\in \RR^{\fluxS}_{\geq0}} \sumsum_{(x,y)\in\fluxS} \bigl[ s\big(j^+_{xy} \mid \kappa_{xy}(\rho_x)\big) + s\big(j^+_{xy} - j_{xy} \mid \kappa_{yx}(\rho_y)\big)\bigr],\\
  \H(\rho,\zeta)&=\sumsum_{(x,y)\in\fluxS} \bigl[ \kappa_{xy}(\rho_x)\big(e^{\zeta_{xy}}-1\big) + \kappa_{yx}(\rho_y)\big(e^{-\zeta_{xy}}-1\big)\bigr],
\end{aligned}
\end{equation}
and $s(\cdot\mid\cdot)$ is defined in~\eqref{def:entropy-s}.

\iparagraph{State-flux triple and L-function.} The manifolds $\Z,\W$ with the corresponding tangent and cotangent spaces and the map $\phi:\Z\rightarrow\W$  with $d\phi_\rho=-\Ddiv,\ d\phi\tp=\Dnabla$ are exactly as in Example~\ref{ex:IPFG-state-flux-triple}. It is easily checked that $\L$ and $\H$ are convex duals of each other, so that $\L$ is indeed convex and lower semicontinuous.

\iparagraph{Quasipotential.} Define $\QP:\Z\rightarrow\RR\cup\{\infty\}$ as
\begin{equation}\label{def:0-range-quasipot}
  \QP(\rho)=
    \begin{cases}
     \displaystyle \sum_{x\in\X} \int_{0}^{\rho_x}\!\log\eta_x\Big(\mfrac{z}{\pi_x}\Big)\,dz,      &\rho\in\P(\X),\\
      \infty,                                                                   &\text{otherwise},
    \end{cases}
\end{equation}


Note that $\QP$ depends on $Q$ through the steady state $\pi$ only. Moreover, the integral is well-defined due to the integrability condition on $\eta_x$. This function can be found as the large-deviation rate of the explicitly known invariant measure $\Pi\super{n}$ using  Theorem~\ref{th:LDP QP},~\cite[Prop.~3.2]{KipnisLandim1999} and \cite[Sec.~4.1]{GabrielliRenger2020}. However, in the next proposition we show that it is the correct quasipotential without any reference to a microscopic particle system, in the macroscopic sense of Definition~\ref{def:I0}. 

\begin{prop}
The function $\QP$ defined in~\eqref{def:0-range-quasipot} satisfies $\H(\rho,d\phi\tp d\QP(\rho))=0$ at all points of differentiability $\rho\in\P_+(\Z)$ of $\V$.
\label{prop:zero-range QP}
\end{prop}
\begin{proof}
At the points of differentiability of $\QP$ we have
\begin{align*}
  \H\big(\rho,d\phi\tp_\rho d\QP(\rho)\big)  &= \H\big(\rho,\Dnabla \log\eta(\tfrac{\rho}{\pi})\big)  = \sumsum_{(x,y)\in\fluxS} \Bigr(\kappa_{xy}(\rho_x) \Bigl[\frac{\eta_y(\rho_y/\pi_y)}{\eta_x(\rho_x/\pi_x)}-1\Bigr] + \kappa_{yx}(\rho_y)\Bigl[\frac{\eta_x(\rho_x/\pi_x)}{\eta_y(\rho_y/\pi_y)}-1\Bigr] \Bigl)\\
& = \sumsum_{\substack{x,y\in\X\\x\neq y}} \bigl(\pi_x Q_{xy}\eta_y\big(\tfrac{\rho_y}{\pi_y}\big)-\pi_x Q_{xy}\eta_x\big(\tfrac{\rho_x}{\pi_x}\big)   \bigr)  =\sum_{y\in\X}\eta_y\big(\tfrac{\rho_y}{\pi_y}\big)\sum_{\substack{x\in\X\\x\neq y}} ( \pi_x Q_{xy}-\pi_y Q_{yx}) = 0,
\end{align*}
where the fourth and fifth equality follows by exchanging indices and the final equality follows since $Q\tp\pi=0$.   
\end{proof}

\begin{remark}\label{rem:eta-prop} Let us discuss the various assumptions on $\eta_x$. Since $\eta_x$ is nonnegative and strictly increasing, it follows that $\QP(\rho)$ is strictly convex for any $\rho\in\P(\X)$, and consequently has a unique minimiser. The property $\eta(1)=1$ ensures that $\pi$ is this unique minimiser of $\QP$. If this condition is not satisfied then, as we show below, one can always construct $\overline Q\in\RR^{\X\times\X}$, $\overline\pi\in \P_+(\X)$ and family $\overline\eta_x$ with $\overline\eta_x(1)=1$, such that 
$\kappa_{xy}(\rho)=\overline Q_{xy}\overline\pi_x \overline\eta_x\big(\mfrac{\rho_x}{\overline\pi_x}\big)$,
$\overline Q\tp\bar\pi=0$, and $\overline\pi$ is the unique stable point of~\eqref{eq:0range-limit}. 
To calculate these modified objects, we minimise $\QP(\rho)$ for $\rho\in\P(\X)$, which gives the minimiser
\begin{equation*}
\overline \pi_x:=\pi_x \eta_x^{-1}(e^{-\lambda}), \quad \text{where} \ \lambda\in\RR \ \text{ satisfies } \ \sum_{x\in\X}\pi_x\eta_x^{-1}(e^{-\lambda})=1,
\end{equation*}
and define
\begin{equation*}
\overline \eta_x(z):=\eta_x\big(z\eta_x^{-1}(e^{-\lambda})\big)e^\lambda, \quad \overline Q_{xy}:= Q_{xy}\frac{e^{-\lambda}}{\eta_x^{-1}(e^{-\lambda})}.
\end{equation*}
It is easily checked that these modified objects satisfy all the properties described above, and one can  work with these objects instead. 
\end{remark}

\iparagraph{Dissipation potential, forces and orthogonality.} As in Example~\eqref{ex:IPFG-diss-pot}, using Definition~\ref{def:F-disspot} the driving force is
\begin{equation*}
  F_{xy}(\rho) = \frac12\log\frac{\kappa_{xy}(\rho_x)}{\kappa_{yx}(\rho_y)}=\frac12\log\frac{\pi_x Q_{xy}\eta_x(\tfrac{\rho_x}{\pi_x})}{\pi_y Q_{yx}\eta_y(\tfrac{\rho_y}{\pi_y})}, \quad \Dom(F)=\P_+(\X).
\end{equation*}
with the dissipation potentials 
\begin{align*}
\Phi^*(\rho,\zeta)&= 2\sumsum_{(x,y)\in\fluxS} \sqrt{\kappa_{xy}(\rho_x) \kappa_{yx}(\rho_y)} \,(\cosh(\zeta_{xy})-1),\\
\Phi(\rho,j)&= 2\sumsum_{(x,y)\in\fluxS} \sqrt{\kappa_{xy}(\rho_x) \kappa_{yx}(\rho_y)} \,\Bigl(\cosh^*\Bigl(\frac{j_{xy}}{2\sqrt{\kappa_{xy}(\rho_x) \kappa_{yx}(\rho_y)}}\Bigr)+1\Bigr).
\end{align*}
Since $\ell\mapsto\cosh(\ell)$ is an even function, using Lemma~\ref{lem:symmetric dissipation} it follows that $\Dom_{\symdiss}(F)=\Dom(F)$, i.e.\ the dissipation potential is symmetric. 

Using Corollary~\ref{def:sym-asym-force} we find 
\begin{equation*}
  F_{xy}^\sym(\rho)=-\Bigl(\frac12d\phi_\rho\tp d\QP(\rho)\Bigr)_{xy} = \frac12\log\frac{\eta_x(\tfrac{\rho_x}{\pi_x})}{\eta_y(\tfrac{\rho_y}{\pi_y})}, \quad
  F_{xy}^\asym(\rho)=F_{xy}(\rho)-F^\sym_{xy}(\rho) = \frac12\log\frac{\pi_xQ_{xy}}{\pi_yQ_{yx}},
\end{equation*}
with $\Dom(\Fsy)=\Dom(\Fasy)=\P_+(\X)$.  Observe that the expressions of $\Fsy$ and $\Fasy$ imply that their domains can be easily extended to $\P_+(\X)$ and $\P(\X)$ respectively; however the theory of Section~\ref{sec:abstract} will not automatically be valid on that extension. Also note that $F^\asym_{xy}=0$ if the particle system satisfies detailed balance with respect to $\pi$.
The orthogonality relations in Proposition~\ref{prop:ortho-rel} apply with (see~\cite{RengerZimmer2021})
\begin{align*}
  \Phi^*_{\zeta^2}(\rho,\zeta^1)&=2\sumsum_{(x,y)\in\fluxS} 
\sqrt{\kappa_{xy}(\rho_x)\kappa_{yx}(\rho_y)}
\, \cosh(\zeta^2_{xy})[ \cosh(\zeta^1_{xy})-1], \\
\theta_\rho(\zeta^1,\zeta^2)
&=2\sumsum_{(x,y)\in\fluxS} \sqrt{\kappa_{xy}(\rho_x) \kappa_{yx}(\rho_y)} \, \sinh(\zeta^1_{xy}) \sinh(\zeta^2_{xy}).
\end{align*}
\iparagraph{Decomposition of the L-function.}  The decompositions in Theorem~\ref{th:FIIRs} hold with the L-functions
\begin{align*}
  \L_{(1-2\lambda)F}(\rho,j) &= \inf_{j^+\in \RR^{\fluxS}_{\geq0}} \sumsum_{(x,y)\in\fluxS} s\big(j^+_{xy} \mid (\pi_xQ_{xy}\eta_x(\tfrac{\rho_x}{\pi_x}))^{1-\lambda}(\pi_y Q_{yx}\eta_y(\tfrac{\rho_y}{\pi_y}))^\lambda \big) \\
&\hspace{3cm}+ s\big(j^+_{xy} - j_{xy} \mid (\pi_xQ_{xy}\eta_x(\tfrac{\rho_x}{\pi_x}))^\lambda (\pi_y Q_{yx}\eta_y(\tfrac{\rho_y}{\pi_y}))^{1-\lambda}  \big),\\
  \L_{F-2\lambda F^\sym}(\rho,j) &= \inf_{j^+\in \RR^{\fluxS}_{\geq0}} \sumsum_{(x,y)\in\fluxS} s\big(j^+_{xy} \mid (\pi_xQ_{xy}\eta_x(\tfrac{\rho_x}{\pi_x}))^{1-\lambda} (\pi_x Q_{xy}\eta_y(\tfrac{\rho_y}{\pi_y}))^\lambda \big) \\
&\hspace{3cm}+ s\big(j^+_{xy} - j_{xy} \mid (\pi_yQ_{yx}\eta_y(\tfrac{\rho_y}{\pi_y}))^{1-\lambda} (\pi_yQ_{yx}\eta_x(\tfrac{\rho_x}{\pi_x}))^\lambda  \big),\\
  \L_{F-2\lambda F^\asym}(\rho,j) &= \inf_{j^+\in \RR^{\fluxS}_{\geq0}}  \sumsum_{(x,y)\in\fluxS} s\big(j^+_{xy} \mid (\pi_xQ_{xy}\eta_x(\tfrac{\rho_x}{\pi_x}))^{1-\lambda}(\pi_yQ_{yx}\eta_x(\tfrac{\rho_x}{\pi_x}))^\lambda\big) \\
  &\hspace{3cm}+ s\big(j^+_{xy} - j_{xy} \mid (\pi_y Q_{yx}\eta_y(\tfrac{\rho_y}{\pi_y}))^{1-\lambda}(\pi_xQ_{xy}\eta_y(\tfrac{\rho_y}{\pi_y}))^\lambda\big),
\end{align*}
and the corresponding Fisher informations
\begin{align*}
  \RF^\lambda_F(\rho)      =-\H\big(\rho,-2\lambda F(\rho)\big)   &= \sumsum_{\substack{x,y\in\X\\x\neq y}} \pi_x Q_{xy}\eta_x(\tfrac{\rho_x}{\pi_x}) - (\pi_x Q_{xy}\eta_x(\tfrac{\rho_x}{\pi_x}))^{1-\lambda}(\pi_y Q_{yx}\eta_y(\tfrac{\rho_y}{\pi_y}))^\lambda,\\
  \RF^\lambda_{\Fsy}(\rho) =-\H\big(\rho,-2\lambda F^\sym(\rho)\big) &=\sumsum_{\substack{x,y\in\X\\x\neq y}} \pi_x Q_{xy}\eta_x(\tfrac{\rho_x}{\pi_x}) - (\pi_x Q_{xy}\eta_x(\tfrac{\rho_x}{\pi_x}))^{1-\lambda}(\pi_x Q_{xy}\eta_y(\tfrac{\rho_y}{\pi_y}))^\lambda,\\
  \RF^\lambda_{\Fasy}(\rho)=-\H\big(\rho,-2\lambda \Fasy(\rho)\big) &=\sumsum_{\substack{x,y\in\X\\x\neq y}} \pi_x Q_{xy}\eta_x(\tfrac{\rho_x}{\pi_x}) - (\pi_x Q_{xy}\eta_x(\tfrac{\rho_x}{\pi_x}))^{1-\lambda}(\pi_y Q_{yx}\eta_x(\tfrac{\rho_x}{\pi_x}))^\lambda.
\end{align*}
In particular, with $\eta_x\equiv\mathrm{id}$, we indeed arrive at the expressions in Example~\ref{ex:IPFG-FIIR}.

With the expressions above the zero-range model satisfies the FIR inequality from Corollary~\ref{corr:FIR-abstract} for $\lambda=\frac12$, which is consistent with~\cite[Cor.~4.3]{RengerZimmer2021} but also holds more generally for $\lambda\in [0,1]$. We 
also mention that the zero-cost flux for the symmetric $\L_{\Fsy}$ satisfies EDI (see Corollary~\ref{cor:symmetric=EDP}), i.e.\ it induces a gradient flow structure. We now turn our attention to its antisymmetric counterpart.

\iparagraph{Zero-cost velocity for antisymmetric L-function.}
As in the IPFG case in Section~\ref{sec:IPFG antisym flow}, we now consider the zero-cost velocity associated to $\L_{\Fasy}$ which for any $x\in\X$ solves the ODE
\begin{equation}\label{eq:zero-range antisym ODE}
\dot\rho_x(t) = \sum_{\substack{y\in\X\\y\neq x}} A_{xy} \sqrt{\pi_x\pi_y\eta_x\big(\tfrac{\rho_x(t)}{\pi_x}\big)\eta_y\big(\tfrac{\rho_y(t)}{\pi_y}\big)}, \quad \text{with} \quad A_{xy}:= Q_{yx} \sqrt{\mfrac{\pi_y}{\pi_x}} - Q_{xy}\sqrt{\mfrac{\pi_x}{\pi_y}}. 
\end{equation}

Note that the corresponding ODE for IPFG~\eqref{eq:IPFG antisymmetric ODE} follows with $\eta_x\equiv\mathrm{id}$.
The geometric arguments of Section~\ref{sec:IPFG antisym flow} cannot be fully repeated, because it is unclear how to transform~\eqref{eq:zero-range antisym ODE} into a linear equation. However, by analogy to that section, we make an educated guess for the energy and the Poisson structure, which is summarised in the following result. We will make use of the following family of functions $g_x:[0,1]\rightarrow\RR$
\begin{equation*}
g_x(a) :=\int_0^a\!\mfrac{1}{\sqrt{\eta_x(\tfrac{b}{\pi_x})}}\,db,
\end{equation*}
for every $x\in\X$. Using these functions we now show that the Conjecture~\ref{conj} holds for the zero-range process. 

\begin{prop}\label{prop:0range-Ham-state}
Assume that $\eta_x$ is such that $g_x$ is well defined for any $x\in\X$. Define the threshold 
\begin{equation}\label{def:0range-thresh}
  \sigma:= \min_{x\in\X} \min_{\substack{\rho\in\RR^{\X}\\\rho_x=0}}\bigl[ 1 - \sum_{\substack{z\in\X\\z\neq x}} g_z(\rho_z) + \lambda_x\bigl(\sum_{\substack{z\in\X\\z\neq x}} \rho_z -1\bigr)\bigr], \quad \ \text{where} \ \lambda_x>0 \ \text{ satisfies } \ \sum_{\substack{z\in\X\\z\neq x}}\pi_z\eta_z^{-1}\big(\mfrac{1}{\lambda_x^2}\big)=1,
\end{equation}
and the energy $\E:\RR^{\X}\rightarrow\RR\cup\{\infty\}$ and the skew-symmetric matrix field $\pJ:\RR^{\X}\rightarrow \RR^{\X\times\X}$  as 
\begin{equation*}
\E(\rho):=1-\sum_{x\in\X} g_x(\rho_x), \quad (\pJ(\rho))_{xy}:= 2\sum_{z\in\X} \sqrt{\pi_x\pi_y\pi_z \eta_{x}\big(\tfrac{\rho_x}{\pi_x}\big)\eta_{y}\big(\tfrac{\rho_y}{\pi_y}\big)\eta_{z}\big(\tfrac{\rho_z}{\pi_z}\big)} \Bigl( \sqrt{\pi_x}A_{yz} - \sqrt{\pi_y} A_{xz}\Bigr),
\end{equation*}
where $A$ is defined in~\eqref{eq:zero-range antisym ODE}. If the energy of initial distribution $\rho^0\in\P(\X)$ for the ODE~\eqref{eq:zero-range antisym ODE}  satisfies $0\leq \E(\rho^0)<\sigma$, then~\eqref{eq:IPFG antisymmetric ODE} has a unique solution and $\dot \rho = \pJ(\rho)\grad\E(\rho)$. 
If the energy of the initial distribution satisfies  $\E(\rho^0)\geq\sigma$, then~\eqref{eq:zero-range antisym ODE} has non-unique, non-energy-conserving solutions.
\end{prop}
\begin{proof}
For any $x\in\X$ we have
\begin{align*}
(\pJ(\rho)\nabla\E(\rho))_x &=\sum_{y\in\X} (\pJ(\rho))_{xy} (\nabla\E(\rho))_y  = \sum_{y,z\in\X} \sqrt{\pi_x\pi_z \eta_{x}\big(\tfrac{\rho_x}{\pi_x}\big)\eta_{z}\big(\tfrac{\rho_z}{\pi_z}\big)} \Bigl(  \pi_y A_{xz}- \sqrt{\pi_x\pi_y}A_{yz} \Bigr)\\
&= \sum_{z\in\X} \sqrt{\pi_x\pi_z \eta_{x}\big(\tfrac{\rho_x}{\pi_x}\big)\eta_{z}\big(\tfrac{\rho_z}{\pi_z}\big)}A_{xz} = \dot\rho_x(t),  
\end{align*}
where the third equality follows since $\sum_y \pi_y=1$ and $(A\tp\sqrt{\pi})_y=0$ for any $y\in\X$. Finally, note that~\eqref{eq:zero-range antisym ODE} has unique solutions if the right hand side is Lipschitz, which follows if $\rho_x>0$, since $\eta_x(0)=0$, for every $x\in\X$. The expression~\eqref{def:0range-thresh} for this threshold follows by solving 
\begin{equation*}
\min \big\{\E(\rho):\rho \in \P(\X), \rho_x=0 \text{ for some } x\in \X\big\} =\min_{x\in\X} \, \min \big\{\E(\rho):\rho \in \P(\X), \, \rho_x=0 \big\},
\end{equation*} 
where $\lambda_x$ in~\eqref{def:0range-thresh} is the Lagrange multiplier for the constraint $\sum_{z\neq x}\rho_z =1$.  The non-uniqueness of solutions follows if $\E(\rho^0)\geq \sigma$ due to non-Lipschitz right-hand side in~\eqref{eq:zero-range antisym ODE}. 
\end{proof}

The equation~\eqref{eq:zero-range antisym ODE} may have an underlying Hamiltonian structure, but while the matrix field $\pJ(\rho)$ proposed here is skew-symmetric, it generally does not satisfy the Jacobi identity.

\subsection{Complex-balanced chemical reaction networks}\label{subsec:reacting particle system}

\iparagraph{Microscopic particle system.} We now describe a particle system that is commonly used to model chemical reactions. For a detailed review of this particle system with motivation and connections to related particle systems see~\cite{AndersonKurtz11}.

Let $\X$ be a finite set of species, $\Reac$ be the finite set of reactions between the species, and let the vectors $\gamma\super{r}\in \RR^\X$ denote the net number of particles of each species that are created/annihilated during a reaction $r\in\Reac$. Furthermore, let $\Reac=\Rf\cup\Rb$ such that each forward reaction $r\in\Rf$ corresponds to a backward reaction $\bw(r)\in\Rb$, meaning that $\gamma\super{\bw(r)} = -\gamma\super{r}$ for all $r\in\Rf$ \footnote{This does not necessarily mean that each forward reaction $\alpha\super{r}\to\alpha\super{\bw(r)}$ corresponds to a backward reaction $\alpha\super{\bw(r)}\to\alpha\super{r}$, which is known in the literature as a reversible network \cite[Def.~2.2]{AndersonCraciunKurtz2010}. The difference between the two notions is clearly seen in the example at the end of our Appendix~\ref{sec:CB equiv HJ}.
\label{ftn:reversible network}}.
The set $\Rf$ will play the role of $\fluxS$ from Example~\ref{ex:indep LDP}.

The microscopic model involves a finite volume $V$ that controls the number of randomly reacting particles in the system. For a fixed $V$, we study the random \emph{concentration} or \emph{empirical measure} $\rho_x\super{V}(t)$, which is the number of particles belonging to species $x\in\X$. Note that the total number of particles may not be conserved here, as opposed to the setting 
of Example~\ref{ex:indep LDP}. 
We also consider the \emph{integrated net reaction flux} for $r\in\Rf$,
\begin{equation*}
  W\super{V}_{r}(t)=\mfrac{1}{V} \#\big\{\text{reactions $r$ occurred in time }(0,t]\big\} - \mfrac{1}{V} \#\big\{\text{reactions $\bw(r)$ occurred in time }(0,t]\big\}.
\end{equation*}

Forward and backward microscopic reactions $r$ take place with given microscopic jump rates $V\k_r\super{V}(\rho\super{V})$ and $V\k_{\bw(r)}\super{V}(\rho\super{V})$ respectively. Typically these jump rates are modelled with combinatoric terms~\eqref{eq:CME rates}, see also~\cite{AndersonKurtz11}. Since our framework is purely macroscopic, the precise expressions for the microscopic jump rates are not relevant; the only crucial point is that both converge sufficiently strongly to macroscopic reaction rates $\kr(\rho)$ and $\kbr(\rho)$.
The pair $(\rho\super{V}(t), W\super{V}(t))$ is a Markov process on $\RR^{\X}\times\RR^{\Rf}$ with generator 
\begin{align*}
  &(\Q\super{V}f)(\rho,w)= \\
  &\quad V\sum_{r\in\Rf} \kappa_r\super{V}(\rho) \bigl[ f(\rho+\tfrac1V \gamma\super{r}, w+\tfrac1V \mathds1_r) - f(\rho,w)\bigr] + \kappa_{\bw(r)}\super{V}(\rho)\bigl[ f(\rho+\tfrac1V \gamma\super{\bw(r)}, w+\tfrac1V \mathds1_{\bw(r)}) - f(\rho,w)\bigr].
\end{align*}

Using the matrix notation $\Gamma:=[\gamma\super{r}]_{r\in\Rf} \in \RR^{\X\times\Rf}$, in the limit $V\rightarrow\infty$ the pair $(\rho\super{V}, W\super{V})$ converges to the solution of (see~\cite{Kurtz1970a} and \cite[Sec.~3.1]{RengerZimmer2021})
\begin{equation}
\begin{cases}
  \dot w_r(t) = \kappa_r(\rho(t)) - \kappa_{\bw(r)}(\rho(t)), \  \ & r \in\Rf\\
  \dot \rho_x(t) = (\Gamma\dot w(t))_x, \ \ & x\in\X.
\end{cases}
\label{eq:reac ODE}
\end{equation}

The Markov process $(\rho\super{V}(t), W\super{V}(t))$ satisfies a large-deviation principle~\eqref{eq:flux LDP} where $\L,\H$ are now given by (see~\cite[Thm.~1.1]{PattersonRenger2019} and \cite[Cor.~3.1]{RengerZimmer2021})
\begin{align*}
  \L(\rho,j)&=  \inf_{j^+\in \RR^{\Rf}_{\geq0}} \sum_{r\in\Rf} s( j^+_{r}\mid \kappa_{r}(\rho)) + s( j^+_{r} - j_r \mid \kappa_{\bw(r)}(\rho) ),  \\
  \H(\rho,\zeta)& = \sum_{r\in\Rf} \kr(\rho)(e^{\zeta_r}-1)+ \kbr(\rho)(e^{-\zeta_r}-1),
\end{align*}
and $s(\cdot\mid\cdot)$ is defined in~\eqref{def:entropy-s}. As in the IPFG and zero-range models, the infimum over one-way fluxes $j^+$ can be derived using the contraction principle.

We mention that at this level of generality one can already derive many interesting MFT properties, see \cite{RengerZimmer2021}. After all, the IPFG and zero-range models fall within this class. However, in order to apply our framework and obtain explicit results, the quasipotential needs to be known. To this aim we make two crucial assumptions.

First, the system satisfies \emph{mass-action kinetics} i.e.\ there exists \emph{stoichiometric vectors} or \emph{complexes} $\alpha\super{r}\in\RR^\X_{\geq0}$ (encoding the number of reactants involved) and \emph{reaction constants} $\cf>0$ for each $r\in\Reac$ such that
\begin{equation*}
  \gamma\super{r} = \alpha\super{\bw(r)} -\alpha\super{r},\qquad
  \gamma\super{\bw(r)} = \alpha\super{r} -\alpha\super{\bw(r)},  
\end{equation*}  
and the forward and backward rates satisfy, setting $\rho^{\alpha\super{r}}:=\prod_{x\in\X} \rho_x^{\alpha\super{r}_x}$,
\begin{align}\label{ass:mass-action}
  \kr(\rho)=\cf\rho^{\alpha\super{r}}, \qquad \forall r\in\Reac.
\end{align}

Second, we assume that the system is in \emph{complex balance}~\cite[Sec.~3.2]{AndersonCraciunKurtz2010} with respect to some $\pi \in \RR^\X_{>0}$, i.e.
\begin{equation}\label{def:complex-bal}
  \forall  \ \psi\in \RR^\CC : \ \sum_{r\in\Rf}(\cf \pi^{\alpha\super{r}} - \cb\pi^{\alpha\super{\bw(r)}})(\psi_{\alpha\super{r}}-\psi_{\alpha\super{\bw(r)}})=0,
\end{equation}
where $\CC:=\{\alpha\super{r}:r\in\Reac\}$ signifies the set of \emph{complexes}. This immediately implies that $\pi$ is a steady state of the macroscopic dynamics~\eqref{eq:reac ODE}. Observe that complex balance w.r.t. $\pi$ is a macroscopic notion, whereas detailed balance of the Markov process w.r.t. $\Pi\super{V}$ is a microscopic notion. However, for reversible networks (see footnote~\ref{ftn:reversible network}) microscopic detailed balance corresponds to the macroscopic notion of detailed balance $\cf \pi^{\alpha\super{r}} = \cb\pi^{\alpha\super{\bw(r)}}$ \cite[Th.~4.5]{AndersonCraciunKurtz2010}, which is clearly a stronger than complex balance. Most importantly, whereas detailed balance corresponds to purely dissipative dynamics~\cite{MielkePeletierRenger14}, complex balance allows for non-dissipative effects. 



\iparagraph{State-flux triple and L-function.} 
Fix a reference or initial concentration $\rho^0\in\RR^\X_{\geq0}$ and recall the matrix notation $\Gamma w =  \sum_{r\in\Rf} \gamma\super{r} w_r$. 
The state space is the flat manifold of concentrations that can be produced from $\rho^0$ via reactions, with corresponding local (co)tangent spaces:
\begin{align}
  \Z=\rho^0+\Ran(\Gamma),
  &&
  T_\rho\Z = \Ran(\Gamma),
  && \qquad T_\rho^*\Z=\RR^\X/\Ker(\Gamma\tp).
\label{eq:chem state space}
\end{align}
As in the case of IPFG and zero-range, we include negative concentrations to simplify the geometric setting; this set $\Z$ is known in the literature as the stoichiometric compatibility class, whereas the subset of $\Z$ of coordinate-wise non-negative concentrations is called the stoichiometric simplex \footnote{Under the complex balance assumption the steady state $\pi$ is unique and stable within such simplex~\cite[Thm.~3.2]{AndersonCraciunKurtz2010}.}. Moreover, as in the previous examples, $T_\rho\Z$ restricts the directions of $\RR^\X$ in which one can differentiate, and $T_\rho^*\Z$ appears as a quotient space. Indeed the Euclidean inner product between tangents $u=\Gamma j\in \Ran(\Gamma)$ and cotangents $\xi\in\RR^\X/\Ker(\Gamma\tp)$ is again invariant under addition of vectors $\nu\in\Ker(\Gamma\tp)$, since $_{T_\rho^*\Z}\langle\xi+c\nu,u\rangle_{T_\rho\Z}=(\xi+c\nu)\cdot \Gamma j)=\xi\cdot u$. The space $\Ker(\Gamma\tp)$ encode the quantities (usually numbers of atoms) that are conserved under the reactions.

The flux space and its associated tangent and cotangent spaces are simply the Euclidean space
\begin{align*}
  \W=T_\rho\W=T_\rho^*\W=\RR^{\Rf},
\end{align*}
and the continuity map $\phi:\W\rightarrow\Z$ and its differential are 
\begin{align*}
  \phi(w)= \rho^0 + \Gamma w,
  &&
  d\phi_\rho  = \Gamma,
  &&
  d\phi\tp_\rho=\Gamma\tp. 
\end{align*}
Note that with this setup, $\phi$ is indeed surjective. Again, $\L$ is convex and lower semicontinuous since $\L$ is its own convex bidual.

\iparagraph{Quasipotential.} The quasipotential is again the relative entropy with respect to the invariant measure,
\begin{align}
  \QP(\rho)=
  \begin{cases}
    \sum\limits_{x\in\X} s(\rho_x\mid\pi_x), & \rho\in\Z,\\
    \infty,                                  &\text{otherwise.}
   \end{cases}
\label{eq:chem QP}
\end{align}
Similar to Example~\ref{eq:chem QP}, as a function on the state manifold $\Z$, this quasipotential is differentiable on $\Z_+:=\{\rho\in\Z:\rho>0 \text{ (coordinate-wise)}\}$, with Gateaux derivative $d\QP(\rho)=\{(\log(\rho_x/\pi_x))_{x\in\X}+\xi: \xi\in \Ker(\Gamma\tp)\}$.

Recall the relation between the quasipotential and the large-deviation rate functional for the invariant measure of the microscopic system from Theorem~\ref{th:LDP QP}. Whereas in the IPFG model this relative entropy appears as the large-deviation rate functional for independent particles by Sanov's Theorem, in the complex balance case this is the rate functional of the explicitly known invariant measure of the microscopic particle system~\cite[Thm.~4.1]{AndersonCraciunKurtz2010}. As in the previous examples, it can also be checked purely macroscopically that this is the correct quasipotential satisfying~\eqref{eq:invariant I_0}. In fact, it turns out that~\eqref{eq:invariant I_0} is \emph{equivalent} to complex balance; both directions of the equivalence will be shown in Theorem~\ref{th:CB equiv HJ} in Appendix~\ref{sec:CB equiv HJ}. 

\begin{remark} As mentioned in Subsections~\ref{subsec:results} and \ref{sec:QuasiLDP}, the quasipotential $\V$ is always a Lyapunov function along the zero-cost dynamics~\eqref{eq:coupled ev eq}. For the case of chemical reactions this was worked out explicitly in \cite{ACGF2015}.
\end{remark}

\iparagraph{Dissipation potential, forces and orthogonality.} The driving force is
\begin{equation*}
  F_r(\rho) = \frac12 \log\frac{\kr(\rho)}{\kbr(\rho)}=\frac12\log\Bigl(\frac{\cf}{\cb} \rho^{-\gamma\super{r}}\Bigr),
  \qquad
  \Dom(F)=\Z_+,
\end{equation*}
recalling that $\kr(\rho)=\cf \rho^{\alpha\super{r}}$ and $\Z_+$ denote the positive concentrations in $\Z$. The dissipation potentials are
\begin{align*}
  \Phi^*(\rho,\zeta)&= 2\sum\limits_{r\in\Rf} \sqrt{\kr(\rho)\kbr(\rho)} \ (\cosh(\zeta_r)-1), \\
  \Phi(\rho,j)&= 2\sum\limits_{r\in\Rf} \sqrt{\kr(\rho)\kbr(\rho)} \ \Bigl(\frac{\cosh^*(j_r)}{2\sqrt{\kr(\rho)\kbr(\rho)}}+1\Bigr).
\end{align*}
Note that $\Dom_{\symdiss}(F)=\Dom(F)$, i.e.\ the dissipation potential is symmetric. 

Following Corollary~\ref{def:sym-asym-force}, the symmetric and antisymmetric forces are
\begin{align*}
  F_{r}^\sym(\rho)=-\Bigl(\frac12d\phi_\rho\tp d\QP(\rho)\Bigr)_{r} = -\frac12\log\Bigl(\mfrac{\rho}{\pi}\Bigr)^{\gamma\super{r}}, \ 
  F_{r}^\asym(\rho)=F_{r}(\rho)-F^\sym_{r}(\rho) = \frac12\log\Bigl( \frac{\cf }{\cb}\pi^{-\gamma\super{r}} \Bigr),
\end{align*}
with $\Dom(\Fsy)=\Dom(\Fasy)=\Dom(F)=\Z_+$. The orthogonality relations in Proposition~\ref{prop:ortho-rel} apply with
\begin{align*}
  \Phi^*_{\zeta^2}(\rho,\zeta^1)&=2\sum_{r\in\Rf} 
    \sqrt{\kr(\rho)\kbr(\rho)} 
    \, \cosh(\zeta^2_{r})[ \cosh(\zeta^1_{r})-1], \\ 
  \theta_\rho(\zeta^1,\zeta^2)
    &=2\sum_{r\in\Rf} \sqrt{\kr(\rho)\kbr(\rho)} \, \sinh(\zeta^1{r}) \sinh(\zeta^2_{r}).
\end{align*}
This notion of generalised orthogonality is consistent with the derivations in \cite{RengerZimmer2021}.

\iparagraph{Decomposition of the L-function.} The decompositions in Theorem~\ref{th:FIIRs} hold with the L-functions 
\begin{align*}
\L_{(1-2\lambda)F}(\rho,j) &= \inf_{j^+\in\RR^\Rf_{\geq0}} \sum_{r\in\Rf} s\big(j^+_{r} \mid (\kr(\rho))^{1-\lambda}(\kbr(\rho))^\lambda \big) 
+ s\big(j^+_{r} - j_{r} \mid (\kr(\rho))^{\lambda}(\kbr(\rho))^{1-\lambda}  \big),\\
\L_{F-2\lambda F^\sym}(\rho,j) &= \inf_{j^+\in\RR^\Rf_{\geq0}} \sum_{r\in\Rf} s\Bigl(j^+_{r} \mid \kr(\rho)\bigl(\mfrac{\rho}{\pi}\bigr)^{\lambda\gamma\super{r}} \Bigr)+s\Bigl(j^+_{r} - j_r\mid \kbr(\rho)\bigl(\mfrac{\rho}{\pi}\bigr)^{-\lambda\gamma\super{r}} \Bigr), \\
 \L_{F-2\lambda F^\asym}(\rho,j)  &= \inf_{j^+\in\RR^\Rf_{\geq0}} \sum_{r\in\Rf} s\Bigl(j^+_{r} \mid (\kr(\rho))^{1-\lambda} (\kbr(\rho))^\lambda \bigl(\mfrac{\rho}{\pi}\bigr)^{-\lambda\gamma\super{r}} \Bigr)\\
 &\hspace{2cm}  +s\Bigl(j^+_{r} - j_r\mid (\kr(\rho))^{\lambda} (\kbr(\rho))^{1-\lambda} \bigl(\mfrac{\rho}{\pi}\bigr)^{\lambda\gamma\super{r}} \Bigr), 
\end{align*}
with the corresponding Fisher informations 
\begin{align*}
\RF^\lambda_F(\rho)&=-\H\big(\rho,-2\lambda F(\rho)\big) =\sum_{r\in\Reac} \kr(\rho)  -  (\kr(\rho))^{1-\lambda} (\kbr(\rho))^{\lambda}  ,\\
\RF^\lambda_{\Fsy}(\rho)&=-\H\big(\rho,-2\lambda F^\sym(\rho)\big)=\sum_{r\in\Reac} \kr(\rho) - \kr(\rho) \bigl(\mfrac{\rho}{\pi}\bigr)^{\lambda\gamma\super{r}} ,\\
  \RF^\lambda_{\Fasy}(\rho)&=-\H\big(\rho,-2\lambda F^\asym(\rho)\big) =\sum_{r\in\Reac} \kr(\rho) - (\kr(\rho))^{1-\lambda}(\kbr(\rho))^\lambda \bigl( \mfrac{\pi}{\rho}\bigr)^{\lambda\gamma\super{r}}.
\end{align*}
The zero-cost flux for $\L_{\Fsy}$ is related to a gradient flow by Corollary~\ref{cor:symmetric=EDP}; this has been discussed in~\cite[Cor.~4.8]{Renger2018a}. As opposed to IPFG and zero-range examples, the construction of a Poisson structure for $\L_{\Fasy}$ is difficult in the chemical reaction setting due to the non-locality of the jump rates and the interplay with the stoichiometric vectors, and remains an open question.

\subsection{Lattice gases}\label{sec:LatticeGas}

In this section we focus on the typical setting of MFT~\cite{BDSGJLL2015MFT}, namely discrete state-space particle systems whose hydrodynamic limit is the following drift-diffusion equation on the torus $\mathbb T^d$:
\begin{align}
  \dot\rho(t) &= - \Cdiv j(t),\notag\\
  j(t)&=j^0\big(\rho(t)\big), \hspace{1cm}\text{with } j^0(\rho):=-\Cnabla \rho - \chi(\rho) (\Cnabla \Pot + A). 
\label{eq:hydro-lim}
\end{align}
As before $\rho\in \P(\mathbb T^d)$ is the limiting density of the particle system, but now $\Cnabla,\Cdiv$ denote the continuous differential operators in $\RR^d$. We assume that the strictly positive potential $\Pot\in C^\infty(\mathbb T^d;(0,\infty))$, covector field $A\in C^\infty(\mathbb T^d;\mathbb R^d)$ and the `mobility' $\chi \in C^\infty(\RR;[0,\infty))$ are smooth, and that furthermore \footnote{In order to make sure that the rate functional $\int_0^T\!\hat\L(\rho(t),\dot\rho(t))\,dt=\infty$ whenever negative concentrations are reached, one should in fact require $\chi(a)\equiv0$ for all $a\notin\lbrack0,1\rbrack$, and assume that $\chi$ is continuous and smooth away from its zeros.},
\begin{align}
  \Cdiv A &= 0 \notag\\
  \nabla U\cdot A&=0\notag\\
  \mfrac{d^2}{da^2} h(a) &= \mfrac{1}{\chi(a)}
  \label{eq:mobility ass}
\end{align}
for some $h: [0,\infty) \rightarrow [0,\infty)$,

Most results about this class of models are well known; we present them here to show that our abstract theory is consistent with `classical' MFT.

\iparagraph{Microscopic particle system.}  Although the macroscopic framework works for general mobilities, we only describe two standard microscopic particle systems that give rise to different mobilities. For independent random walkers $\chi(a) = a$, $h(a) = a\log a - a + 1$ and for the simple-exclusion process $\chi(a) = a(1-a)$, $h(a)=\ a \log a + (1-a) \log (1-a)$. Since these two particle systems with limit~\eqref{eq:hydro-lim} have been extensively studied in the literature, we only present the essential features here. 

For both particle systems, the particles can jump to neighbouring sites on the lattice $\TT^d\cap (\frac1n\ZZ)^d$. In order to pass to the hydrodynamic limit~\eqref{eq:hydro-lim} and derive the corresponding large deviations, the state space will be embedded in the continuous torus. The first particle system consists of independent random walkers with drift. For any $n\in\NN$, the corresponding empirical measure-flux pair $(\rho\super{n}(t),W\super{n}(t))$ is a Markov process in $\P(\TT^d)\times\M(\TT^d;\RR^d)$ with generator (see~\cite{Renger2018b})
\begin{align*}
  (\Q\super{n} f)(\rho,w) &= n^2\sum_{\substack{\tau\in\ZZ^d\\\lvert \tau\rvert=1}} \int_{\TT^d}\!n^d\rho(dx) e^{-\left[\frac12 U(x+\frac1n \tau)-\frac12 U(x)+\frac1{2n} A(x)\cdot \tau\right]}\\[-0.4cm]
 &\hspace{3cm} \times\Big\lbrack f\big(\rho - \tfrac1{n^d}\delta_{x} + \tfrac1{n^d}\delta_{x+\frac1n \tau}, w + \tfrac{1}{n^{d+1}} \tau\delta_{x+\frac1{2n}\tau}\big) - f(\rho,w)\Big\rbrack.
\end{align*}
This system can also be derived as the spatial discretisation of interacting stochastic differential equations, although in such continuous-space setting it becomes less straight-forward how to define particle fluxes.

The second particle system is the weakly asymmetric simple exclusion process (WASEP) which has been extensively studied in the MFT literature (see for instance~\cite{BdSGjLL07,BDSGJLL2015MFT}). In this case the Markov process $(\rho\super{n}(t),W\super{n}(t))$ has generator
\begin{align*}
  (\Q\super{n} f)(\rho,w) &=n^2\sum_{\substack{\tau\in\ZZ^d\\\lvert \tau\rvert=1}} \int_{\TT^d}\! n^d\rho(dx) \big(1-n^d \rho(\{x+\tfrac1n \tau\})\big) e^{-\left[\frac12 U(x+\frac1n \tau)-\frac12 U(x)+\frac1{2n} A(x)\cdot \tau\right]}\\[-0.4cm]
&\hspace{3cm} \times \Big\lbrack f\big(\rho - \tfrac1{n^d}\delta_{x} + \tfrac1{n^d}\delta_{x+\frac1n \tau}, w + \tfrac1{n^{d+1}} \tau\delta_{x+\frac1{2n} \tau}\big) - f(\rho,w)\Big\rbrack.
\end{align*}

Observe that in both generators, the flux $w$ has a different scaling than the particle density $\rho$. This is required to ensure that the discrete-space, finite-$n$ continuity equation converges to the continuous-space continuity equation with differential operator $-\div$. 

Letting $n\rightarrow\infty$ we arrive at the hydrodynamic limit~\eqref{eq:hydro-lim}
with $\chi(a):=a$ for the first particle system and $\chi(a):=a(1-a)$ for the second particle system. The corresponding large-deviation cost function and its dual are

\begin{equation}\label{eq:lattice-LD-lag}
\begin{aligned}
  \H(\rho,\zeta)&=\lVert \zeta\rVert^2_{L^2(\chi(\rho))} + \int_{\T^d}\!\zeta(x) j^0(\rho)(x)\,dx,\\
  \L(\rho,j)    &= \begin{cases}\mfrac14\lVert j -j^0(\rho)\rVert^2_{L^2(1/\chi(\rho))}, &  \  \text{if } j-j^0(\rho)\in  L^2(1/\chi(\rho))\\
  +\infty, & \ \text{else.}
  \end{cases}		
\end{aligned}
\end{equation}
Here $L^2(\chi(\rho))$ is the $\chi(\rho)$-weighted $L^2$-space on $\T^d$ with  $\|f\|_{L^2(\chi(\rho))}^2:=\int_{\T^d} f(x)^2 \chi(\rho(x))dx$ and $\|\cdot\|_{L^2(1/\chi(\rho))}$ is the dual norm to $\|\cdot\|_{L^2(\chi(\rho))}$. Note that $\L$ is constructed by taking the convex dual of $\H$ which is defined in terms of $\|\cdot\|_{L^2(\chi(\rho))}$.
See~\cite[Sec.~5]{Renger2018b} for the large-deviations of the random walkers (with $A=0$),  ~\cite[Chap.~10]{KipnisLandim1999} for exclusion process without fluxes, and~\cite[Thm.~2.1]{BdSGjLL07} for exclusion process with fluxes (with $A=0$).


\iparagraph{State-flux triple and L-function.} Apart from the fact that the state space is infinite-dimensional, the lattice gas example differs from the previous examples in a number of ways. First of all, in this setting, one actually has a microscopic state-flux triple $(\Z^n,\W^n,\phi^n)$ that converges to the macroscopic one $(\Z,\W,\phi)$ in a suitable sense, see for example \cite[Sec.~5]{Renger2018a}. For simplicity we only present the macroscopic structure. The second difference is that the cost~\eqref{eq:lattice-LD-lag} happens to be a quadratic functional, which induces a norm on the cotangent space. However as in the finite-dimensional examples, we regard such induced geometry to be a posteriori; one first needs a basic geometric setup in order to derive the dissipation potentials. Therefore we shall work with the following setting, and discuss the geometry induced by \eqref{eq:lattice-LD-lag} in Remark~\ref{rem:lattice gas induced geometry}.

For the state space we choose, analogous to \eqref{eq:IPFG state space} and \eqref{eq:chem state space},
\begin{align*}
	\Z:=\big\{{ \textstyle\rho\in L^1(\TT^d):\int_{\TT^d}\!\rho(x)\,dx=\int_{\TT^d}\!\rho^0(x)\,dx }\big\},\\
  T_\rho\Z=\big\{{ \textstyle u\in L^1(\TT^d):\int_{\TT^d}\! u(x)\,dx=0 }\big\},&&
  T_\rho^*\Z=\big\{{ \textstyle \{ \xi+c: c\in\RR\}:\xi\in L^\infty(\TT^d) }\big\}.  
\end{align*}
For $\chi(a)=a$ one might be tempted to choose $\Z$ as the space of signed measures with total mass $\int\!\rho^0\,dx$, but then the quasipotential $\V$ will fail to be differentiable.

For the flux space we choose the flat Banach manifold (see~\cite[Theorem 3.12]{AdamsFournier03})
\begin{align*}
	\W=T_\rho\W=W^{1,1}(\T^d;\RR^d), &&T_\rho^*\W=W^{-1,\infty}(\T^d),
\end{align*}
and for the continuity operator the usual one:
\begin{align}
  \phi\lbrack w\rbrack:=\rho^0-\div w && \text{with differential} && d\phi_\rho j=-\div j && \text{and its adjoint} && d\phi_\rho\tp\xi=\grad\xi.
\label{eq:lattice gas continuity operator}
\end{align}
As a validity check, this setup indeed satisfies $\phi:\W\to\Z$, $d\phi_\rho:T_\rho\W\to T_\rho\Z$ and $d\phi_\rho\tp:T_\rho^*\Z\to T_\rho^*\W$.
Finally, $\L$ is clearly convex and lower semicontinuous in the $L^2_{1/\chi(\rho)}$-norm.


\begin{remark} A posteriori we could also choose the state-flux triple implied by the large deviations~\eqref{eq:lattice-LD-lag}. Then $\Z=(\P(\mathbb T^d),\Wasser_2)$, the space of probability measures on the (compact) torus, endowed with the Wasserstein-2 metric $\Wasser_2$. For any $\rho\in\Z$, 
the corresponding cotangent and tangent spaces and the associated norms are 
\begin{align*}
  T_\rho^*\Z:=\overline{\{C^\infty(\mathbb T^d) \}}^{\|\cdot\|_{1,\chi(\rho)}},
  &&
  T_\rho\Z = \Big\{-\Cdiv (\chi(\rho) h) \text{ (in distr. sense)} :   h \in   \overline{\{\Cnabla\varphi: \varphi\in C^\infty(\mathbb T^d) \}}^{\|\cdot\|_{L^2(\chi(\rho))}} \Big\} .
\end{align*}
with the standard (semi)norms from Wasserstein-2 geometry~\cite[Sec. 3.4.2]{PeletierLN2014}  
\begin{equation*}
  \|\xi\|_{1,\chi(\rho)}^2:=\|\nabla\xi\|^2_{L^2(\chi(\rho))}, \quad  \ \|u\|^2_{-1,\chi(\rho)}:=\inf_{\substack{j\in T_\rho\W\\ \, u= -\Cdiv j}} \|j\|^2_{L^2(1/\chi(\rho))}.
\end{equation*}
The induced flux space is the metric space 
\begin{align*}
  &\W = \left\{ w\in \M(\TT^d;\RR^d): \rho^0-\Cdiv w \text{ (in distr. sense)} \in \P(\T^d)\right\},\\
  &d^2_\W(w_1,w_2):=\inf_{\substack{\hat w:[0,1]\rightarrow\W \\ \hat w(0)=w_1,\hat w(1)=w_2}} \int_0^1\! \|\hat w(t)\|^2_{L^2(1/\chi(\rho_0-\Cdiv \hat w(t)))}\,dt,\\
  &T_\rho^*\W=L^2(\chi(\rho)), \qquad\qquad T_\rho\W=L^2(1/\chi(\rho)).
\end{align*}
And the continuity operator is again \eqref{eq:lattice gas continuity operator}. This setup is slightly different from the standard Wasserstein geometry, where by convention the fluxes are defined so as to satisfy $\dot\rho=\div(\rho\, j)$, while in our context the fluxes satisfy $\dot\rho = \div j$.

However, this induced state-flux triple is formal, as $\Z$ and $\W$ are not Banach manifolds, and differentiability of the quasipotential $\QP$ becomes less straightforward. We therefore work with the simpler triple described above.
\label{rem:lattice gas induced geometry}
\end{remark}

\iparagraph{Quasipotential.} The quasipotential $\QP:\Z\rightarrow\RR$ is defined as, recalling~\eqref{eq:mobility ass},
\begin{equation*}
  \QP(\rho)=\int_{\mathbb T^d}\!\bigl[ h(\rho(x))+\Pot(x)\rho(x)\bigr]\,dx,
\end{equation*}
Its Gateaux derivative in $\Z$ is simply, recalling \eqref{def:Domain-Fsym}, 
\begin{equation*}
  d\QP(\rho) = h^\prime( \rho) +\Pot  \text{(modulo constants )}, \qquad\text{for }
  \rho \in \Dom(F^\sym)=\{\rho\in\Z:h^\prime(\rho)\in L^\infty(\T^d)\}.
\end{equation*}
It is easy to verify that $\H(\rho,d\phi\tp_\rho d\QP(\rho))=0$ and therefore $\QP$ is indeed a quasipotential in the sense of Definition~\ref{def:I0}. In the case $\chi(a)=a$, $\QP$ is the relative entropy with respect to the Gibbs-Boltzmann measure $\mu(dx)=Z^{-1}e^{-\Pot(x)}\,dx$. 
 
\iparagraph{Dissipation potential, forces and orthogonality.}
Using Definition~\eqref{def:F-disspot} the driving force is 
\begin{equation*}
  F(\rho)=  \mfrac12 (\chi(\rho))^{-1} j^0(\rho) , \ \ \Dom(F)=\big\{\rho\in\Z:  \ \chi(\rho(x))>0 \text{ almost everywhere}  \big\}.
\end{equation*}
The dissipation potential and its dual are 
\begin{equation*}
  \Phi^*(\rho,\zeta)=\|\zeta\|^2_{L^2(\chi(\rho))}+\langle\zeta,j^0(\rho)-2\chi(\rho)F(\rho)\rangle = \|\zeta\|^2_{L^2(\chi(\rho))}, \ \ \Phi(\rho,j)=\frac14\|j\|^2_{L^2(1/\chi(\rho))}.
\end{equation*}
Observe that $\Dom_{\symdiss}(F)=\Dom(F)$, i.e.\ the dissipation potential is symmetric. 
Following Corollary~\ref{def:sym-asym-force}, the symmetric and antisymmetric forces are
\begin{align*}
  \Fsy(\rho)= -\mfrac12 d\phi_\rho\tp d\QP(\rho) =  -\mfrac12 \big[ (\chi(\rho))^{-1}\Cnabla \rho + \Cnabla\Pot \big],  \ \ 
  \Fasy(\rho) = F(\rho)-F^\mathrm{sym}(\rho) =  -\mfrac12 A.
\end{align*}
Indeed the antisymmetric force $\Fasy$ is again independent of $\rho$.

The generalised orthogonality relations in Proposition~\ref{prop:ortho-rel} apply with
\begin{equation*}
  \Phi^*_{\zeta^2}(\rho,\zeta^1)= \|\zeta^1\|^2_{L^2(\chi(\rho))}, \ \ 
\theta_\rho(\zeta^1,\zeta^2)
=2(\zeta^1,\zeta^2)_{L^2(\chi(\rho))},
\end{equation*}
where $(\cdot,\cdot)_{L^2(\chi(\rho))}$ is the $\chi(\rho)$-weighted $L^2$ norm. This shows that for quadratic dissipation potentials, the generalised expansion of Proposition~\ref{prop:ortho-rel} indeed collapses to the usual expansion of squares, i.e.:
\begin{align*}
  \Phi^*(\rho,\zeta^1+\zeta^2) &= \lVert\zeta^1+\zeta^2\rVert^2_{L^2(\chi(\rho))} 
  = \lVert\zeta^1\rVert^2_{L^2(\chi(\rho))} + 2(\zeta^1,\zeta^2)_{L^2(\chi(\rho))} + \lVert\zeta^2\rVert^2_{L^2(\chi(\rho))}\\
&  =\Phi^*(\rho,\zeta^1) + \theta_\rho(\zeta^2,\zeta^1) + \Phi^*_{\zeta^1}(\rho,\zeta^2).
\end{align*}

\iparagraph{Decomposition of the L-function.} The decompositions in Theorem~\ref{th:FIIRs} hold with the L-functions 
\begin{align}
\L_{2\lambda F}(\rho,j) &= \mfrac14 \|j- 4\lambda \chi(\rho) F(\rho)
\|^2_{L^2(1/\chi(\rho))},\nonumber \\
\L_{F-2\lambda \Fsy}(\rho,j) &= \mfrac14 \|j-2\chi(\rho)\Fasy-2(1-2\lambda)\chi(\rho)\Fsy(\rho)\|^2_{L^2(1/\chi(\rho))},\label{def:Latt-asymL} \\
\L_{F-2\lambda \Fasy}(\rho,j) &= \mfrac14 \|j-2(1-2\lambda)\chi(\rho)\Fasy-2\chi(\rho)\Fsy(\rho)\|^2_{L^2(1/\chi(\rho))}, \label{def:Latt-symL}
\end{align}
and the corresponding Fisher informations
\begin{align*}
\RF^\lambda_{F}(\rho) &= \H(\rho,-2\lambda F(\rho)) = \lambda(1-\lambda)\norm{-2F(\rho)}^2_{L^2(\chi(\rho))},\\
\RF^\lambda_{\Fsy}(\rho) &= \H(\rho,-2\lambda \Fsy(\rho)) = \lambda(1-\lambda)\norm{-2\Fsy(\rho)}^2_{L^2(\chi(\rho))},\\
\RF^\lambda_{\Fasy}(\rho) &= \H(\rho,-2\lambda \Fasy) = \lambda(1-\lambda)\norm{-2\Fasy}^2_{L^2(\chi(\rho))}.
\end{align*}
The positivity of these Fisher informations is obvious from the definition. In this setting, the  decompositions  in Theorem~\ref{th:FIIRs} can be derived simply by expanding the squares in the the L-function. 

Repeating the calculations in Corollary~\ref{corr:FIR-abstract} for $\chi(a)=a$, we arrive at the local FIR equality for diffusion processes (with $u$ as a placeholder for $\dot\rho$)~\cite[Eq.~(14)]{HilderPeletierSharmaTse20}
\begin{equation*}
\langle d\,\RelEnt(\rho|\mu),u\rangle + \left\|\nabla\log\mfrac{\rho}{\mu}\right\|_{L^2(\rho)} \leq \hat\L(\rho,j),
\end{equation*}
where the contracted L-function $\hat\L$ is defined in~\eqref{def:FIIR-FIR-contract}, the relative entropy with respect to $\mu$ is defined as $\RelEnt(\cdot|\mu):=\QP(\cdot)$.   

We now briefly comment on the symmetric and antisymmetric L-functions. 
Substituting $\lambda=\tfrac12$ in~\eqref{def:Latt-symL} and expanding the square we find   
\begin{align*}
\L_{\Fsy}(\rho,j) &=  \mfrac14 \|j\|^2_{L^2(1/\chi(\rho))} + \mfrac14\|-2\chi(\rho) \Fsy(\rho)\|^2_{L^2(1/\chi(\rho))} - \mfrac12 \langle j , -2\Fsy(\rho) \rangle\\
& =  \mfrac14 \|j\|^2_{L^2(1/\chi(\rho))} + \mfrac14\|\nabla d\QP(\rho)\|^2_{L^2(\chi(\rho))} - \mfrac12 \langle \Cdiv j , d\QP(\rho) \rangle,
\end{align*}
where we have used $-2\Fasy(\rho)=\Cnabla d\QP(\rho)$ and the definition of $\|\cdot\|_{-1,\chi(\rho)}$. Using this decomposition of $\L_{\Fsy}$, the contracted symmetric L-function 
\begin{equation*}
\hat\L_{\Fsy}(\rho,u):= \inf_{\substack{j\in T_{\rho}\W: \, u=-\Cdiv j}} \L_{\Fsy}(\rho,j), 
\end{equation*}
admits the decomposition
\begin{equation}\label{eq:latt-EDI}
\hat \L_{\Fsy}(\rho,u)= \hat\Psi(\rho,u) + \hat\Psi^*(\rho, - \tfrac12 d\QP(\rho))  + \mfrac12 \langle  d\QP(\rho),u\rangle,
\end{equation}
where the contracted dissipation potential $\hat\Psi(\rho,u) =\mfrac14 \|u\|^2_{-1,\chi(\rho)}$ and its dual $\hat\Psi^*(\rho,s)=\|s\|^2_{1,\chi(\rho)}$ (recall abstract definition in~\eqref{def:FIIR-FIR-contract-diss}).
The decomposition~\eqref{eq:latt-EDI} is the standard Wasserstein-based EDI for the drift-diffusion equation~\eqref{eq:hydro-lim} (see for instance~\cite[Sec.~4.2]{MielkePeletierRenger14}).

Similarly, the purely antisymmetric L-function and its contraction read 
\begin{equation*}
  \L_{\Fasy}(\rho,j)     =  \mfrac14 \|j+\chi(\rho)A\|^2_{L^2(1/\chi(\rho))},\quad
  \hat\L_{\Fasy}(\rho,u) =  \mfrac14 \|u+\Cdiv(\chi(\rho)A)\|^2_{-1,\chi(\rho)},
\end{equation*}
with zero-cost velocity $u^0(\rho)=-\Cdiv(\chi(\rho)A)=-\Cnabla(\chi(\rho))\cdot A$. While the corresponding evolution equation $\dot\rho(t)= \Cdiv(\chi(\rho)A)$ preserves the energy 
\begin{equation*}
\E:\Z\rightarrow\RR, \ \ \E(\rho):=\int_{\mathbb T^d} U(x)\,d\rho(x),
\end{equation*}
it is not clear if we can define an operator $\pJ$ such that Conjecture~\ref{conj} holds. However in the case $A=J\nabla U$ where $J\in\RR^{d\times d}$ is a constant skew-symmetric matrix and $\chi(a)=a$, we define the operator 
\begin{equation*}
\pJ:\Z\rightarrow (T_\rho^*\Z\rightarrow T_\rho\Z), \ \ \pJ(\rho)(\zeta):=\Cdiv(\rho J\nabla\zeta).
\end{equation*}
Using the antisymmetry of $J$ it follows that 
\begin{equation*}
\langle \zeta^1,\pJ(\rho)\zeta^2\rangle= \int_{\mathbb T^d} \zeta^1 \Cdiv(\rho J\nabla\zeta^2) = - \int_{\mathbb T^d} \nabla\zeta^1\cdot J \nabla\zeta^2\rho = -  \langle \pJ(\rho)\zeta^1,\zeta^2\rangle,
\end{equation*}
i.e.\ $\pJ$ is a skew-symmetric operator. Furthermore $\pJ$ satisfies the Jacobi identity by an elementary but tedious calculation which we skip. Therefore the antisymmetric zero-cost velocity indeed evolves according to the standard Hamiltonian system (see for instance~\cite[Section 3.2]{DuongPeletierZimmer13}) with energy $\E$ and Poisson structure $\pJ$.

\section{Conclusion and discussion}\label{sec:discuss}
In this paper we have presented an abstract macroscopic framework, which, for a given flux-density L-function, provides its decomposition into dissipative and non-dissipative components and a generalised notion of orthogonality between them. This decomposition provides a natural generalisation of the gradient-flow framework to systems with non-dissipative effects. 
Specifically we prove that the symmetric component of the L-function corresponds to a purely dissipative system and conjecture that the antisymmetric component corresponds to a Hamiltonian system, which has been verified in several examples.
We then apply this framework to various examples, both with quadratic and non-quadratic L-functions. 

We now comment on several related issues and open questions. 

\emph{Why does the density-flux description work?} While at the level of the evolution equations which are of continuity-type, the density-flux description does not offer any advantage (recall~\eqref{eq:coupled ev eq}), at the level of the cost functions it allows us to naturally encode divergence-free effects. This is clearly visible for instance in Theorem~\ref{th:FIIRs}, where the evolutions corresponding to $\L_{\Fsy}$, $\L_{\Fasy}$ are dissipative and energy-preserving respectively, while the zero of the full L-function characterises the macroscopic evolution. A simple contraction argument allows us to retrieve the classical gradient-flow structure as well as the FIR inequalities in a fairly general setting, which further reveals the power of this description.

\emph{Antisymmetric force and L-function.} While in the abstract theory the antisymmetric force $\Fasy=\Fasy(\rho)$ is a function of $\rho\in\Dom(\Fasy)$, in all the concrete examples studied in this paper, $\Fasy$ is independent of $\rho$. It is not clear to us if this is a general property of the antisymmetric force or a special characteristic of the examples studied in this paper.  

In Section~\ref{subsec:flows} we conjectured that the zero-velocity flux for the contracted antisymmetric L-function admits a Hamiltonian structure, which was concretely proved for IPFG and zero-range process in Proposition~\ref{prop:IPFG-Ham-state},~\ref{prop:0range-Ham-state} respectively. While this gives insight into the associated zero-flows, it is not clear if $\L_{\Fasy}$ admits  a variational formulation akin to the gradient-flow structure for $\L_{\Fsy}$ discussed in Corollory~\ref{cor:symmetric=EDP}.

\emph{Chemical-reaction networks.} In Appendix~\ref{sec:CB equiv HJ} we provide a new interpretation of systems in complex balance as being exactly those systems which admit the relative entropy as the quasipotential. This also restricts the search for invariant measures of the CME without complex balance to measures that are not exponentially equivalent to the product-Poisson form. However, motivated by the example in that appendix, an interesting question would be to identify the class of systems which admit a rescaled relative entropy as their quasipotential.

Furthermore, the Hamiltonian structure of the zero-velocity for $\L_{\Fasy}$ in the chemical-reaction setting is open. As pointed out in Section~\ref{subsec:reacting particle system}, the non-locality of the jump rates for chemical-reaction networks offers a challenge as opposed to the local jump rates for IPFG and zero-range process.

\emph{Generalised orthogonality.} The notion of generalised orthogonality as introduced in Section~\ref{sec:Gen-ortho} allows us to decompose the L-function as in Theorem~\ref{th:FIIRs} for the special case $\lambda=\frac12$. However a natural question is whether this notion of orthogonality encoded via $\theta_{\rho}$ can be generalised to allow for any $\lambda\in[0,1]$. This would provide a deeper understanding of our main decomposition Theorem~\ref{th:FIIRs} as well as a clear interpretation of the Fisher information in terms of a modified dissipation potential.

\emph{Quasipotentials for multiple invariant measures. } In Remark~\ref{rem:non-unique inv meas} we discussed the possibility of having multiple quasipotentials. On a macroscopic level, forcing uniqueness for non-quadratic Hamilton-Jacobi-Bellman equations is generally challenging. This is not merely a technical issue, since even on a microscopic level there may be multiple invariant measures; we have not pursued this possibility any further.

\emph{Global-in-time decompositions. } In this paper we have focussed on the local-in-time description of the L-function as opposed to working with time-dependent trajectories. While it is not obvious how to generalise the various abstract results to allow for global-in-time descriptions, we expect that it can be worked out case by case for the examples presented in this paper. The main difficulty here is that the time-dependent trajectories are allowed to explore the boundary of the domain where the forces are not well-defined, and therefore an appropriate regularisation procedure is required to extend the domain of definition of these forces.

\iparagraph{Acknowledgements.} The authors are grateful to Alexander Mielke who provided the proof of the Hamiltonian structure for linear antisymmetric flows discussed in Appendix~\ref{app:HamFlow}, and to the participants of the discussions at the AIM workshop ``Limits and control of stochastic reaction networks'' who helped develop the results of Appendix~\ref{sec:CB equiv HJ}, in particular Daniele Cappelletti, Anne J. Shiu and Artur Stephan. Furthermore, the authors thank Jin Feng, Davide Gabrielli, Alberto Montefusco, Mark Peletier, Jim Portegies, Richard Kraaij and P{\'e}ter Koltai  for insightful discussions. US thanks Julien Reygner who first pointed out the possibility of a connection between FIR inequalities and MFT. This work was presented at the MFO Workshop 2038 `Variational Methods for Evolution' and the authors would like to thank the participants for stimulating interactions.

The work of RP and MR has been funded by the Deutsche Forschungsgemeinschaft (DFG) through grant CRC 1114 ``Scaling Cascades in Complex Systems", Project C08.
RP received further support from the Math+ excellence cluster through project EF4-10.
The work of US was supported by the Alexander von Humboldt foundation and the DFG under Germany's Excellence Strategy--MATH+: The Berlin Mathematics Research Center (EXC-2046/1)-project ID:390685689 (subproject EF4-4).

\iparagraph{Data availability statement.} Data sharing not applicable to this article as no datasets were generated or analysed during the current study.

\begin{appendices}

\section{Hamiltonian structure for linear antisymmetric flow}\label{app:HamFlow}

The results in this appendix are due to Alexander Mielke.

We study the linear ODEs of the form
\begin{equation}\label{def:linODE}
\dot \omega = \frac12 A\omega \in \mathbb R^d \quad \text{with} \quad A\tp \omega_* = A\omega_*=0 \text{ for some } \omega_*\neq 0.
\end{equation}
In Theorem~\ref{thm:lin-Ham} we provide a complete characterisation of a natural Hamiltonian structure for these ODEs. In contrast to the typical settings of Hamiltonian systems, where $A\in\mathbb R^{d\times d}$ is assumed to be skew-symmetric, here we assume the existence of an invariant vector $\omega_*$ for the dynamics. The zero-cost antisymmetric flux for the IPFG system discussed in Section~\ref{sec:IPFG antisym flow} is of the form~\eqref{def:linODE}.

The following lemma provides a useful alternative characterisation of the Jacobi identity for Poisson structures which will be used to prove Theorem~\ref{thm:lin-Ham} below. 

\begin{lem}\label{lem:JacId-Alternate} Define a (Poisson-like) bracket $\{\cdot,\cdot\}$ by
\begin{equation}\label{def:Pbrack}
\{\pG_1,\pG_2\}(\omega):= (\nabla \pG_1)(\omega)\cdot \widetilde\pJ(\omega) (\nabla\pG_2)(\omega)
\qquad \forall \omega \in \mathbb{R}^d,
\end{equation}
where $\pG_1,\pG_2:\mathbb R^d\rightarrow\mathbb R$ are $C^2$-mappings, and the $C^1$ matrix-valued function $\omega\mapsto \widetilde\pJ(\omega)\in\mathbb R^{d\times d}$ is antisymmetric, i.e.\ $\widetilde\pJ\tp=-\widetilde\pJ$.  The bracket~\eqref{def:Pbrack} satisfies the Jacobi identity if and only if for all smooth $\pG_1,\pG_2,\pG_3:\mathbb R^d\rightarrow\mathbb R$ and for all $\omega\in\RR^d$ we have 
\begin{equation}\label{eq:mod-JacId}
 (\nabla \pG_1)\cdot \nabla\widetilde\pJ[\widetilde\pJ(\nabla\pG_2)](\nabla\pG_3) + (\nabla\pG_2) \cdot \nabla\widetilde\pJ[\widetilde\pJ(\nabla\pG_3)](\nabla\pG_1) + (\nabla\pG_3)\cdot \nabla\widetilde\pJ[\widetilde\pJ(\nabla\pG_1)](\nabla\pG_2) = 0,
\end{equation}
where $\nabla \widetilde\pJ[v]$ is the matrix valued function with $(\nabla \widetilde\pJ[v])_{ij}(\omega) = \sum_k v_k (\partial_{\omega_k}  \widetilde\pJ(\omega)_{ij})$.
\end{lem}

The proof follows by straightforward manipulation of the Jacobi identity.
We now present the Hamiltonian structure for~\eqref{def:linODE}.
 
\begin{thm}\label{thm:lin-Ham}
The linear ODE~\eqref{def:linODE} admits the Hamiltonian system $(\mathbb R^d,\tilde\E,\widetilde\pJ)$ with the linear energy and the linear Poisson structure 
\begin{equation*}
\tilde\E(\omega)=c- \omega_*\cdot\omega , \quad \widetilde\pJ(\omega)=\frac{1}{2|\omega_*|^2}\Bigl( \omega_*\otimes \left(A\omega\right) - \left(A\omega\right)\otimes \omega_*  \Bigr),
\end{equation*}
for any $c\in\RR$. Consequently $\dot \omega = \widetilde\pJ(\omega)\nabla\tilde\E(\omega)$.  
\end{thm}
\begin{proof}
For any $b\in\mathbb R^d$ we have 
\begin{equation*}
\widetilde\pJ(\omega)b=\frac{1}{2|\omega_*|^2}\left( (A\omega\cdot b ) \omega_* - (\omega_* \cdot b) A\omega\right)=\frac{1}{2|\omega_*|^2}\left( (\omega\cdot A\tp b) \omega_*- ( \omega_*\cdot b) A\omega   \right)
\end{equation*}
and  inserting $b=\omega_*$ in this relation and using $A\tp \omega_*=0$ it follows that $\frac12A\omega=-\widetilde\pJ(\omega)\omega_*=\widetilde\pJ(\omega)\nabla\tilde\E(\omega)$. Since $\widetilde\pJ(\omega)\tp = -\widetilde\pJ(\omega)$ by construction, we only need to prove the Jacobi identity~\eqref{eq:mod-JacId} to prove this result.
Using the linearity of $\widetilde\pJ$ we find $\nabla\widetilde\pJ[v](\omega)=\widetilde\pJ(v)$, and therefore for any $\pF\in\mathbb R^d$ we have 
\begin{equation*}
\nabla\widetilde\pJ[\widetilde\pJ(\omega)\pF](\omega) = \widetilde\pJ(\widetilde\pJ(\omega)\pF) = \frac{1}{2|\omega_*|^2}\Bigl( \omega_*\otimes \left(A\widetilde\pJ(\omega)\pF\right) - \left(A\widetilde\pJ(\omega)\pF\right)\otimes \omega_*\Bigr).
\end{equation*}
Using  $A\omega_*=0$ we find
\begin{equation*}
A\widetilde\pJ(\omega)\pF = -\frac{\pF \cdot\omega_*}{2|\omega_*|^2} A^2 \omega.
\end{equation*}
Using the above two relations we arrive at 
\begin{equation*}
  \pF_1\cdot \nabla\widetilde\pJ[\widetilde\pJ(\omega)\pF_2](\omega)\pF_3 = \frac{1}{4|\omega_*|^4}\Bigl( (\pF_1\cdot A^2 \omega)( \pF_2 \cdot \omega_*)( \omega_*\cdot\pF_3) - (\pF_1 \cdot\omega_*) ( \pF_2\cdot\omega_*)( A^2\omega\cdot\pF_3) 
\Bigr).
\end{equation*}
Setting $\pF_i = \nabla \pG_i$ we can compute the remaining two terms on the left hand side of~\eqref{eq:mod-JacId} by rotating the indices.
\end{proof}



\section[Complex balance and quasipotential condition]{Complex balance and quasipotential condition~\eqref{eq:invariant I_0}
}
\label{sec:CB equiv HJ}

The results in this appendix were developed during discussions with the participants of the online AIM workshop ``Limits and control of stochastic reaction networks''. This appendix explores the relation between the quasipotential for chemical reaction networks on the one hand (defined via~\eqref{eq:invariant I_0}) and the notion of complex balance on the other. Interestingly, it turns out that complex balance is equivalent to having the relative entropy as the quasipotential.

Consider chemical-reaction networks satisfying mass-action kinetics~\eqref{ass:mass-action} as explained in Section~\ref{subsec:reacting particle system}. To stress that these results are independent of the flux formulation and do not require a decomposition into forward and backward reactions, we will simply work with the density formulation~\eqref{def:FIIR-FIR-contract}, with the dual of the contracted L-function given by
\begin{equation*}
  \hat\H(\rho,\xi)
  :=
  \sup_{u \in T_\rho\Z} \xi\cdot j - \hat\L(\rho,u)
  =
  \sum_{r\in\Reac} c_r \rho^{\alpha\super{r}} \big( e^{\gamma\super{r}\cdot\xi}-1\big).
\end{equation*}
The equation~\eqref{eq:invariant I_0} for the quasipotential reads
\begin{equation}\label{eq:hat HJ}
  \hat\H\big(\rho,\grad\QP(\rho)\big)=0 \qquad \text{for all coordinate-wise positive } \rho \in \RR^\X_{>0}.
\end{equation}

The following result shows that the notion of complex balance is inherently connected to~\eqref{eq:hat HJ}, in the case when the quasipotential is the relative entropy.
This result has also appears in \cite[Lemma 3.6]{GaoLiu2022}.
\begin{theorem}\label{th:CB equiv HJ}
The following two statements are equivalent.
\begin{enumerate}
\item Complex balance~\eqref{def:complex-bal} holds with respect to $\pi\in\Z$.
\item Equation~\eqref{eq:hat HJ} holds with $\V(\rho)=\sum_{x\in\X} s(\rho_x\mid\pi_x)$.
\end{enumerate}
\end{theorem}

\begin{proof}
We first prove the forward implication. Let $\V(\rho)=\sum_{x\in\X} s(\rho_x\mid\pi_x)$ and calculate for $\rho\in \RR^\X_{>0}$
\begin{equation*}
  \hat\H\big(\rho,\grad\V(\rho)\big)
  =
  \sum_{r\in\Reac} c_r \rho^{\alpha\super{r}} \big( (\tfrac{\rho}{\pi})^{\alpha\super{\bw(r)}-\alpha\super{r}}-1\big)
  =
  \sum_{r\in\Reac} c_r \pi^{\alpha\super{r}} \big( (\tfrac{\rho}{\pi})^{\alpha\super{\bw(r)}} - (\tfrac{\rho}{\pi})^{\alpha\super{r}} \big)=0,
\end{equation*}
where the final equality follows by choosing $\psi_\alpha=(\tfrac{\rho}{\pi})^\alpha$ in the complex-balance condition~\eqref{def:complex-bal}. 

Now we present the backward implication (proof due to Artur Stephan). Assume that~\eqref{eq:hat HJ} holds with $\V(x)=\sum_{x\in\X} s(\rho_x\mid\pi_x)$. Sorting the expression with respect to the complexes, one obtains
\begin{equation*}
  \hat\H\big(\rho,\grad\V(\rho)\big)=p\big(\tfrac{\rho}{\pi}\big),
 \ \text{ where }    
  p(a):=\sum_{\alpha\in\CC} A_\alpha a^\alpha
\text{ and }
  A_\alpha := \sum_{r\in\Reac: \alpha\super{r}=\alpha} c_r \pi^{\alpha\super{r}} - \sum_{r\in\Reac: \alpha\super{\bw(r)}=\alpha} c_r \pi^{\alpha\super{r}}.
\end{equation*}
Since all $\alpha\in\CC$ are distinct, the polynomial $p(\tfrac{\rho}{\pi})$ can only be $0$ for all $\rho/\pi\in\RR^\X_{>0}$ when all coefficients are zero. Hence $A_\alpha=0$ for all $\alpha\in\CC$, which is equivalent to complex balance~\cite[Eq.~(8)]{AndersonCraciunKurtz2010}.
\end{proof}

The forward implication in Theorem~\ref{th:CB equiv HJ} can also be shown indirectly via the Chemical Master Equation (CME), describing the probability of the microscopic random particle system with jump rates
\begin{equation}
  \kappa\super{V}_r(\rho):= \frac{c_r}{V^{\lvert\alpha\super{r}\rvert_1}} \frac{(V\rho)!}{(V\rho-\alpha\super{r})!}\mathds1_{\{V\rho\geq \alpha\}}.
\label{eq:CME rates}
\end{equation} 
If complex balance holds with respect to $\pi\in\Z$, then the CME is known to have the invariant measure of product-Poisson form~\cite[Thm.~4.1]{AndersonCraciunKurtz2010}, i.e.\ 
\begin{equation}\label{eq:product Poisson}
  \Pi\super{V}(\rho)=\prod_{x\in\X} \frac{(V\pi_x)^{V\rho_x}}{(V\rho_x)!} e^{-V\pi_x}.
\end{equation}
\sloppy{The large-deviation rate can be explicitly
calculated using Stirling's formula, which gives $\lim_{V\to\infty} -\frac1V\log\Pi\super{V}(\rho)=\sum_{x\in\X} s(\rho_x\mid\pi_x)$,
and so by Theorem~\ref{th:LDP QP} the equation~\eqref{eq:hat HJ} must hold.}

Very little can be said about the invariant measures for the CME without the assumption of complex balance. However, the following is a straightforward consequence of Theorem~\ref{th:CB equiv HJ}.
\begin{corollary} If complex balance does not hold, any invariant measure $\Pi\super{V}$ of the CME will not be exponentially equivalent to the product-Poisson form~\eqref{eq:product Poisson}, i.e.\ $\Pi\super{V}$ and \eqref{eq:product Poisson} will not have the same large-deviation rate.
\end{corollary}

The following simple example, constructed by Daniele Cappelletti and Anne J. Shiu, shows that an appropriately rescaled relative entropy can still be a quasipotential. Consider a simple birth-death process
\begin{equation*}
  \emptyset\xrightarrow{\kappa_b} \mathsf{A} \xleftarrow{\kappa_d} 2\mathsf{A},
\end{equation*}
for which the CME has the explicit invariant measure (for simplicity writing $\rho:=\rho_\mathsf{A}$)
\begin{align*}
  \Pi\super{V}(\rho)=\frac1{Z_V} \prod_{i=1}^{V\rho-1} \frac{V\kappa_b}{V^{-1}\kappa_d(i+1)i} = \frac1{Z_V}\frac1{(V\rho-1)!}\frac1{(V\rho)!} (V^2\tfrac{\kappa_b}{\kappa_d})^{Vx-1},
\end{align*}
where $Z_V$ is the V-dependent normalisation constant. Note that the corresponding (zero-cost) reaction rate equation reads $\dot \rho(t)=\kappa_b -\kappa_d \rho(t)^2$, which clearly has the steady state $\pi:=\sqrt{\kappa_b/\kappa_d}$. The CME is in detailed balance with respect to $\Pi\super{V}$, but the reaction network is \emph{not} in complex balance w.r.t. $\pi$. Again by Stirling's formula and using the fact that $\inf\QP=0$, we find
\begin{equation*}
  \V(x):=\lim_{V\to\infty}-\frac1V\log\Pi\super{V}(x)= 2 s(\rho\mid \sqrt{\kappa_b/\kappa_d}).
\end{equation*}
Although the CME is in detailed balance, this result does not contradict the findings of \cite{MPPR2017}, since this reaction network is not reversible in the sense of footnote~\ref{ftn:reversible network}.

\end{appendices}

\bibliographystyle{alphainitials}
{\small
\bibliography{NonMP}
}

%

\end{document}